\documentclass[manuscript, screen]{acmart}

\usepackage{comment}

\AtBeginDocument{%
  }

\usepackage[font=small,labelfont=bf]{caption}
\usepackage{subcaption}
\usepackage[clock]{ifsym}

\usepackage{comment}
\usepackage{graphicx}
\usepackage[ruled]{algorithm2e} 
\usepackage{float}

\usepackage{amsmath,bm,esint}
\usepackage{blkarray, bigstrut}
\usepackage{xparse}
\usepackage{rotating}
\usepackage{multirow}
\usepackage{tabularx}
\usepackage{tablefootnote}

\usepackage{booktabs}
\newcommand{\ra}[1]{\renewcommand{\arraystretch}{#1}}

\usepackage[para,online]{threeparttable}
\usepackage{pifont}

\usepackage{array}
\usepackage{colortbl}
\usepackage{wrapfig}

\usepackage{lstautogobble}

\definecolor{light-gray}{gray}{0.95}
\definecolor{tbl-row-color}{gray}{0.95}
\definecolor{pastelred}{rgb}{1.0, 0.41, 0.38}
\definecolor{radicalred}{rgb}{1.0, 0.21, 0.37}
\usepackage{listings}

\definecolor{keywordsColor}{RGB}{134, 14, 11}
\definecolor{commentsColor}{RGB}{54, 54, 54}

\lstdefinelanguage{OurLanguage}{
    alsoletter={:,=, <, >, &, |},
    keywords={while, end, types, if, else, else:, and, or, not, Normal, Uniform, sin, cos},
    backgroundcolor=\color{light-gray},
    morekeywords={=, <, >, <=, >=, ==, !=, &&, ||},
    basicstyle={\ttfamily\small\normalfont},
    keywordstyle={\color{keywordsColor}\ttfamily\bfseries},
    comment=[l]{\#},
    commentstyle={\color{commentsColor}\ttfamily},
    autogobble=true,
    mathescape=true
}
\lstdefinestyle{program}{basicstyle=\small\ttfamily,keywordstyle=\bfseries}

\usepackage{tikz}
\usetikzlibrary{shapes.geometric, arrows}
\tikzstyle{startstop} = [rectangle, rounded corners, minimum width=3cm, minimum height=1cm,text centered, draw=black, fill=red!30]

\tikzstyle{io} = [rectangle, minimum width=4.5cm, minimum height=1.8cm, text centered,  text width=4cm, draw=black, fill=white!30]

\tikzstyle{process} = [rectangle, minimum width=3cm, text width=8cm, minimum height=1cm, text centered, draw=black, fill=orange!30]

\tikzstyle{decision} = [diamond, minimum width=3cm, minimum height=1cm, text centered, draw=black, fill=green!30]

\tikzstyle{arrow} = [thick,->,>=stealth, text width = 110]

\newcommand{\E}{\mathbb{E}}
\newcommand{\var}{\mathbb{V}\mathrm{ar}} 

\newcommand{\real}{{\mathbb R}}
\newcommand{\nat}{{\mathbb N}}

\newcommand{\gc}{\mathrm{GC}}
\newcommand{\mm}{\mathrm{MM}}
\newcommand{\ks}{\mathrm{KS}}

\newcommand{\A}{{\mathbf A}}
\newcommand{\X}{{\mathbf X}}

\newcommand{\Z}{{\mathbf Z}}
\newcommand{\z}{{\mathbf z}}

\newcommand{\M}{{\mathbf M}}

\newcommand{\x}{{\mathbf x}}

\newcommand{\f}{{\mathbf f}}

\newcommand{\mbf}{\mathbf{m}}
\newcommand{\hbf}{{\mathbf h}}

\newcommand{\Qcal}{\mathcal{Q}}

\newcommand{\Lambdabf}{{\bm \Lambda}}

\newtheorem{definition}{Definition}
\newtheorem{Assumption}{Assumption}
\newtheorem{theorem}{Theorem}

\newtheorem{proposition}{Proposition}

\definecolor{BrickRed}{HTML}{B6321C}
\definecolor{BlueViolet}{HTML}{473992}
\definecolor{Maroon}{HTML}{AF3235}
\definecolor{ForestGreen}{HTML}{009B55}

\begin{document}

\title{Moment-based Density Elicitation with Applications in Probabilistic Loops}
\author{Andrey Kofnov}
\email{andrey.kofnov@tuwien.ac.at}
\orcid{0000-0002-1734-2918}
\affiliation{%
  \institution{Institute of Statistics and Mathematical Methods in Economics, Faculty of Mathematics and Geoinformation, TU Wien}
  \streetaddress{Wiedner Hauptstrasse 8-10}
  \city{Vienna}
  \state{Vienna}
  \country{Austria}
  \postcode{1040}
}

\author{Ezio Bartocci}
\affiliation{%
 \institution{Faculty of Informatics, TU Wien}
 \streetaddress{}
 \city{Vienna}
  \postcode{1040}
  \country{Austria}}
\orcid{0000-0002-8004-6601}

\author{Efstathia Bura}
\affiliation{%
  \institution{Institute of Statistics and Mathematical Methods in Economics, Faculty of Mathematics and Geoinformation, TU Wien}
  \streetaddress{Wiedner Hauptstrasse 8-10}
  \city{Vienna}
  \state{Vienna}
  \postcode{1040}
  \country{Austria}}
\email{efstathia.bura@tuwien.ac.at}
\orcid{0000-0003-4972-5320}

\renewcommand{\shortauthors}{Kofnov et al.}

\begin{abstract} 
We propose the K-series estimation approach for the recovery of  unknown univariate and multivariate distributions given knowledge of a finite number of their moments. Our method is directly applicable to the probabilistic analysis  of systems that can be represented as probabilistic loops; i.e., algorithms that express and implement non-deterministic processes ranging from robotics to macroeconomics and biology to software and cyber-physical systems.      
K-series statically approximates the joint and marginal distributions of a vector of continuous random variables updated in a probabilistic non-nested loop with nonlinear assignments given a finite number of moments of the unknown density. Moreover, K-series automatically derives the distribution of the systems'  random variables symbolically as a function of the loop iteration. 
K-series density estimates are accurate, easy and fast to compute. We demonstrate the feasibility and performance of our approach on multiple benchmark examples from the literature. 
\end{abstract}



\maketitle

\section{Introduction}

There are several methods in statistics to infer the distribution of a random sample. If data are sampled from the unknown distribution, the default ``go-to'' approach is nonparametric density estimation (e.g., histogram, $k$-NN, kernel density estimation, etc. (see, for example, \cite{Silverman1998})) that involves local smoothing and lets ``the data speak for themselves.'' In absence of any information about the data generating process, nonparametric estimation is the only available tool. 

When features of the unknown distribution are available, such as moments, nonparametric density estimates can be significantly improved upon. Yet, knowledge of moments of an unknown distribution is typically rare. One such setting is probabilistic programming analysis, where the moments of unknown distributions are computable, either exactly or approximately via sampling. 

In this paper,  we propose \textit{K-series} to estimate the probability density function (pdf) of a marginal or joint distribution based on knowledge of a finite number of its moments.
Our approach was motivated by recent developments in probabilistic programming analysis. 
This is exemplified by the stochastic dynamical system  in  Fig.~\ref{fig:DDMR}, where the moments of the random location variables $X$ and $Y$  are computable at each iteration using the approach in~\cite{Kofnovetal2024}.
The program in the top left panel of Fig.~\ref{fig:DDMR} encodes the \textit{stochastic} dynamics of the position of a mobile robot, referred to as the \textit{Differential-Drive Mobile Robot} in ~\cite{Jasouretal2021}, in the presence of external disturbances. 
The position of the robot on the 2D plane is reflected in the $(X,Y)$ coordinates, and its orientation in $\theta$. The speed of the left and right wheels are constant and already incorporated into the equations. External disturbances are modeled as  $\Omega_{l} \sim $ Uniform$(-0.1,0.1)$, and $\Omega_{r} \sim $ Beta$(1, 3)$. The initialization of the location variables $X$ and $Y$ is also random (Uniform$(-0.1,0.1)$) and the angle $\theta$ is initialized as Gaussian with mean 0 and variance 0.1. 

The \textit{Differential-Drive Mobile Robot} program is a \textit{probabilistic loop}, which is a special case of a \textit{Probabilistic Program} (PP). In simple terms, a probabilistic loop is a program loop that contains random assignments such as draws from random distributions (normal, Bernoulli, uniform, etc.). A formal definition of a probabilistic loop is given in \cite{Moosbruggeretal2022} (see Fig. 3 in \cite{Moosbruggeretal2022} for the detailed syntax of probabilistic loops).   
A main challenge in probabilistic programming analysis is to automatically derive the \textit{probability density function} (pdf) of the program random variables 
\cite{Gehretal2016}, which becomes even more challenging in the presence of loops with potentially infinite execution.

This problem has been partially addressed by computing the moments of the unknown densities in specific classes of probabilistic programs, such as the so called \textit{prob-solvable} loops~\cite{BartocciKS19}. A prob-solvable loop consists of a set of initialization statements followed by a non-nested loop body where the variables are updated via polynomial assignments and/or by drawing samples from statistical distributions determined by their moments. In this class of probabilistic loops, moments of any order of the program random variables are computed automatically as closed-form expressions in the number of iterations using symbolic summation and polynomial algebra~\cite{BartocciKS19}.
For non-polynomial assignments,~\cite{Kofnovetal2022} renders probabilistic loops with general continuous functional assignments compatible with the automatic tool of \cite{BartocciKS19} for exact moment computation and provides approximations of exact moments of the target pdf. For trigonometric functional updates, commonly encountered in stochastic dynamical systems (see \cite{Jasouretal2021}) as in the \textit{Differential-Drive Mobile Robot} in Fig.~\ref{fig:DDMR},  \cite{Kofnovetal2024} computes exact moments of any order across iterations. 

 \begin{figure}
 \centering
 \includegraphics[width=0.95\textwidth, height=0.50\textwidth]
 {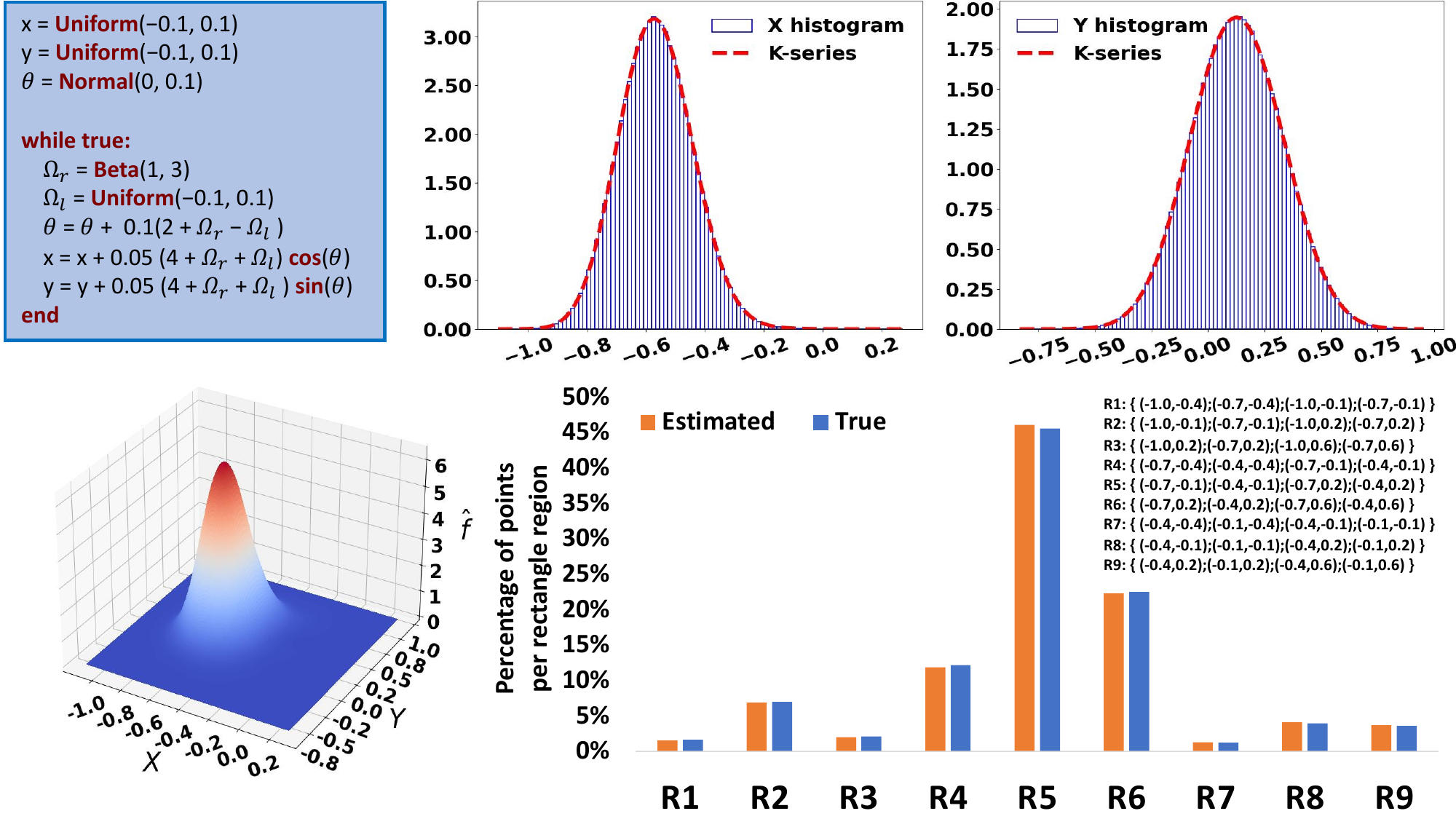}
  \caption{Probabilistic loop with non-polynomial assignment for the Differential-Drive Mobile Robot~\cite{Jasouretal2021} (top left), the approximations of the marginal distributions with K-series (top right), the approximation of the joint distribution with K-series (bottom left)
and comparison with true histogram (bottom right).} 
  \label{fig:DDMR}
\end{figure}

\medskip

\textbf{Contributions.} 
Our proposal, \textit{K-series estimation}, is a density estimation method that recovers the probability density function (pdf) with bounded support from a finite set of the  moments of a random vector $\X=(X_1,\ldots,X_k)^T$, $m_{i_1,i_2,\ldots,i_k}=\E\left(X_1^{i_1}X_2^{i_2} \cdot \ldots \cdot X_{k}^{i_k}\right)$, $0 \leq i_j \leq d_{j}$, $d_j \in \mathbb{N}$, $j=1,\ldots,k$,  using a basis of orthogonal polynomials that target the unknown density via the choice of the reference pdf. Our approach is general in that it allows for any reference distribution, whose effect is incorporated in the construction of the orthogonal polynomials via the Gram-Schmidt orthogonalization procedure and can be tailored to improve the accuracy of the estimation. 
A summary of our contributions in this paper follows.

\begin{itemize}
\item[a)] We adapted the mathematical framework for estimating the distribution of a random variable, proposed by \cite{Filipovicetal2013}, to the setting of probabilistic loops where multiple random variables are generated at each iteration;
    \item[b)] Based on this framework, we introduce K-series: the first method to automatically derive the distributions of multiple state variables in probabilistic programs, such as \emph{prob-solvable} loops~\cite{Bartoccietal2019,Moosbruggeretal2022}, for any number of iterations and symbolically;
    \item[c)] We derive the theoretical foundation in Proposition \ref{prop:GC} and Theorem \ref{thm:MM}, where we prove that other methods, such as the Gram-Charlier (GC) expansion~\cite{Cramer1957,KendallStuart1977,Hald2000} and Method of Moments ~\cite{Munkhammar_etal_2017}, are special cases of our general approach;
    \item[d)] We obtain the convergence rate of our density estimator to the true pdf 
    in Theorem~\ref{thm:L1_convergence};
    \item[e)] We show that K-series is an accurate estimator of the unknown true pdf by proving the moment matching principle of K-series in Theorem~\ref{thm:moment_match}; that is, we show that  the first $n$ moments of the true pdf and the K-series estimator are the same; 
    \item[f)] We derive the approximation to the true support of the target pdf.
\end{itemize}

K-series is an estimation method of the distribution of a multivariate random variable with the following \textit{theoretical guarantees}. 
\begin{enumerate}
    \item  The K-series estimate converges to the true distribution in  $L_{1}$.
    \item The first $n$ moments on which the estimate is based equal  the first $n$ moments of the true distribution (which is essential under the assumption that the random variable is uniquely identifiable by its moments\footnote{Refer to Section~\ref{sec::univar_k_series} for more details.}).
    \item The interval estimate of the support of the random variables is minimal.
\end{enumerate}

Important features of our method are (a) its ease of computation and application, (b) its speed and (c) its ability to recover multivariate distributions.  K-series is a natural complement to automated tools for exact moment computation in probabilistic loops, such as 
\texttt{Polar}~\cite{Moosbruggeretal2022}, which can also accommodate \emph{if-then-else} conditions under certain restrictions. 

Specifically, K-series can be used to derive the distribution of probabilistic programs with a non-nested loop and acyclic state variable dependencies (e.g., an assignment for a variable $x_i$ must not reference variables $x_j$ with $j > i$). These PPs correspond to what is called a directed acyclic graph that is equivalent to a Bayesian network.  The original \textit{prob-solvable loops} introduced by \cite{Bartoccietal2019} were the first such hierarchical probability structures for which it is possible to automatically compute moment-based invariants of any order over the program state variables as closed-form expressions in the loop iteration. In particular, \cite{Stankovicetal2022} represented Bayesian networks as while loops in probabilistic programs with polynomial assignments over random variables; i.e., prob-solvable loops.   That is, our approach applies to probabilistic  loops that are algorithmic representations of non-self-referential conditional distributional structures.

\medskip

\textbf{Related Literature.} 
Gram-Charlier (GC) expansion~\cite{Cramer1957,KendallStuart1977,Hald2000} is the standard statistical technique to estimate a continuous pdf given a set of its moments. The Gram-Charlier estimate is a series expansion of a density in terms of the normal density and its derivatives. Even though it can recover the \textit{normal} perfectly, it can be fairly inaccurate when the target pdf differs noticeably from it. 

In the context of probabilistic loops, the problem of estimating statically the probability distribution of random variables from their moments has been recently considered  
in~\cite{Karimietal2022} and \cite{Moosbruggeretal2022}.
\cite{Karimietal2022}  estimate univariate distributions with  Maximum Entropy 
(ME)~\cite{Biswas2010,Lebazetal2016} and GC expansion in 
\textit{prob-solvable loops}~\cite{BartocciKS19} with  polynomial assignments. ME maximizes the Shannon information entropy subject to a finite set of moments provided as input. It cannot be expressed  symbolically in terms of the moments in the number of the loop iteration. GC expansion estimates  the unknown pdf in a symbolic expression in the number of the loop iteration in terms of its cumulants that can be computed from its moments. GC's  inaccuracy as an estimator of non bell-shaped distributions and lack of convergence (see \cite{Cramer1957} or \cite{Kolossa2006}) are its main limitations.

\cite{Munkhammar_etal_2017} used the estimation method of moments\footnote{This method estimates parameters of a target  distribution by equating sample moments with the corresponding moments of the distribution.} to develop an algorithm for an "$n$-order polynomial approximation of a pdf" based on its consecutive $n$ first moments. 
We show that the method of moments, as well as the GC expansion, 
are special cases of K-series in Section \ref{sec:general}. 

The method in \cite{Tekel_Cohen}, though similar,
is based on the moments of both the unknown target and the reference pdf.  K-series, on the other hand,  uses only the moments of the unknown pdf in the construction of its estimate, removing the need for an additional tuning parameter; that is, the number of moments of the reference. Moreover, \cite{Tekel_Cohen} developed no theory on the statistical properties of the proposed estimator nor did they provide any  connection with the Gram-Charlier or other series estimators. 
Our approach is general in that it allows for any reference distribution, whose effect is incorporated in the construction of the orthogonal polynomials via the Gram-Schmidt orthogonalization procedure and can be tailored to improve the accuracy of the estimation. 

\cite{Filipovicetal2013} uses a similar estimation procedure as K-series, 
even though the parameters  of the reference pdf in \cite{Filipovicetal2013} are  computed from the moments of the target unknown pdf in a circular manner. 
Importantly, similarly to \cite{Tekel_Cohen}, \cite{Filipovicetal2013} does not study the statistical properties of the estimator but rather focuses on how to select the reference pdf. In particular, \cite[Th. 2.1]{Filipovicetal2013} is a straightforward  consequence of the conditions imposed on the reference and the unknown pdf.
We prove the convergence of the K-series estimator to the true pdf in $L_{1}$ for a wide class of reference pdfs, and in $L_2$ for the uniform reference pdf.  
Our choice of reference reflects our lack of knowledge of the true target pdf. It is, nevertheless, flexible and can pivot the estimation closer to the truth in the presence of additional information. \cite{Filipovicetal2013} formulated the moment matching principle as a guiding principle for choosing the reference without connecting it to the actual estimator. We go one step further to prove the moment matching principle of the K-series estimator itself, establishing the equality of its first moments to the corresponding moments of the true pdf.  Other new theoretical contributions are Theorem~\ref{thm:MM}, where we prove that the Method of Moments ~\cite{Munkhammar_etal_2017} is a special case of our general approach, and our approximation of the support of the target unknown  pdf in Section~\ref{sec:suppapprox}. 



$\lambda$PSI  is a solver for computing exact distributions of a PP as symbolic mathematical expressions "with first-class functions, nested inference and discrete, continuous and mixed random variables" ~\cite{Gehretal2020}. However, not only is this solver limited to bounded loops but it also returns very complex symbolic mathematical expressions that are hard to compute and implement even for very few loop iterations. 

More recently, \cite{Klinkenbergetal2023} developed an approach (Prodigy) that can carry out exact inference in PPs that are loop-free with discrete random states. In order to extend to potentially  infinite while-loops, they try to identify classes of while-loop programs that are equivalent to loop-free programs. This work applies  to probabilistic programs involving discrete random states whose distribution depends on parameters that are updated in a Bayesian framework.

Although, $\lambda$PSI \cite{Gehretal2020} or Prodigy \cite{Klinkenbergetal2023} provide hard guarantees, they do so either for a very restricted class of problems (Prodigy works only with discrete random variables and only with loop-free programs or their equivalents), or for a very small number of iterations (for $\lambda$PSI, $gauss(0, 1) + uniform(0, 1)^2$ is already a challenge at the first iteration and it is practically infeasible after three or four iterations). Neither has broad practical applicability. We argue that our tool is a usable extension of these two tools.



\paragraph{Paper organization}  In Section~\ref{kseries} we introduce univariate and multivariate K-series estimators, derive an interval estimate of the support of the pdf and show that the Gram-Charlier expansion 
\cite{KendallStuart1977} and method of moments \cite{Munkhammar_etal_2017} are special cases. 
In Section~\ref{symbolic}, we derive the K-series functional formula as a symbolic expression in the number of the loop iteration.   
Section~\ref{experiments} provides the experimental evaluation of our approach. We conclude in  Section~\ref{conclusion}.

\section{K-series}
\label{kseries}

We develop the  \textit{K-series} estimation method to recover the joint and marginal distributions of a vector of random variables given a finite number of their moments. Our proposal generalizes GC series to estimate an unknown pdf with bounded support. 
Both K-series and GC require a known reference distribution in order to derive the unknown continuous pdf. The normal reference pdf  is instrumental in GC series as it dictates the choice of Hermite polynomials. Our approach allows using \textit{any} continuous pdf provided its support covers the support of the target pdf we want to estimate. 
We present the univariate and its multivariate extension in Sections \ref{sec::univar_k_series} and \ref{ssec:multivar_K_series}, respectively.

\subsection{Univariate K-series}\label{sec::univar_k_series}

Let $X$ be a continuous random variable, supported on an arbitrary interval $\Omega \subseteq \mathbb{R}$, with cumulative distribution function (cdf) $F_{X}(x)$ that is continuously differentiable on $\Omega$ and the corresponding pdf $f(x) = dF_{X}(x) / dx$ is non-negative upper bounded everywhere on $\Omega$ with countable  zeros.
Let $M = \{m_{1}, m_{2},\ldots,m_{n},\ldots \}$ be the set of all moments of the random variable $X$ and suppose only the first $n$  are known. We denote this finite subset of $M$ by $M_n = \{m_{1}, \ldots, m_{n}\}, n \in \mathbb{N}$ and the vector with elements the moments in $M_n$ by $\mbf_n=(1,m_1,\ldots, m_n)^T$. Boldface symbols denote vectors and matrices throughout the paper.

\begin{definition}\label{expon_int}
    A probability density function is said to be exponentially integrable, if there exists a positive $a>0$ such that  $\int\limits_{\real}\exp\{a|x|\}f(x)dx < \infty$ (see \cite{Ernstetal2012,Rahman2018}).
\end{definition}

Moments can serve as a means to characterize probability distributions.  A pdf supported on an unbounded set is uniquely identifiable by its moments if and only if it is exponentially integrable \cite{Ernstetal2012}. \footnote{See also \cite[Th. 30.1, p. 388]{Billingsley2012}  and \cite[Th. 2.3.11]{CasellaBerger2001}.} This encompasses a very broad class, including most widely used densities. However, notable  counterexamples are the log-normal and Cauchy distributions.

When a distribution can be uniquely identified by its moments, then it  is completely  determined by these moment values, which  enables a succinct and accurate representation of the distribution. Moreover, distribution identification by moments facilitates meaningful comparisons between diverse distributions and streamlines statistical inference procedures.

Let $\phi(x)$ be an arbitrary continuous pdf  that is positive everywhere on its support  $\Theta$, where $\Omega \subseteq \Theta$. We require either $\Theta$ be unbounded and $\phi(x)$ uniquely identifiable by its moments, or $\Theta$ be finite (bounded). Let $H = \{h_{0}(x), h_{1}(x), \ldots, h_{n}(x)\}$, $h_{0}(x) \equiv 1$ be a sequence of orthonormal polynomials on $\Theta$ with respect to $\phi(x)$; i.e., 

\begin{align}
\left<h_{i},h_{j}\right>_{\phi}&= \int\limits_{\Theta} h_{i}(x)h_{j}(x)\phi(x) dx 
=\begin{cases} 1 & i=j \\
0 & i \ne j \end{cases}. \label{orth.poly}
\end{align}
    A function $l(x)$ is said to belong to $L_{p}(\Sigma, \rho)$ if  $\int_{\Sigma}|l(x)|^{p}\rho(x)dx <\infty$
    (see \cite{rudin1986}).  
Throughout the paper, $f$ is used to denote the target and $\phi$ the reference pdf, respectively. Also, at least one of the following two assumptions is assumed to hold. 

\begin{Assumption}\label{Cond:1}
The support $\Omega$ of the pdf of $X$ 
 is a bounded set.
\end{Assumption}
\begin{Assumption}\label{Cond:2}
 The ratio $f(x) / \phi(x)$ is in 
 $L_{1}(\Omega, f)$. 
\end{Assumption}

We define  $\widetilde{f}(x)$ on $\Theta$ to be 
\begin{equation}\label{def:ftilde}
    \widetilde{f}(x) = \begin{cases}
    f(x),& x \in \Omega,\\
    0,& x \in \Theta \setminus \Omega.
    \end{cases}
\end{equation}
Since $H$ is an orthonormal system on $\Theta$ with respect to pdf $\phi$, any function in $L_{2}(\Theta, \phi)$ can be expanded into a Fourier series (see, e.g., \cite{KolmogorovFomin1976} or \cite{Rudin1976}) along the $H$ basis  elements. Under  Assumption \ref{Cond:1} or~\ref{Cond:2},  $g(x) = \widetilde{f}(x) / \phi(x)$ satisfies 
\begin{equation}\label{g_L2}
    \int\limits_{\Theta}g^{2}(x)\phi(x)dx = 
    \int\limits_{\Theta \setminus \Omega} \frac{\widetilde{f}(x)}{\phi(x)}\widetilde{f}(x)dx + \int\limits_{\Omega}\frac{f(x)}{\phi(x)}dF_{X}(x) < \infty,
\end{equation}
so that $g(x) \in L_{2}(\Theta, \phi)$. 
In consequence, $g$ has a Fourier series representation
\begin{equation}\label{fraction_expansion}
    g(x) = \sum\limits_{i = 0}^{\infty}\alpha_{i}h_{i}(x), 
\end{equation}
with
\begin{align*}
    \alpha_{i} &= \left<g, h_{i}\right>_{\phi} = \int\limits_{\Theta} g(x) h_{i}(x)\phi(x)dx = \int\limits_{\Theta} \frac{\widetilde{f}(x)}{\phi(x)} h_{i}(x) \phi(x)dx \\
    &= 
    \int\limits_{\Theta} \widetilde{f}(x)  h_{i}(x)dx 
    = \int\limits_{\Omega} f(x) h_{i}(x)dx + \int\limits_{\Theta / \Omega} \widetilde{f}(x)  h_{i}(x)dx = 
    \left<1, h_{i}\right>_{f}.
\end{align*}
The series in \eqref{fraction_expansion} converges in $L_2(\Theta,\phi)$. From \eqref{def:ftilde} and \eqref{fraction_expansion}, an estimator of $f$ is
\begin{equation} \label{estimator_3}
    \hat{f}(x) = \phi(x)\sum\limits_{i = 0}^{n}\left<1, h_{i}\right>_{f}h_{i}(x). 
\end{equation}
Each polynomial $h_{i}(x)$ is a sum of monomials, $h_{i}(x) = \sum_{j=0}^{i}a_{ij}x^{j}, i = 0,\ldots, n$. Since the first $n$ moments of $f(x)$ are known, 
\begin{equation}\label{aij}
    \left<1, h_{i}\right>_{f} = \sum_{j=0}^{i}a_{ij}\left<1, x^{j}\right>_{f} = \sum_{j=0}^{i}a_{ij}m_{j}, 
\end{equation}
where $m_j$ is the $j$th \textit{raw} moment of $X$ for $j=0,\ldots, i$, $i=0,\ldots, n$. 

\begin{definition}
The series-based estimator \eqref{estimator_3} of the pdf $f$ of $X$  is called a \textit{K-series estimator with reference $\phi$}, or simply \textit{K-series}.
\end{definition}

Let $\A = \{a_{ij}\}_{i, j = 0}^{n}$ be a lower triangular matrix with entries the coefficients of the \textit{ordered} vector of polynomials $h_{i}(x)$, $\hbf_n(x) = (h_{0}(x), \ldots, h_{n}(x))^{T}$  from $H$. 
Then,  \eqref{estimator_3} can also be computed by
\begin{equation} \label{estimator_3_mat}
    \hat{f}(x) = \phi(x) \left( \A \cdot \mbf_n \right)^{T} \cdot \hbf_n(x). 
\end{equation}
The only requirements for the \textit{K-series} estimator are (a) the unknown target pdf $f$ have bounded support and (b) the support of the reference $\phi$ be large enough to cover it. 
The only constraint for the choice of the reference distribution is to be continuous with support larger than that of the target pdf. Any such pdf can serve as a reference and thus polynomials $h_i$ of any order can be computed using the Gram-Schmidt orthogonalization procedure in \eqref{estimator_3}.  

There is no technical necessity to strictly adhere to the ordered sequence of the sequence of moments. It is possible to use moments of any order. What is necessary is to generate orthogonal polynomials in equation \eqref{estimator_3} in a specific sequence. 
For example, if moments of order 1, 5, and 12 are available, we can construct a system of orthogonal polynomials from the monomial set $\{1, x, x^5, x^{12}\}$ by the Gram-Schmidt process.

\subsection{K-series estimation in practice}\label{sec:practice}

We illustrate K-series estimation with two examples. For the first, we let the target pdf be truncated exponential with known parameters and support and derive its first two moments and its K-series estimate. In the second (Irwin-Hall Distribution), we express the distribution generating algorithm as a prob-solvable loop, compute its \textit{exact} moments using the \textsc{Polar} tool \cite{Moosbruggeretal2022} and then its K-series estimate.


\textbf{Truncated Exponential.} Suppose $X \sim  Trunc$ $Exp(1, [0,1])$ with  support $\Omega = \left[0, 1\right]$. We assume the  first two moments are known, specifically, we let 
$M_{2} =\{m_1=  0.418023, m_2=0.254070\}$, and 
 the reference distribution is uniform with the same support as the target unknown distribution; i.e.,
$ \phi(x) = 1$ for $x \in \left[0, 1\right]$. 

Legendre polynomials $l_{n}(\tau)$ are a standard basis of orthogonal polynomials on the interval [-1,1] with a weight function of 1. Consequently, for any uniform pdf on an arbitrary bounded interval, a corresponding set of orthonormal polynomials can be derived from the standard Legendre polynomials through the substitution $\tau \rightarrow (\tau - \mu) / \sigma$ and subsequent normalization. 

Since $\phi$ is uniform, we use the shifted and scaled Legendre polynomials as the orthonormal basis in the series (see \cite{XiuKarniadakis2002a}); $\bar{l}_{0} = 1$, $\bar{l}_{1} = \sqrt{3}(2x-1)$, $\bar{l}_{2} = \sqrt{5}(6x^{2}-6x+1)$.
To compute the unknown pdf estimator in \eqref{estimator_3}, we need to compute the $\alpha_{i}$  coefficients in \eqref{fraction_expansion}. By \eqref{aij}, this  requires the substitution of  $x^{i}$ with the corresponding moment $m_{i}$ in $M_{2}$, for $i=1, 2$. Doing so yields $\alpha_{0} = 1, \alpha_{1} = \sqrt{3}(2 \cdot 0.418023 - 1) = -0.283976$, $\alpha_{2} = \sqrt{5}(6 \cdot 0.25407 - 6 \cdot 0.418023 + 1) = 0.036407$. The K-series estimator is  \[\hat{f}(x) = 1 - 0.283976\cdot \bar{l}_{1}(x) + 0.036407\cdot \bar{l}_{2}(x),\] and almost fully coincides with the true truncated exponential pdf in panel (a) of Figure~\ref{fig:3}.

\begin{figure}[htbp]
    \begin{minipage}{.55\textwidth}
    \centering
    \vspace{0.3cm}
    \includegraphics[height=.63\textwidth]{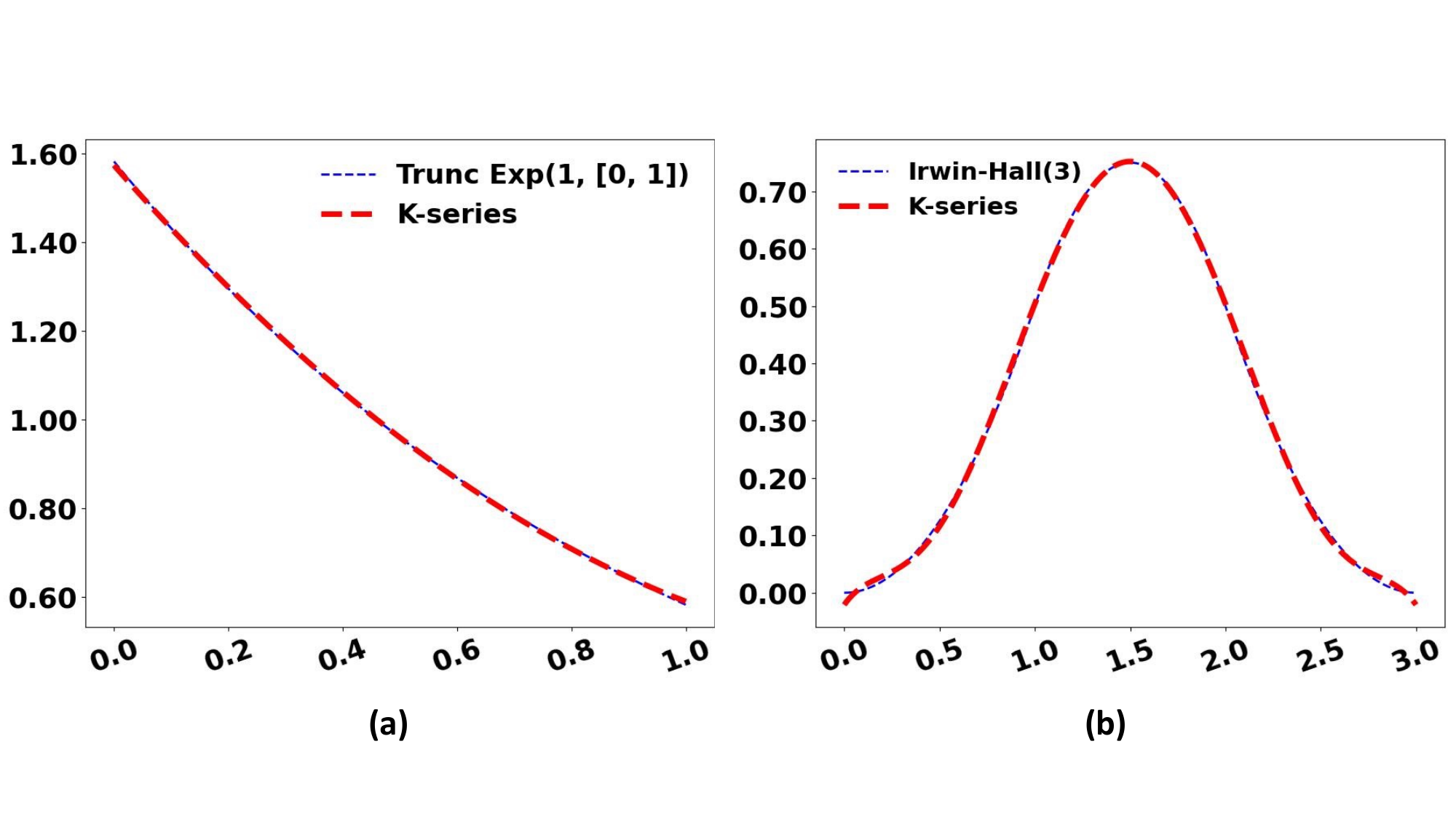}
    \end{minipage}
    \hspace{0.8cm}
    \begin{minipage}{.35\textwidth}
    \centering
\begin{lstlisting}[mathescape=true,frame=none,
%caption={(c)}, 
backgroundcolor=\color{white},
label = {lst:IrH}]
x := 0
while true:
    u := Uniform(0, 1)
    x := x + u
end
\end{lstlisting}
\end{minipage}
\caption{K-series approximation of a truncated exponential distribution (panel (a)) and the Irwin-Hall distribution (panel (b)).}
\label{fig:3}
\end{figure}

\textbf{The Irwin-Hall Distribution.} 
Irwin–Hall is the probability distribution of a sum of independent uniform random variables on the unit interval (uniform sum distribution). That is, 
$X \sim $ Irwin–Hall$(t)$ if $X=\sum_{i=1}^t U_i$, for $U_i$ independent and identically distributed (i.i.d.) as Uniform(0,1). 
This distribution, parameterized by the number of its summands, is encodable as the prob-solvable loop in the right panel of Fig.~\ref{fig:3}.

At each iteration $t$, the support of $x$ is $(0, t)$. 
Since the Irwin-Hall distribution is equivalent to a prob-solvable loop, its exact $n$ first moments can be computed with the algorithm in \cite{BartocciKS19}:
\begin{align*}\notag
    M(t) = \Bigg\{&\frac{t}{2}, \frac{t(3t + 1)}{12}, \frac{t^2(t+1)}{8}, \frac{t(15t^3+30t^2 + 5t - 2)}{240}, \\ &\frac{t^2(3t^3 + 10t^2 + 5t - 2)}{96}, \frac{t(63t^5+315t^4 + 315t^3 - 91t^2 - 42t + 16)}{4032}, \ldots\Bigg\}.
\end{align*} 



The first 6 moments of Irwin-Hall (3) are
$M_{6}(3) = \left\{\frac{3}{2}, \frac{5}{2}, \frac{9}{2}, \frac{43}{5}, \frac{69}{4}, \frac{3025}{84}\right\}.$
We use the Uniform$\left[0, 3\right]$ as a reference and construct the K-series estimator of the pdf of $x$ at iteration $t = 3$ with the 6 first moments and the first 7 shifted and scaled Legendre polynomials. The true pdf and its K-series estimate are plotted in panel (b) of Fig.~\ref{fig:3}, where we can see their almost perfect agreement. 

While iteration $t = 3$ is used for illustration purposes, iteration  number is not important for our method. One only needs to specify an appropriate reference and support (for the uniform reference the support is $\left[0, t\right]$). Alternatively, we can use a reference that has the appropriate support for any iteration; normal, truncated normal, gamma, etc.

\subsection{Special cases of K-series}\label{sec:general}

The K-series density estimator generalizes the widely used Gram-Charlier (GC) series density estimator. GC represents the pdf $f$ of a random variable $X$ as a series in terms of its cumulants and a normal reference distribution $\phi$ by using Hermite polynomials  (see, e.g., \cite{Cramer1957,KendallStuart1977}). 
The GC (type-)A  estimate of the pdf $f$ of $X$ is given  
by 
\begin{align}\label{pdfGCGeneralC}	
f_{\gc}(x) &= 
 \phi(x)\sum_{n=0}^\infty (-1)^n c_n He_n \left(x\right),
\end{align}
where $c_n=(-1)^n \int_{-\infty}^\infty f(t) He_n(t)dt / n!$, $\phi$ is the standard normal pdf and $$ He_n \left(x\right)=n! \sum_{k=0}^{[n/2]} \left((-1)^k x^{n-2k}\right)/\left(k!(n-2k)! 2^k\right)$$
Proposition \ref{prop:GC} shows that the GC series A estimator in \eqref{pdfGCGeneralC} is a special case of the K-series estimator. 

\begin{proposition}\label{prop:GC}
Suppose the reference pdf $\phi$ is normal with mean and variance corresponding to the first and second moments of the target pdf $f$. Then, the K-series estimator \eqref{estimator_3} equals the Gram-Charlier estimator \eqref{pdfGCGeneralC}. 
\end{proposition}

Proposition \ref{prop:GC} is easy to obtain using the standard normal as reference pdf and replacing the polynomials $h_i$ in \eqref{estimator_3} by $He_i/\sqrt{i!}$.  

\cite{Munkhammar_etal_2017} developed the \textit{Method of Moments (MM)} estimation algorithm for parameters of a target distribution $f$ by equating sample moments with the corresponding moments of the distribution. The approximation is carried out on the interval where they wish to maximize accuracy. In practice, this is the same as assuming finite or bounded support.  \cite{Munkhammar_etal_2017} showed that MM beats the GC expansion for several distributions, such as the Weibull on a positive finite support, in simulation experiments.

The MM algorithm 
starts by choosing an interval $\left[a, b\right]$ that is thought to contain most of the mass of the target unknown distribution. Using a finite set of moments $\{m_{1}, \ldots, m_{n}\}$ and $m_{0} = 1$, MM constructs a polynomial estimator $\hat{f}(x)$ by solving a linear system of equations,
\begin{equation}\label{MM_estimator}
    m_{i} = \int\limits_{a}^{b}x^{i}\hat{f}(x)dx, \hspace{0.1cm} i = 0,\ldots,n,  
\end{equation}
which yields the coefficients $p_i$ of the series representation $\hat{f}_{\mm}(x) = \sum_{i=0}^n  p_{i}x^i$. 

Let $\mathbf{m}_n=(1,m_1,\ldots,m_n)^{T}$,  $\mathbf{p}_n=(p_0,p_1,\ldots, p_n)^{T}$, and $\mathbf{x}_n=(1,x,\ldots, x^n)^{T}$. The linear system \eqref{MM_estimator} can be expressed in matrix form as
   $ \mathbf{m}_n = \M_{ab} \cdot \mathbf{p}_n$,
where  $\M_{ab}$ is the matrix with elements the integrals of powers of $x$ over the interval $\left[a, b\right]$.
Theorem \ref{thm:MM} shows that MM 
is a special case of the K-series estimator.
Its proof is provided in Appendix~\ref{app:proofs}.

\begin{theorem}\label{thm:MM}
Suppose the reference pdf $\phi$ is the uniform with  the same support as the target pdf $f$. Then, the MM estimator 
coincides with the K-series estimator \eqref{estimator_3}.
\end{theorem}

MM and GC are special cases of  K-series estimation. As such, they also enjoy the theoretical properties of K-series in the constrained setting in which they apply. We next show in   Theorem~\ref{thm:L1_convergence} that the general K-series estimator \eqref{estimator_3} converges to the true target pdf. Its proof is given in Appendix~\ref{app:proofs}.

\begin{theorem}\label{thm:L1_convergence}
Let $\phi(x)$ be continuous, positive everywhere on $\Theta$: $\Omega \subseteq \Theta$ and either 
(a) $\Theta$ is unbounded and $\phi(x)$ is uniquely identifiable by its moments, or (b) $\Theta$ is finite (bounded). 
Under  Assumption \ref{Cond:1} or \ref{Cond:2}, the K-series estimator \eqref{estimator_3} converges to the true pdf \eqref{def:ftilde}, $\widetilde{f}(x)$, in $L_{1}(\Theta, 1)$. Moreover, if $\phi(x)$ is a uniform pdf, it converges in $L_{2}(\Theta, 1)$.
\end{theorem}

The following theorem provides formal guarantees that the moments of the obtained estimate coincide with the corresponding moments of the target random variable based on which the estimate is constructed.

\begin{theorem}[Moment matching]\label{thm:moment_match}
Suppose the  K-series estimator \eqref{estimator_3} is constructed using the first $n$ moments $\{m_{1}, \ldots, m_{n}\}$ of the random variable $X$ with pdf $f(x)$ and set $m_{0} = 1$. Then,
\[\int\limits_{\Theta}x^{i}\hat{f}(x)dx = \int\limits_{\Omega}x^{i}f(x)dx = m_{i},\] 
for all $0 \leq i \leq n$. \end{theorem}

\begin{proof}
    Let $h_{i}(x)$ be the $i$th orthonormal polynomial with respect to the reference pdf $\phi(x)$ in \eqref{estimator_3}. Then, 
        $ \int\limits_{\Omega}h_{i}(x)f(x)dx = \alpha_{i}$.
Also, by the orthogonality of $h_{i}$s, the following holds
\begin{align*}
        \int\limits_{\Theta}h_{i}(x)\hat{f}(x)dx = \int\limits_{\Theta}h_{i}(x)\phi(x)\sum\limits_{j = 0}^{n}\alpha_{j}h_{j}(x)dx = 
        \sum\limits_{j = 0}^{n}\alpha_{j}\int\limits_{\Theta}h_{i}(x)\phi(x)h_{j}(x)dx = \alpha_{i}.
    \end{align*}
It remains to observe that any monomial $x^{i}$, $0 \leq i \leq n$, can be expressed as a  linear combination of the orthogonal polynomials $h_{j}$, $0 \leq j \leq i$.
\end{proof}

\subsection{Approximation of the support}\label{sec:suppapprox}

The space spanned by $\lfloor(n+1)/2\rfloor$ orthogonal polynomials with respect to the target density $f(x)$ can be constructed using the sequence of its first $n$ moments (see \cite{Szego1939}). The determinant
\begin{align}\label{ort_poly_matrix}
    D_{j}(x) = \begin{vmatrix}
        m_{0} & m_{1} & \ldots & m_{j} \\
        m_{1}  & m_{2} & \ldots &  m_{j+1} \\
        \vdots & \vdots & \ddots & \vdots \\
        m_{j-1} & m_{n} & \ldots & m_{2j-1} \\
        1 & x & \ldots & x^{j}
    \end{vmatrix}
\end{align}
defines the corresponding orthogonal polynomial (non-normalized) of degree $j$. That is, if the first $n$ moments of a random variable $X$ are known, then we can construct the first $\lfloor(n+1)/2\rfloor$ orthogonal polynomials. 

We let $e_{j}(x)$ denote an orthogonal polynomial of degree $j$ of the random variable $X$.
Theorems \ref{thm:zeros1}, \ref{thm:zeros2} and \ref{thm:zeros3} (see \cite{Szego1939,Chihara1978}) state elementary properties of zeros of orthogonal polynomials.

\begin{theorem}\label{thm:zeros1}
    Let $\Omega$ be an interval which is a supporting set for the distribution of $X$. The zeros of $e_{j}(x)$ are all real, simple and are located in $\Omega$.
\end{theorem}

\begin{theorem}\label{thm:zeros2}
    Between two zeros of $e_j(x)$ there is at least one zero of $e_{i}(x)$, $i > j$.
\end{theorem}

\begin{theorem}\label{thm:zeros3}
    The zeros $\{x_{j, \nu}\}_{\nu = 1}^{j}$ and $\{x_{j+1, \nu}\}_{\nu = 1}^{j+1}$ of $e_{j}(x)$ and $e_{j+1}(x)$ respectively, mutually separate each other. That is, 
    $$ x_{j+1, \nu} < x_{j, \nu} < x_{j+1, \nu+1}, \hspace{0.8cm} \nu = 1, \ldots, j$$
\end{theorem}

From Theorems \ref{thm:zeros1},  \ref{thm:zeros2} and \ref{thm:zeros3}, we can conclude that all zeros of orthogonal polynomials are simple and located precisely within the interior of the support. Moreover, as the polynomials' degree increases, the distance between the two outermost zeros also increases, resulting in a more accurate inner approximation of the random variable's support. The higher number of moments available, the higher the degree of polynomials that can be obtained, and the more accurate the estimation of the support becomes.
One only needs to calculate the polynomial of the highest possible degree using formula \eqref{ort_poly_matrix}, determine its zeros, and identify the lowest and highest values. 

We demonstrate this method using the example of the Irwin-Hall distribution in Sec. \ref{sec:practice}.
Let us suppose, that the first 6 moments of the random variable $X$ are available: $M_{6} = \left\{\frac{3}{2}, \frac{5}{2}, \frac{9}{2}, \frac{43}{5}, \frac{69}{4}, \frac{3025}{84}\right\}.$ We are interested in the minimum possible support of the pdf of  $X$. 
Since the first 6 moments are known, we can construct orthogonal polynomials of the random variable $X$ up to degree  $\lfloor(n+1)/2\rfloor=3$. Applying \eqref{ort_poly_matrix} yields the highest degree computable polynomial,
\begin{align}\label{ort_poly_matrix_example}
    D_{3}(x) = \begin{vmatrix}
        1.00 & 1.50 & 2.50 & 4.50 \\
        1.50  & 2.50 & 4.50 &  8.60 \\
        2.50 & 4.50 & 8.60 & 17.25 \\
        1 & x & x^{2} & x^{3}
    \end{vmatrix} = 0.025x^{3} - 0.1125x^{2} + 0.1525x - 0.06.
\end{align}
The polynomial in \eqref{ort_poly_matrix_example} has 3 distinct roots: $\{0.693774, 1.5, 2.306226\}$. Since all the roots belong to the interior of the support, the inner approximation of the support is $\left[0.693774, 2.306226\right]$.

\subsection{Validity of the input}\label{ssec:quality_control}

Not every sequence of real values can form a valid set of moments for any probability distribution. This issue is known as the Hamburger moment problem (see \cite{Chihara1978}). Given a sequence of real numbers $\{m_{i}\}_{i=0}^{\infty}$, the question is whether there exists a positive Borel measure $F$ such that 
\[ \int\limits_{-\infty}^{\infty}x^{i}dF(x) = m_{i}, \hspace{0.5cm} i = 0, 1, 2, \ldots \]
We introduce a procedure to examine whether the input set of values can be moments of a distribution. We require the input sequence of moments to be consecutive and without gaps.
Since we are dealing with a truncated set of moments, we refer to it as the \textit{truncated moment problem}. 
Let  
\begin{align}\label{mom_determinants}
    \Delta_{r} = \mathrm{det}(m_{i+j})_{i,j=0}^{r} = \begin{vmatrix}
        m_{0} & m_{1} & \ldots & m_{r}\\
        m_{1} & m_{2} & \ldots & m_{r+1} \\
        \vdots & \vdots & \ddots & \vdots \\
        m_{r} & m_{r+1} & \ddots & m_{2r} 
    \end{vmatrix}
\end{align}
be a sequence of determinants.
\begin{theorem}\cite{Chihara1978}\label{thm:Hamburger} The Hamburger moment problem has a solution if and only if the determinants $\Delta_{r}$ in \eqref{mom_determinants} are all positive.
\end{theorem}
By Theorem \ref{thm:Hamburger}, the truncated moment problem admits a solution only if all the determinants $\Delta_{r}, \hspace{0.1cm} r = 0,\ldots, \lfloor n/2 \rfloor$, are positive. 
The complete process of univariate K-series estimation is described in Algorithm~\ref{alg:Univar_K_series}.

\subsection{Multivariate K-series}\label{ssec:multivar_K_series}

K-series density estimation is easily generalizable to multivariate  distributions by considering the product of independent univariate distributions as the reference joint pdf. The coefficients of the corresponding multivariate orthogonal polynomials recover the multivariate dependence structure via their joint moments.

Let  $\X=(X_1,\ldots,X_k)^T$  be a vector of continuous random variables with joint 
non-negative pdf $f(\x)$, upper bounded and supported on $\Omega$ with countable zeros. 
Suppose that there exists a $k$-dimensional compact cube $\Qcal$, such that $\Omega \subseteq \mathcal{Q}$.
We assume that a finite number of moments, not necessarily an equal number for all, is known for each $X_{j}$, $j=1,\ldots, k$, and all cross-product moments are also known. That is,  we assume the set 
\begin{equation}\label{cross_moments}
    M_{d_{1},\ldots,d_{k}} = \left\{m_{i_{1},\ldots,i_{k}}=\E\left(X_1^{i_1}\ldots X_k^{i_k}\right):i_j=0,\ldots, d_j, d_j \in \mathbb{N}, j=1,\ldots, k\right\}
\end{equation}
is known.
Let $\Z=(Z_1,\ldots, Z_k)^T$ be a vector of continuous independent random variables and $\widetilde{\phi}(\z) = \prod_{j = 0}^{k}\phi_{j}(z_{j})$ be its pdf that is positive everywhere on its support $\Theta$, where $\Omega \subseteq \Theta$. We require either $\Theta$ be unbounded and  $\widetilde{\phi}(\z)$ uniquely identifiable by its moments, or $\Theta$ be bounded (see \cite{Rahman2018}).

\resizebox{1\columnwidth}{!}{%
\begin{algorithm}[H]\label{alg:Univar_K_series}
\caption{Univariate K-series procedure}
\label{alg:generator}
\SetKwProg{generate}{Construct}{}{end}
\SetKwProg{create}{Create}{}{end}
\SetKwProg{Proc}{Proc}{}{end}
\SetKwInput{Input}{Input}
\SetKwInput{Output}{Output}
\SetKwInput{Compute}{Compute}
\SetKwInput{return}{return}
\SetKwInput{Subs}{Substitute}
\SetKwInput{Searchfor}{Search for}
\Input{ 
\begin{itemize}
    \item $\{m_{i}\}_{i = 0}^{n}$ - sequence of $n$ moments, $m_{0}=1$\\
    \item $\phi(x)$ - reference pdf\\
    \item $\Theta$ - support of the reference
\end{itemize}}
\Output{ \begin{itemize}
    \item \textit{True} / \textit{False} - is the sequence $\{m_{i}\}_{i = 0}^{n}$ feasible?
    \item $\left[root_{min}, root_{max}\right]$- inner approximation of the support
    \item $\hat{f}(x) = \phi(x)\sum\limits_{i = 0}^{n}\left<1, h_{i}\right>_{f}h_{i}(x)$- K-series estimator
\end{itemize}}

\Compute {Determinants $\Delta_{r}$ according to \eqref{mom_determinants}, $0 \leq r \leq \lfloor n/2 \rfloor$.}
    \If{ $\exists r$:  \hspace{0.1cm} $\Delta_{r} \leq 0$}{
        \return{\textit{False}}
    }
\tcc{Approximation of the support}

\Compute {Orthogonal polynomial $e_{\lfloor (n+1) \rfloor/2}$ of the highest degree using \eqref{ort_poly_matrix}.}
\Searchfor {The lowest and the highest roots of  $e_{\lfloor (n+1) \rfloor/2}$: $\{root_{min}, root_{max}\}$}

\tcc{Orthogonal Polynomials Construction:}
    $h_{0}(x) = 1$\\
    \ForAll{$i$ in $\{1, 2, \ldots, n\}$}{
    { \tcc{Gram-Schmidt Orthogonalization} 
        $\widetilde{h}_{i}(x) = x^{i} - \sum\limits_{j=0}^{i-1}\frac{\left<x^{i}, h_{j}(x)\right>_{\phi}}{\left< h_{j}(x), h_{j}(x)\right>_{\phi}}$\; 
        $h_{i}(x) = \widetilde{h}_{i}(x) / \parallel \widetilde{h}_{i}(x) \parallel_{\phi}$;

     }}
     \ForAll{polynomial $h_{i}$ in $\{h_{1}, \ldots, h_{n}\}$}{
        \ForAll{monomial $x^{j}$ in $h_{i}(x) = \sum\limits_{j = 0}^{i}a_{ij}x^{j}$}{
            \Subs {$x^{j} \leftarrow m_{j}$}
        }
        \Compute{Fourier coefficients $\alpha_{i} = \left<h_{i}(x), 1\right>_{f} =  \sum\limits_{j = 0}^{i}a_{ij}m_{j}$}
     }
     \Compute{ $\hat{f}(x) = \phi(x)\sum\limits_{i = 0}^{n}\left<1, h_{i}\right>_{f}h_{i}(x)$}
\return{$\textit{True}$, $\left[root_{min}, root_{max}\right], \hat{f}(x)$}
\end{algorithm}
}

\medskip
Let 
\begin{equation}
   \Tilde{h}_{i_{1},\ldots,i_{k}}(\z) = \prod\limits_{j=1}^{k}h_{i_j}^{j}(z_{j}),
\end{equation}
where $h_{i_j}^{j}(z_{j})$ is a polynomial of degree $i_{j}$ that belongs to the set of orthogonal polynomials with respect to $\phi_{j}(z_{j})$, $i_j=0,\ldots, d_j$, $j=1,\ldots, k$, that are calculated with the Gram-Schmidt orthogonalization procedure. 
The set  
$H = \{\Tilde{h}_{i_{1},\ldots,i_{k}}(\z),  i_j=0,\ldots, d_j, d_j \in \mathbb{N}, j=1,\ldots, k\}$
contains the $k$-variate orthonormal polynomials on $\Theta$ with respect to $\widetilde{\phi}(\z)$. 
As in the univariate case, we require Assumption  \ref{Cond:1} hold and  let 
\begin{equation}
    \widetilde{f}(\z) = \begin{cases}
    f(\z),& \z \in \Omega,\\
    0,& \z \in \Theta \setminus \Omega.
    \end{cases}
\end{equation}
Then, $ \widetilde{f}(\z)/\widetilde{\phi}(\z) = g(\z)$  is approximated by 
\begin{equation*}
   \hat{g}(\z) = 
    \sum_{\substack{i_j \in \{0,\ldots,d_j\},\\ j=1,\ldots, k}} \alpha(i_1,\ldots,i_k) \Tilde{h}_{i_{1},\ldots,i_{k}}(\z)
= \sum_{\substack{i_j \in \{0,\ldots,d_j\},\\ j=1,\ldots, k}} \alpha(i_1,\ldots,i_k)\prod\limits_{j=1}^{k}h_{i_j}^{j}(z_{j}), 
\end{equation*}
where 
the Fourier coefficients $\alpha(i_1,\ldots,i_k)$ are calculated as follows.
\begin{align}
    \alpha(i_1,\ldots,i_k) &= \left<g, \Tilde{h}_{i_{1},\ldots,i_{k}}\right>_{\widetilde{\phi}} = \int\limits_{\Theta} g(\z) \Tilde{h}_{i_{1},\ldots,i_{k}}(\z)\widetilde{\phi}(\z)d\z \notag
    \\
    &= \int\limits_{\Omega}f(\z)  \Tilde{h}_{i_{1},\ldots,i_{k}}(\z)d\z + \int\limits_{\Theta / \Omega} \widetilde{f}(\z)  \Tilde{h}_{i_{1},\ldots,i_{k}}(\z)d\z \notag \\
    &= 
    \left<1, \Tilde{h}_{i_{1},\ldots,i_{k}}(\z)\right>_{f}. \label{in.prod}
\end{align}
Since for all $i_{j} = 0,\ldots, d_{j}, j = 1,\ldots, k,$  $h_{i_j}^{j}(z_{j})= \sum_{l=0}^{i_{j}}a^{j}_{i_{j}l}z_{j}^{l},$  their product is
\begin{align*}
    \Tilde{h}_{i_{1},\ldots,i_{k}}(\z) &= \prod\limits_{j = 1}^{k}h_{i_j}^{j}(z_{j}) = \prod\limits_{j = 1}^{k}\sum_{l=0}^{i_{j}}a^{j}_{i_{j}l}z_{j}^{l} 
    = 
    \sum_{\substack{l_j \in \{0,\ldots,i_{j}\},\\ j=1,\ldots, k}}  z_1^{l_1}  \cdots z_k^{l_k}\prod\limits_{j = 1}^{k}a^{j}_{i_{j} l_{j}}.
\end{align*}
Assuming all first cross-moments of $f(\z)$, $m_{l_1,\ldots,l_k} = \E_{f}\left(Z_1^{l_1}  \cdots Z_k^{l_k} \right)$ are known, we can compute \eqref{in.prod} as
\begin{align}\label{K_series_multivar_coefs}
    \left<1, \Tilde{h}_{i_{1},\ldots,i_{k}}(\z)\right>_{f} &= \sum_{\substack{l_j \in \{0,\ldots,i_{j}\},\\ j=1,\ldots, k}}  \left<1, z_1^{l_1}  \cdots z_k^{l_k}\right>_{f}\prod\limits_{j = 1}^{k}a^{j}_{i_{j} l_{j}} = \sum_{\substack{l_j \in \{0,\ldots,i_{j}\},\\ j=1,\ldots, k}}   m_{l_1,\ldots,l_k}\prod\limits_{j = 1}^{k}a^{j}_{i_{j} l_{j}}. 
\end{align}
The \textit{multivariate K-series} estimator of $f$ is  
\begin{equation} \label{mult_k_series}
    \hat{f}(\x) = \widetilde{\phi}(\x)\sum_{\substack{i_j \in \{0,\ldots,d_j\},\\ j=1,\ldots, k}} \left<1, \Tilde{h}_{i_{1},\ldots,i_{k}}(\z)\right>_{\f} \Tilde{h}_{i_{1},\ldots,i_{k}}(\z). 
\end{equation}
A probabilistic loop application of K-series estimation is shown in Fig.~\ref{fig:DDMR}, where the pdf of the location $(X,Y)$ of the \textit{Differential-Drive Mobile Robot} is estimated, assuming it is characterizable by its moments. The joint and marginal distributions of the location variables $X$ and $Y$ are derived from a finite set of moments at iteration $t = 25$. 
The value of 25 was chosen to provide a non-trivial example of the capabilities of our approach. Moreover,  it serves as a juxtaposition to competing tools (such as $\lambda$PSI \cite{Gehretal2020}), which fail to generate a meaningful expression even by iteration 5.
We use the first $6$ moments for the marginals and the first $48$ moments for the joint distribution. 
The  middle and right top panels of Fig.~\ref{fig:DDMR}) plot the marginal pdfs of $X$ and $Y$, respectively. The histograms are based on $10^{6}$ draws from the true marginals. Our K-series estimates are in dashed red and agree almost perfectly with the true marginal pdfs. The left bottom panel plots the estimated joint pdf of $(X,Y)$. The right panel draws comparative frequency bar plots of $10^{6}$ true and estimated values of the bivariate random variable $(X,Y)$ over 2-dimensional grids of the support of the bivariate distribution. 
Our estimate (red bars) practically coincides with the true joint pdf (blue bars) over the grid. 

Another illustration of this algorithm on the truncated bivariate normal is given in Appendix~\ref{app:bivar_normal}.

\section{Symbolic K-series representation along iterations}
\label{symbolic}

In this section, we demonstrate the unique ability of our method to express the distribution of one or multiple state variables as a function of the iteration number in closed form.

We introduce the semantics of  \emph{prob-solvable loops}, introduced by \cite{Bartoccietal2019}, as we are considering infinite probabilistic loops and the properties of state variables at each iteration. For the class of \emph{prob-solvable loops}, moments of all orders of program variables can be symbolically computed. Given a prob-solvable loop and a program variable $x$, \cite{Bartoccietal2019} calculate a closed-form solution for $\E(x_t^k)$ for any arbitrary $k \in \mathbb{N}$, with $t$ representing the $t$-th loop iteration. Prob-solvable loops were initially restricted to polynomial variable updates. \cite{Kofnovetal2022} relaxed the restriction to allow square-integrable function updates.

\begin{definition}[Prob-solvable loops \cite{Bartoccietal2019,Kofnovetal2022}]\label{def:probsolvable}
Let $m \in \nat$ and $x_1,\ldots x_m$ denote real-valued program
variables. 
A Prob-solvable loop with program variables $x_1,\ldots x_m$ is a loop of the form
\begin{equation*}\label{eq:ProbModel}
  I; \texttt{ while true: } U \texttt{ end}, \quad \text{such that}
\end{equation*}
\begin{itemize}
    \item $I$ is a sequence of initial assignments over a subset of $\{x_1,\ldots, x_m\}$. The initial values of $x_i$ can be drawn from a known distribution. 
    They can also be real constants.
    \item $U$ is the loop body and a sequence of $m$ random updates, each of the form,
    \begin{equation*}\label{eq:ProbModel:prob_assignments}
        x_i = \textit{Dist} \quad \text{or} \quad x_i = a_i x_i + G_{i}(x_1,\dots x_{i-1}),
    \end{equation*}
    where $a_i \in \real$, $G_{i} \in \real[x_1,\ldots,x_{i-1}]$ is a square-integrable function over program
    variables $x_1,\ldots,x_{i-1}$ and \textit{Dist} is a random variable whose distribution is independent of program variables with computable moments.
    $a_i$ can be random variables with the same constraints as for \emph{Dist}.
\end{itemize}
\end{definition}

The K-series estimator can be expressed as a quantitative invariant in the sense that its formula is a function of loop iteration. In the univariate case, the K-series estimator \eqref{estimator_3} of the unknown pdf of the random variable $X$ is $\hat{f}(x)=\phi(x) \sum_{i=0}^n \left(\sum_{j=0}^i a_{ij}m_j\right)h_i(x)$, where $m_j=\E(X^j)$. The estimator is a function of the moments of $X$, which in turn, vary along iterations in a probabilistic loop. That is, the K-series estimator can be equivalently expressed as  
\begin{equation}\label{est:iter}
    \hat{f}_t(x) = \phi(x) \sum_{i=0}^{n} \left(\sum_{j=0}^i a_{ij}m_{j}(t)\right)h_{i}(x), 
\end{equation}
where $m_j(t)=\E(X^j(t))$ is the moment of the random variable $X$ at iteration $t$. Formula \eqref{est:iter} is the symbolic representation of the K-series pdf estimator as a function of iteration number. 
Similarly, the multivariate K-series estimator \eqref{mult_k_series} can be written as
\begin{align} \label{mest:iter}
    \hat{f}(\x) 
    &= \widetilde{\phi}(\x)\sum_{\substack{i_j \in \{0,\ldots,d_j\},\\ j=1,\ldots, k}}  \left(\sum_{\substack{l_j \in \{0,\ldots,i_{j}\},\\ j=1,\ldots, k}}   m_{l_1,\ldots,l_k}(t)\prod\limits_{j = 1}^{k}a^{j}_{i_{j} l_{j}} \right)\Tilde{h}_{i_{1},\ldots,i_{k}}(\z)
\end{align}
where $m_{l_1,\ldots,l_k}(t)=\E\left(X_1^{l_1}(t)\cdot \ldots \cdot X_k^{l_k}(t)\right)$ at iteration $t$, since the moments of the random vector depend on the iteration in a probabilistic loop. 

We illustrate \eqref{est:iter} by considering the probabilistic loop in Fig.~\ref{fig:example_symbolic}(A): the target random variable $r$ is modeled as the minimum of random variables $x$ and $y$. Variable $y$ is uniformly distributed on $(0, 20)$, while 
$x$ follows a mean-reverting process and is affected by the stochastic shock $\theta \sim $Uniform$(-8, 8)$ at each iteration.
We can now use the approach from \cite{BartocciKS19} to estimate moments for arbitrary iterations and use them to receive the symbolic expression for the pdf of $r$ for the corresponding iteration. Since $\min(\cdot, \cdot)$ is a non-polynomial function, we apply the approach in \cite{Kofnovetal2022} to represent $\min(\cdot, \cdot)$ as an expansion in orthogonal polynomials. The transformation is given in the bottom panel of Fig.~\ref{fig:example_symbolic}.

\begin{figure}[!h]
\centering
\includegraphics[scale=0.25]{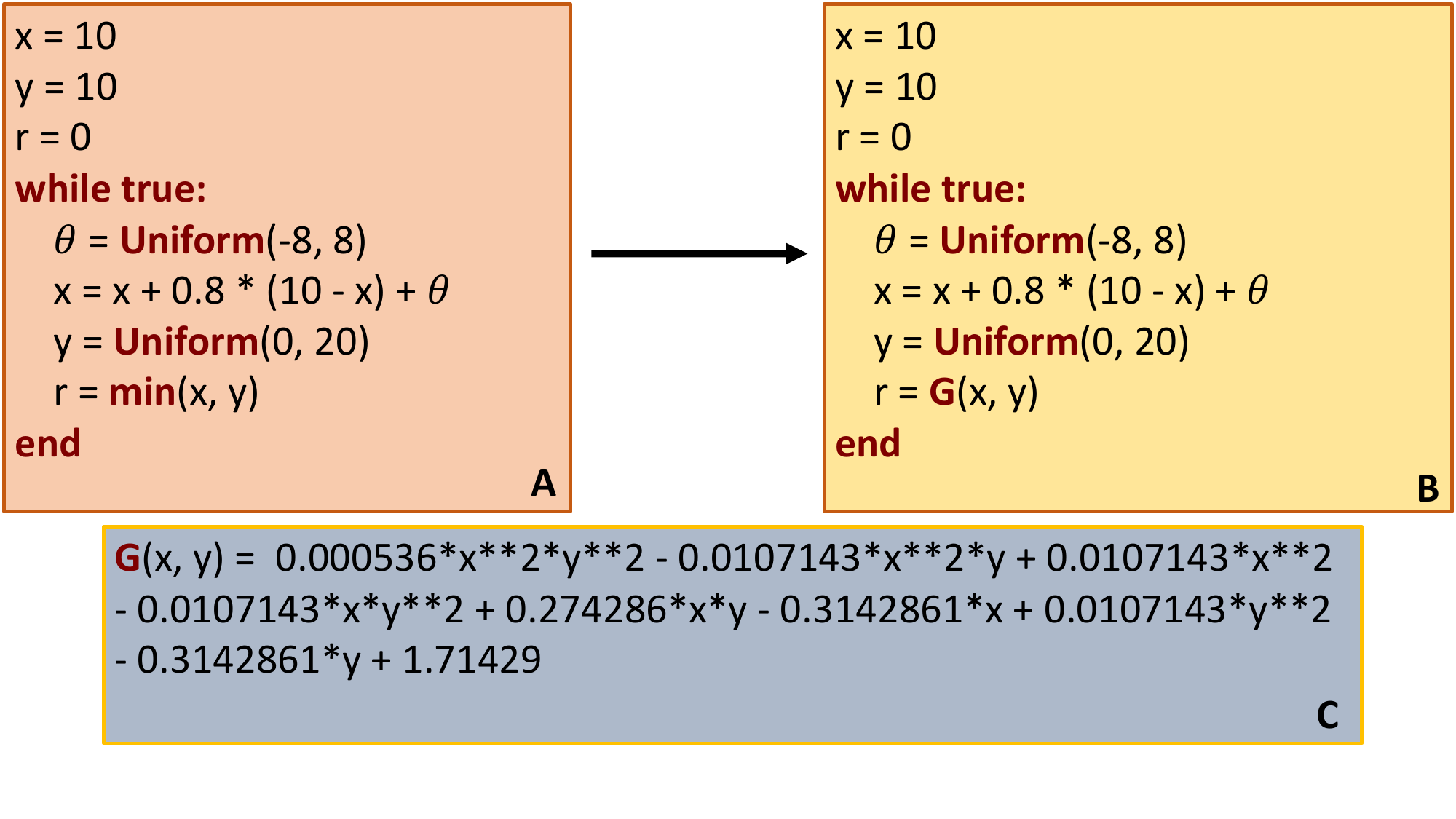}
  \caption{(A) Probabilistic loop with non-polynomial assignment, (B) Transformation of the program A using Polynomial Chaos Expansion~\cite{Kofnovetal2022}, by replacing the function $\min(\cdot, \cdot)$ with the polynomial $G(x,y)$.} 
  \label{fig:example_symbolic}
\end{figure}

Once this is computed, the program in Fig.~\ref{fig:example_symbolic} (B) can be handled using the algorithm in \cite{BartocciKS19}. The equations estimating the first four moments for each iteration are in the left panel of Fig.~\ref{fig:example_symbolic_moms}. We choose the uniform distribution on $(0, 20)$ as the reference pdf. We compute the shifted and scaled Legendre polynomials and substitute the moment equations as functions of iteration $t$. Similarly, we can derive the symbolic expression of the pdf estimate for any arbitrary iteration $t$. The right panel of Fig.~\ref{fig:example_symbolic_moms} plots the pdf estimate of the random variable $r$ at iteration $t = 30$ given by
\[\hat{f}_{30}(r) = 5.165866e-7*r^{4} + 2.561246e-5*r^{3} - 0.001472*r^{2} + 0.012320*r + 0.055246. \]

\begin{figure}[!h]
\centering
\includegraphics[width=10.5cm, height=5.5cm]
{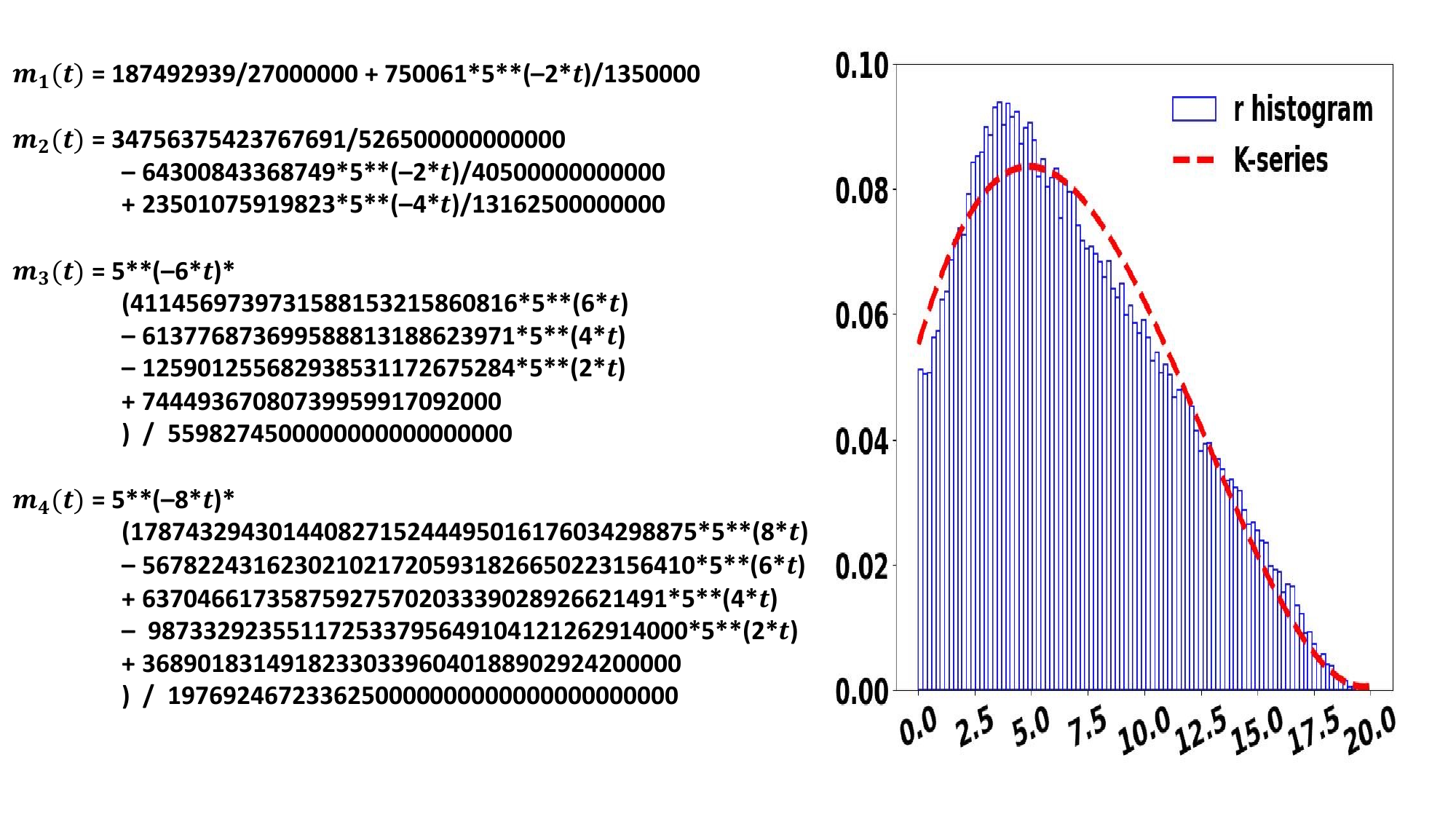}
  \caption{Left panel: First four moments expressed symbolically in the number of iterations. Right panel: Comparison between 
  the histogram of the sampling pdf and the 
  symbolic K-series estimation at $t=30$.}
  \label{fig:example_symbolic_moms}
\end{figure}

\section{Experiments}
\label{experiments}

We carried out K-series estimation of the distributions of the random variables generated in the execution of several probabilistic loops. The implementation code is available upon request. 
The first application is the \textit{Differential-Drive Mobile Robot} in Fig.~\ref{fig:DDMR}, where we observe a practically perfect approximation of both marginal and joint pdfs of the location of the robot. 
All experiments were conducted on a machine equipped with 16 GB of RAM and an Apple M1 Pro processor. The runtimes of all experiments in this section are displayed in Table \ref{tab:Experiments_time}. The Python code for the experiments in this paper can be found at \href{https://anonymous.4open.science/r/K-Series_TOPML-8D2F}{GitHub}. The runtimes of the other examples in this paper are reported in Table \ref{tab:Illustrations_time} in Appendix \ref{app:bivar_normal}. We distinguish between the time required for the Gram-Schmidt process and the time for the estimator construction. Our approach is highly time-efficient. Additionally, users can leverage precomputed standard sets of orthogonal polynomials to avoid recomputing them using the Gram-Schmidt process. Formal statistical tests for the goodness-of-fit of our estimates and the true (sampling) pdfs are carried out in Appendix~\ref{sec:KS} and the results, which overwhelmingly support our estimation procedure, are reported in Table~\ref{tab:Experiments_results}.

The  program in  Panel A1 of Fig.~\ref{fig:otherexamples} encodes the \textit{turning vehicle model} in \cite{Srirametal2020,Kofnovetal2022}. 
We use 
the truncated normal on $(1,18) \times (-15,15)$ with mean the sample mean  and variance  4 for $X$, and the sample variance of the  $Y$ distribution as reference pdfs. While the support of $X$ is not important, the accuracy of the estimation depends on the variance for $X$. When the variance is very small, the estimation becomes numerically unstable. This effect on the estimation is reflected in the K-series detecting, possibly artificially, two modes in Fig. \ref{fig:turning_vehicle_results_sv}.

\begin{figure}
     \centering
     \begin{minipage}{.32\textwidth}
        \centering
{\footnotesize
\begin{lstlisting}[mathescape=true,frame=none,
%caption={(c)}, 
backgroundcolor=\color{white},
label = {lst:TVM}]
v0 = 10, $\tau$ = 0.1, K = -0.5
$\psi$ = Normal(0, 0.01)
v = Uniform(6.5, 8.0)
x = Uniform(-0.1, 0.1)
y = Uniform(-0.5, -0.3)
while true:
    w1 = Uniform(-0.1, 0.1)
    w2 = Normal(0, 0.01)
    x = x + $\tau$ v cos($\psi$)
    y = y + $\tau$ v sin($\psi$)
    v = v + $\tau$ (K(v-v0)+w1)
    $\psi$ = $\psi$ + w2
end
\end{lstlisting}
}
\label{fig:(a)}
     \end{minipage}
     \hfill 
     \begin{minipage}{.65\textwidth}
         \centering
         \includegraphics[width=0.85\textwidth]{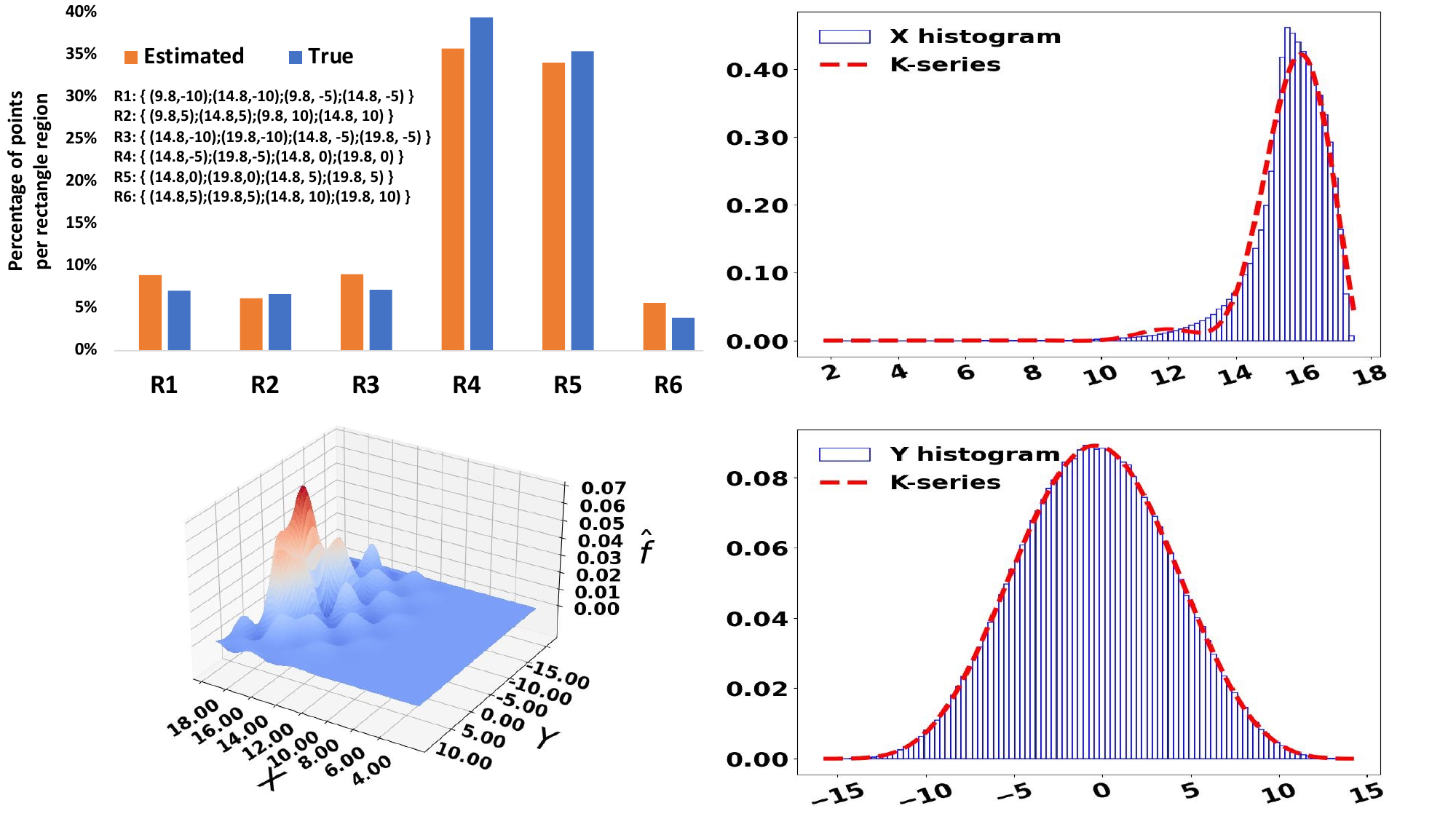}
  \label{fig:(b)}
     \end{minipage}
     
       \caption{Turning Vehicle Model: Code and K-series estimates of the marginal pdfs of $X$ and $Y$ (right upper and lower panels), the joint (lower left panel) and comparison bar plot (upper left panel) at iteration $t = 20$.}
        \label{fig:turning_vehicle_results_sv}
\end{figure}

The program in Panel A of Fig.~\ref{fig:otherexamples} is the same as the \textit{turning vehicle model} \cite{Srirametal2020,Kofnovetal2022} in Panel A1 of Fig.~\ref{fig:otherexamples}, 
with the difference that the variance of the basic random variables $\psi$ and $w_2$  is about 3 times larger. 
The effect of this on the joint and marginal distributions of $X,Y$ can be seen in Fig. \ref{fig:turning_vehicle_results}. 
In this case, the reference is truncated normal on $(-18,18) \times (-20,20)$ with mean the sample mean and variance the sample variance of the marginal distributions of $X$ and $Y$, respectively. 
While the support of $X$ is not important, the accuracy of the estimation depends on the variance for $X$. The K-series estimator is a sum of weighted orthonormal polynomials whose Gram-Schmidt orthogonalization with respect to the reference distribution involves the variance of the generated variables in the denominator. Thus, when the variance is very small, the fraction explodes and the estimation becomes numerically unstable. This can be managed by increasing the variance of the reference, as done in Fig. \ref{fig:turning_vehicle_results}. 

\begin{wrapfigure}{l}{0.45\textwidth}
 \centering
 \vspace{-0.65cm}
 \hspace{-0.02cm}
 \includegraphics[width=0.44\textwidth, height=0.31\textwidth]{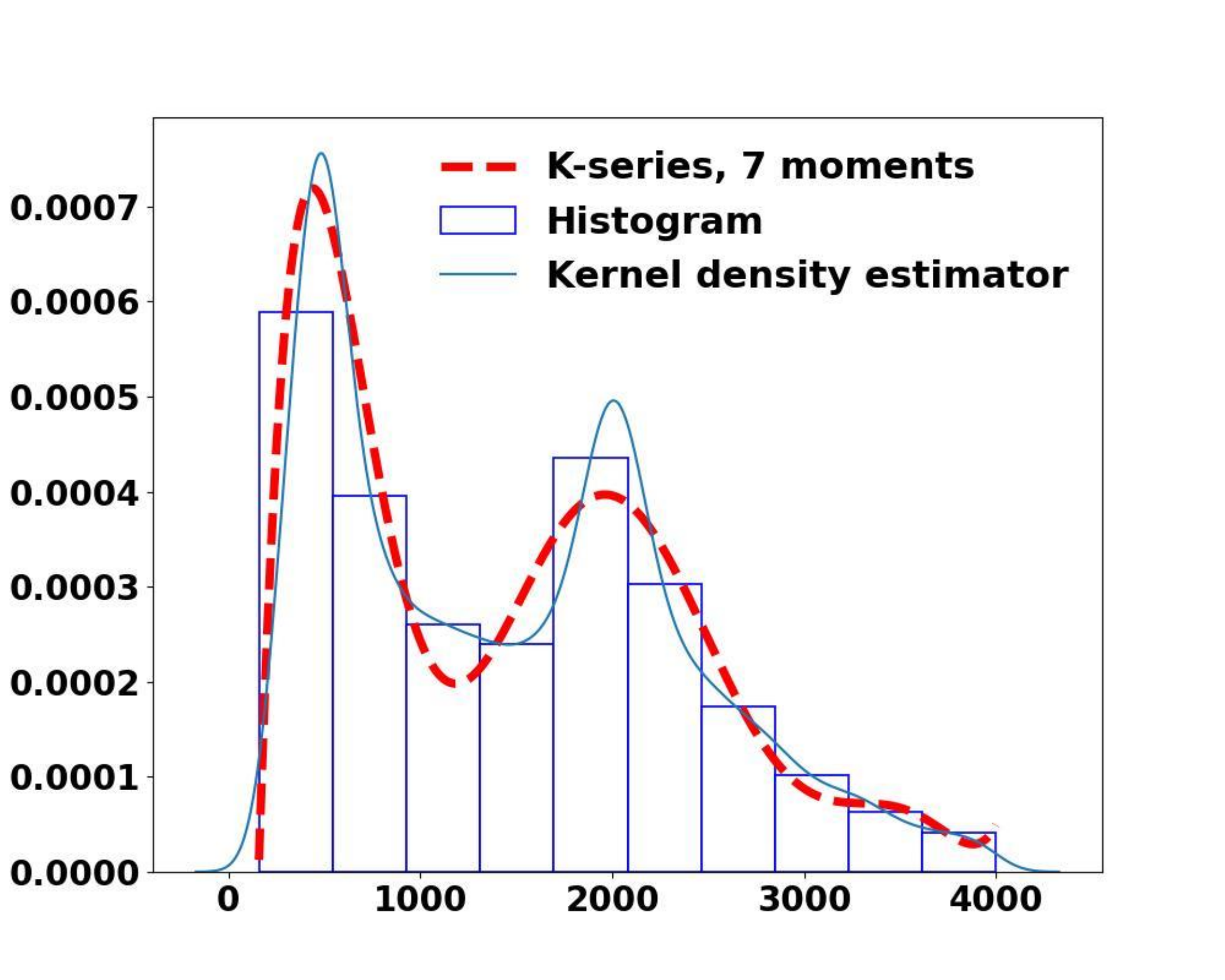}
 \caption{K-series estimate of pdf of household \\electricity} 
  \label{fig:real_data}
\end{wrapfigure}

In Fig. \ref{fig:otherexamples}, Panel B encodes the \textit{Taylor rule}, a model for monetary policy \cite{Taylor1993,Kofnovetal2022},  
D the \textit{rimless wheel walker} \cite{SteinhardtRuss2012}, and E the \textit{Vasicek} \cite{Vasicek1977,Karimietal2022} model. The \textit{Taylor rule} (B), \textit{rimless wheel} (D) and \textit{Vasicek} model (E) generate a single random variable at each iteration. We plot the histograms from sampling the probabilistic loop programs (blue) and the overlaid pdf K-series estimates of these models in  Fig. \ref{fig:taylor_results}. 
The \textit{2D robotic arm} model \cite{Bouissouetal2016} in panel C of Fig. \ref{fig:otherexamples} 
generates a bivariate random variable. 
We plot the marginal  K-series pdf estimates in the right panels, the joint pdf approximation in the bottom left panel, and the comparison of the true (blue bars) with our estimate (red bars) over a  2D parallelogram grid in the top left panel of Fig.~\ref{fig:robotic_arm_results}.  The moments of the true distribution were computed with the method in \cite{Kofnovetal2022} for the Taylor rule, and in \cite{BartocciKS19} for the rimless wheel, Vasicek and 2D robotic arm models.  
We used the following reference pdfs: truncated normal on $(-30,30)$ for the Taylor rule, truncated normal on $(0,10)$ for rimless wheel, normal distribution for the Vasicek model, and truncated normal on $(260,280) \times (525,540)$  for the 2D robotic arm model. 
For all univariate and bivariate models, our K-series estimator exhibits excellent estimation accuracy. 
The pdfs of the random variables in 1D and 2D random walks are estimated in Appendix~\ref{app:examples}.

\begin{figure}[h!]
\centering
\includegraphics[width=15.4cm, height=9cm]
{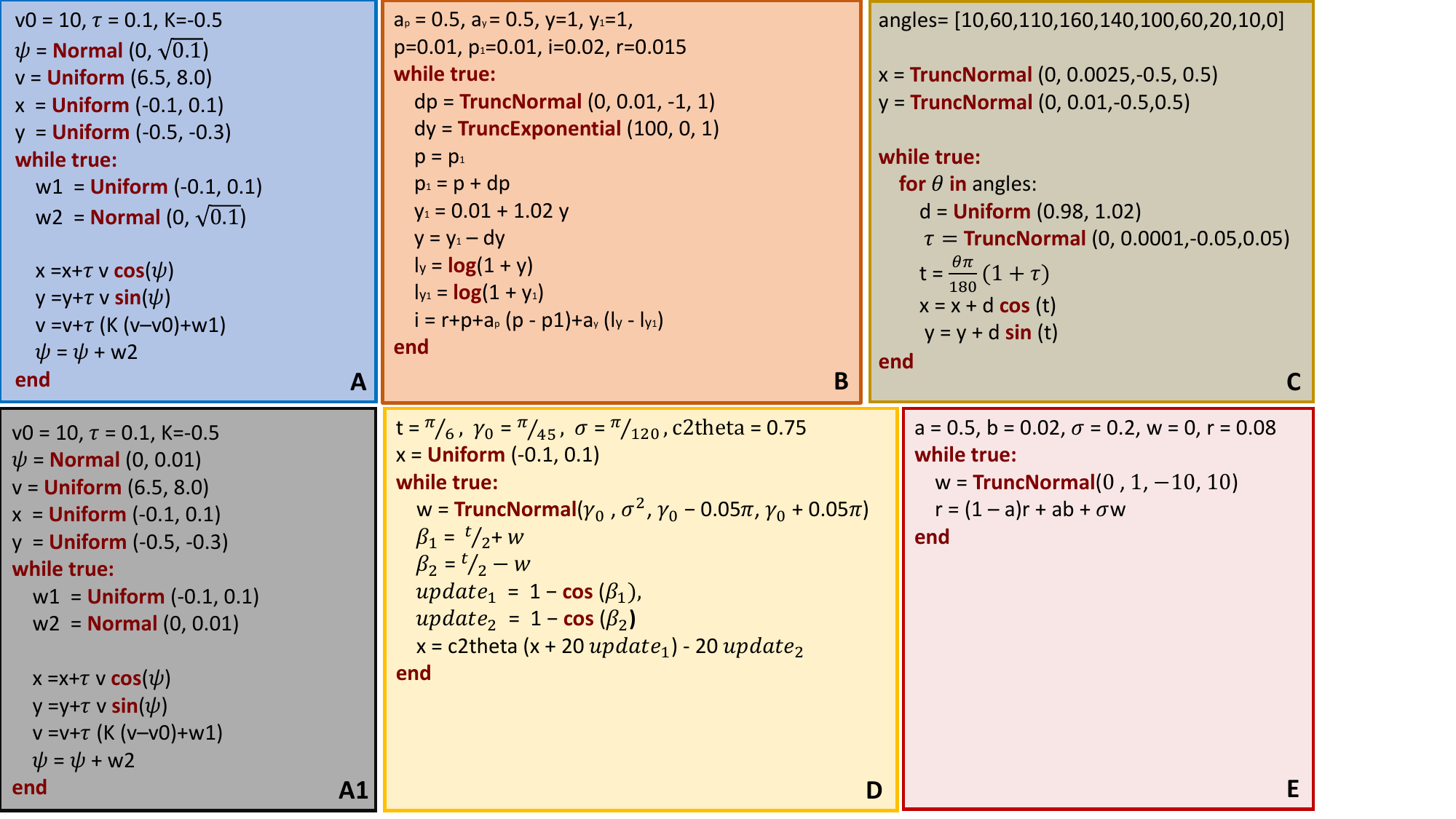}
  \caption{Probabilistic loops: (A) Turning vehicle model~\cite{Srirametal2020,Kofnovetal2022}, (A1) Small variance Turning vehicle model ~\cite{Srirametal2020,Kofnovetal2022}, (B) Taylor rule~\cite{Taylor1993,Kofnovetal2022}, (C) 2D Robotic Arm~\cite{Bouissouetal2016}, (D) Rimless Wheel Walker~\cite{SteinhardtT12}, (E) Vasicek model (truncated version)~\cite{Vasicek1977,Karimietal2022}.} 
  \label{fig:otherexamples}
\end{figure}

\begin{figure}[h!]
\centering
\includegraphics[width=0.90\textwidth, height=0.40\textwidth]{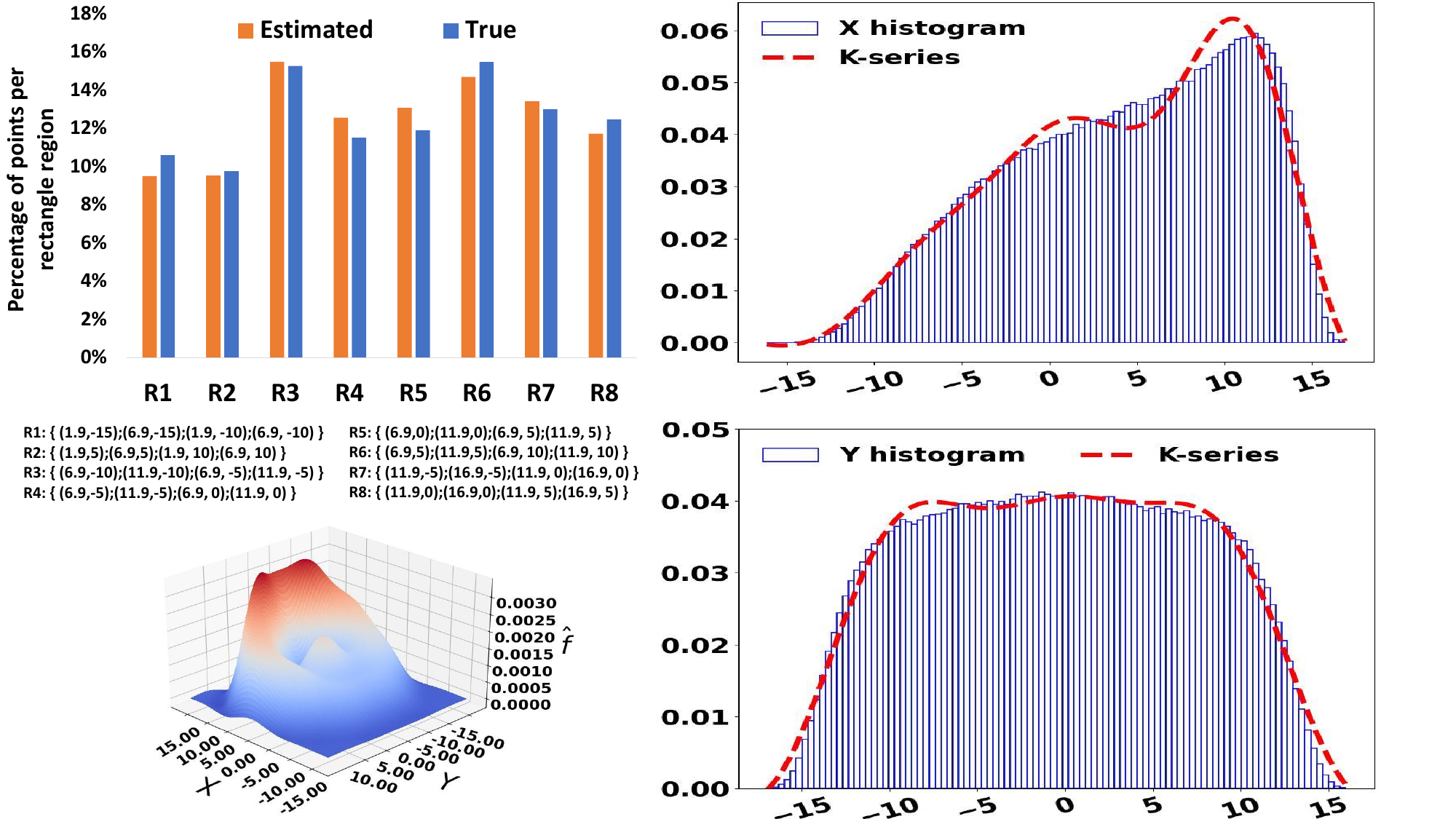}
  \caption{Turning vehicle model Fig.~\ref{fig:otherexamples} (A): K-series estimates of the marginal pdfs of $X$ and $Y$ (right upper and lower panels), the joint (lower left panel) and comparison bar plot (upper left panel) at iteration $t = 20$.} 
  \label{fig:turning_vehicle_results}
\end{figure}

\begin{figure}[h!]
\centering
\includegraphics[width=0.95\textwidth, height=0.50\textwidth]{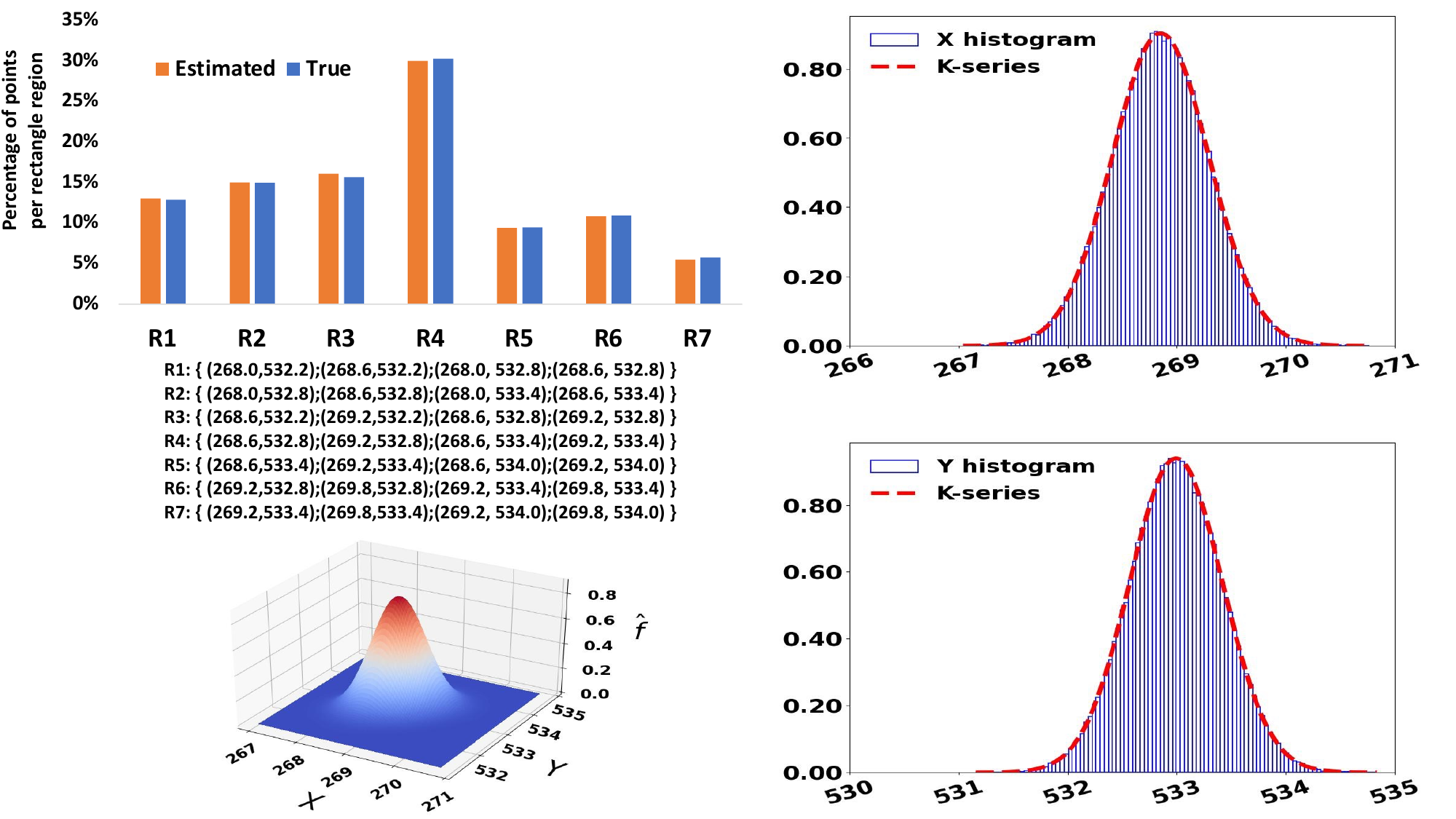}
  \caption{Robotic arm model Fig.~\ref{fig:otherexamples} (C): K-series estimates of the marginal pdfs of $X$ and $Y$ (right upper and lower panels), the joint (lower left panel) and comparison bar plot (upper left panel) at iteration $t = 100$.} 
  \label{fig:robotic_arm_results}
\end{figure}

\begin{figure}[h!]
\centering
\includegraphics[width=0.90\textwidth, height=0.35\textwidth]{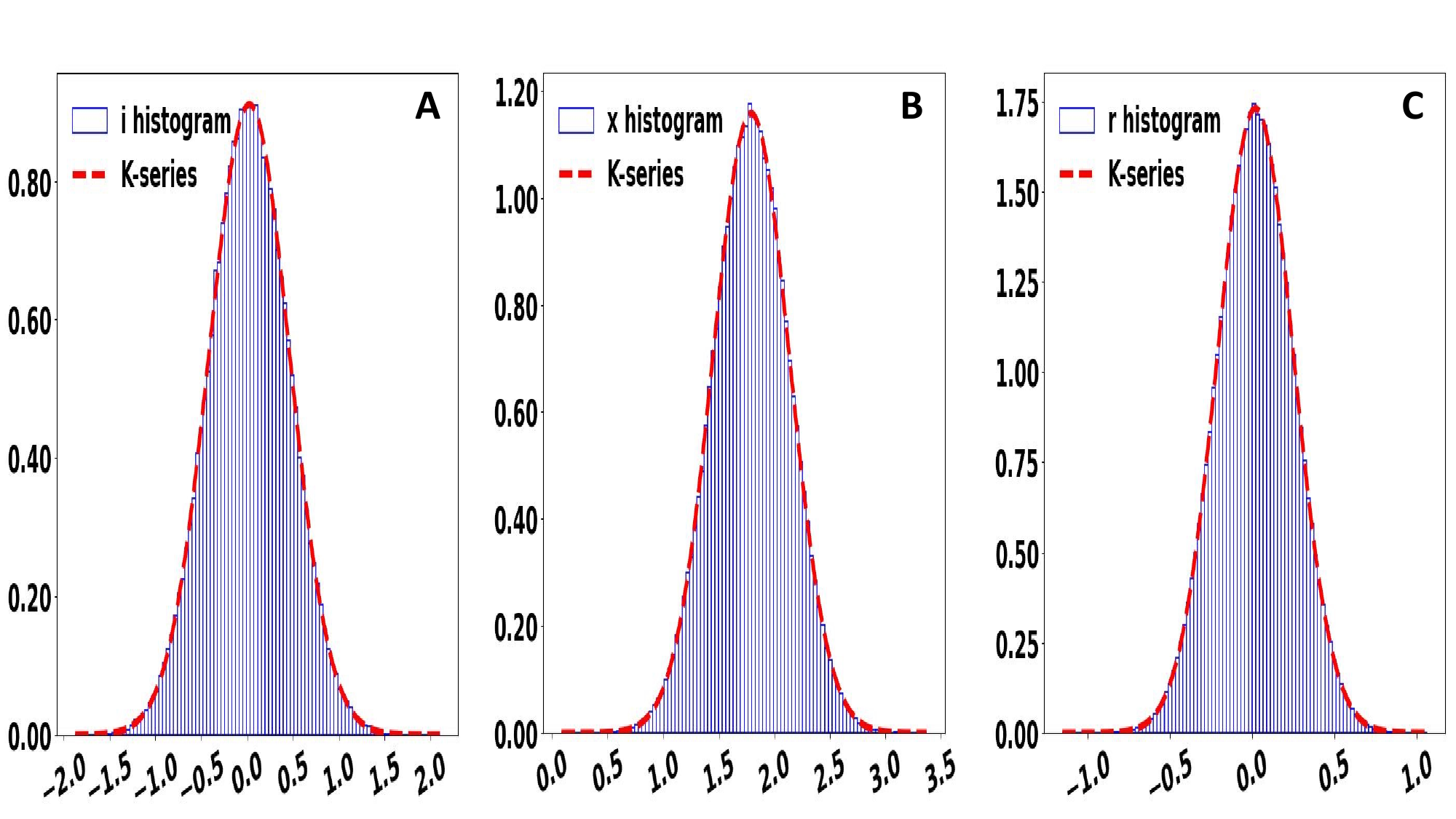}
\caption{K-series estimates of pdf for variable \textbf{A)} $i$ at iterations $t = 20$ in Fig.~\ref{fig:otherexamples} (B): Taylor rule model, \textbf{B)} $x$ at iteration $t = 2000$ in Fig.~\ref{fig:otherexamples} (D): rimless wheel model, \textbf{C)} $r$ at iteration $t = 100$  in Fig.~\ref{fig:otherexamples} (E): Vasicek model.
}
  \label{fig:taylor_results}
\end{figure}

In a real data application, we use K-series to estimate the density of "household electricity use with a ten-minute resolution for a detached house over one year" \cite{Munkhammar_etal_2017}. The data were analyzed by \cite{Munkhammar_etal_2017}, who estimated the unknown pdf by using sample-based estimates for the true unknown moments of the target distribution. Histograms of the data indicate that the pdf is bimodal. In real data examples, the true moments are unknown, so we also use the sample-based moment estimates to compute our K-series estimate, which is drawn in Fig.~\ref{fig:real_data}. We juxtapose our sample-moment-informed estimate with a nonparametric kernel density estimate, a standard data-driven approach for density estimation, for visual comparison.  The K-series estimate fits the data better, especially at both endpoints of the support, than \cite{Munkhammar_etal_2017}'s MM estimate, which can be viewed at the \href{https://journals.plos.org/plosone/article/figure?id=10.1371/journal.pone.0174573.g004}{PLOS One site}.

Regarding the time efficiency of K-series vis-\`a-vis other methods, Gram-Charlier is a special case of K-series for a normal reference distribution (Proposition \ref{prop:GC}), and the Method of Moments \cite{Munkhammar_etal_2017} is a special case of K-series for a uniform reference distribution (Theorem \ref{thm:MM}).  As such, the computational time for their implementation is the same as for K-series. Kernel density estimation (KDE) is not based on moments but requires a large number of samples from the population at hand. That is, the probabilistic program would have to be run many times to compute the kernel density estimator over its realized range of values to achieve comparative accuracy if even possible.  
Theoretically, K-series cannot be beaten in accuracy when true moments are available.

As an example, in Fig.~\ref{fig:KDE_example1} we plot the true pdf of a mixture of an equal-weighted mixture of four beta distributions with parameters (1.3, 5), (5, 1.3), (6, 7) and (7, 6), respectively, in green. The K-series estimator is the red dashed curve and the KDE estimate, based on the Gaussian kernel, is the blue curve. We sampled 10000 observations from the true pdf and plotted their histogram in gray. The time to produce the KDE estimate is $\{0.00644s + 0.0138s\}$ (sample and compute pdf, resp.). The time to compute the K-series estimate is longer, $\{0.73s + 0.662s + 4.58s\}$ (compute moments, construct a system of orthogonal polynomials and compute K-series, respectively). But Fig.~\ref{fig:KDE_example1} reveals that K-series tracks the true pdf much more accurately than the KDE, which is also subject to boundary effects, a well-known problem in nonparametric estimation.

In Fig.~\ref{fig:KDE_example2}, we visually compare the cdf estimates of the two approaches using 50000 samples. Again, the K-series cdf is closer to the true cdf, especially at the endpoints.  Also, the Kolmogorov-Smirnov distance between the K-series and the true cdf is $0.0012478$, much smaller than $0.0226002$, the value of the Kolmogorov-Smirnov distance between the true cdf and the cdf of the KDE estimate.  
As an aside comment, we note that we used a grid of 1000 points to compute all the pdfs for all other examples in the paper. Here, we had to use a much larger number of points to receive a sample of reliable size for KDE.

\begin{figure}
     \centering
     \begin{minipage}{.45\textwidth}
        \centering
\includegraphics[scale=0.28]{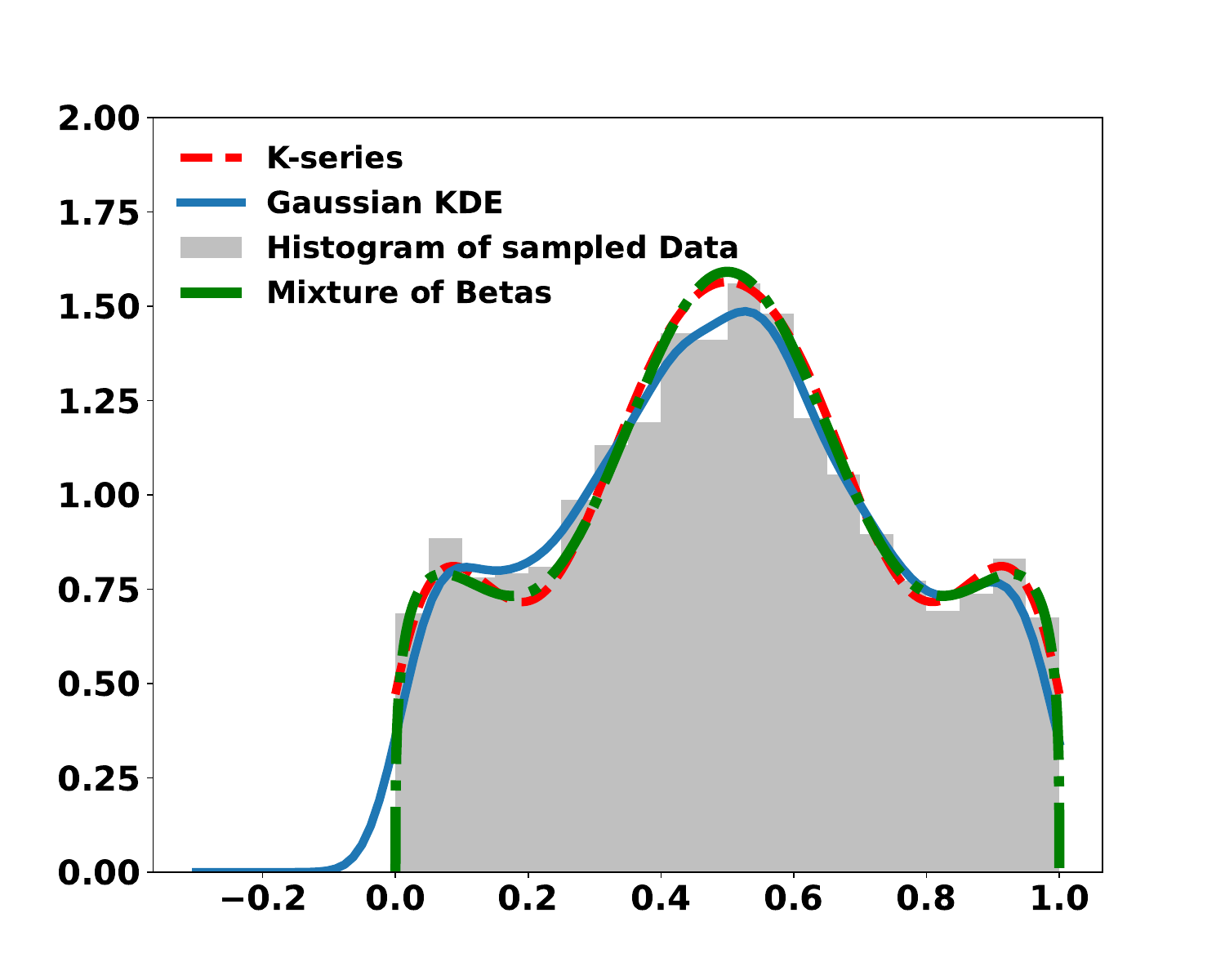}
  \caption{Comparison of K-series and KDE with Gaussian kernel performance in estimating a mixture of four Beta pdfs. We used 9 moments for K-series.} 
  \label{fig:KDE_example1}
\label{fig:(a)}
     \end{minipage}
     \hfill 
     \begin{minipage}{.45\textwidth}
         \centering
         \includegraphics[scale=0.28]{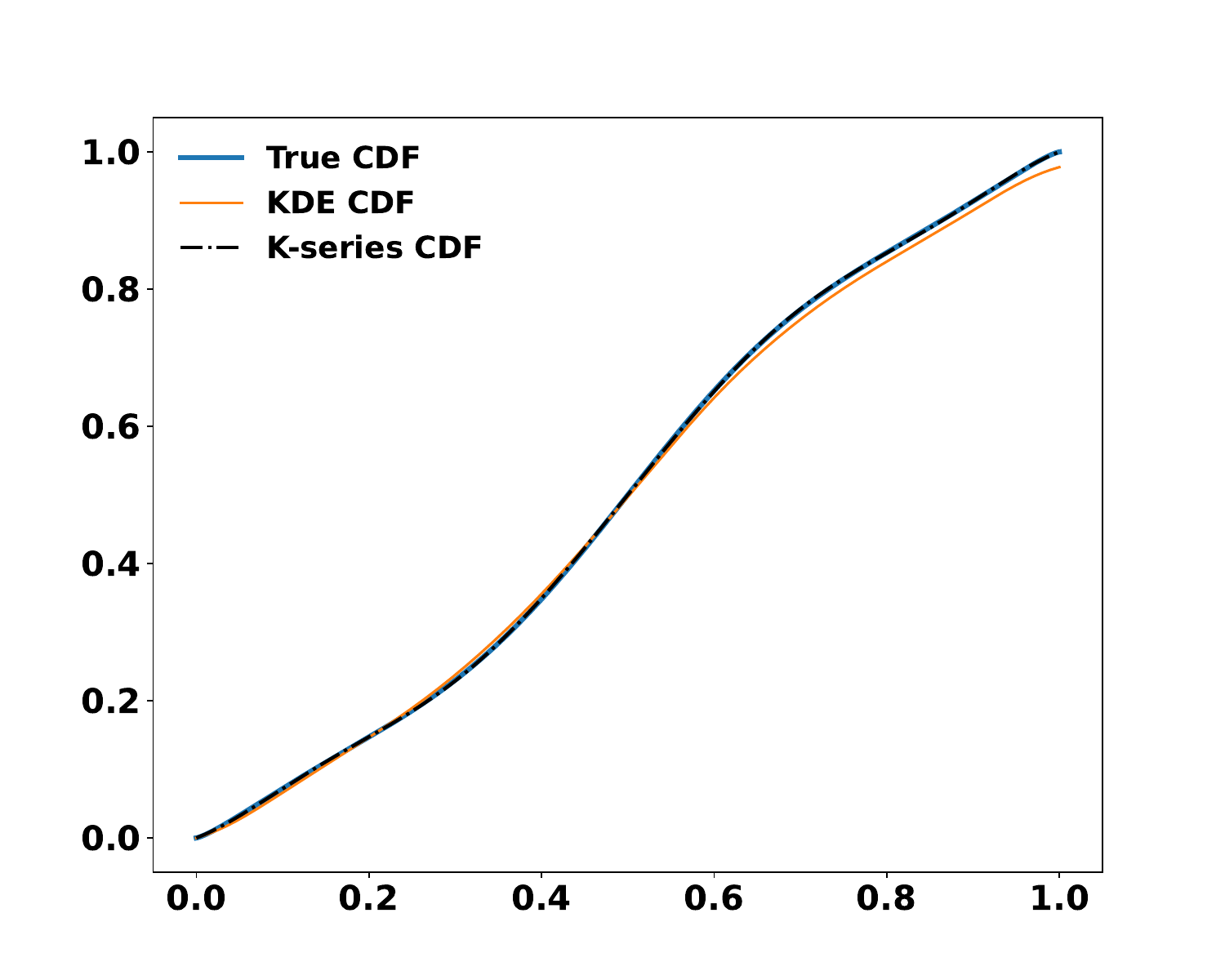}
  \caption{Comparison of K-series and KDE with Gaussian kernel performance in estimating a mixture of four Beta pdfs. We used a grid of 50.000 points.} 
  \label{fig:KDE_example2}
  \label{fig:(b)}
     \end{minipage}
\end{figure}

We explore the robustness of our method to violations of the assumption of continuity of random variables in  Fig. \ref{fig:discrt_vars}, where we estimate the distributions of random variables generated in Prob-solvable loops with discrete random components. 
Panel A in Fig. \ref{fig:discrt_vars} encodes the \textit{Stuttering P model} in  \cite{BartocciKS19}  and panel B the piece-wise deterministic process, or \textit{PDP model},  modeling gene circuits that can be used to estimate the bivariate distribution of protein $x$ and the mRNA levels $y$ in a gene~\cite{Innocentinietal2018}. 

For the \textit{Stuttering P} model, we used a truncated normal distribution on $(0, 50)$ with true mean and variance as the reference pdf. For the \textit{PDP} model, we used the  truncated normal on $(100, 1800)$  for $X$ and uniform on $(8, 80)$ for  $Y$ as reference pdfs to estimate marginal pdfs of $X$ and $Y$ and joint pdf $(X, Y)$. The parameters of the truncated normal distribution are the exact mean and variance of the marginal pdf of the corresponding variable computed using  \cite{BartocciKS19}.  


\begin{figure}[h!]
\centering
\includegraphics[width=0.95\textwidth, height=0.46\textwidth]{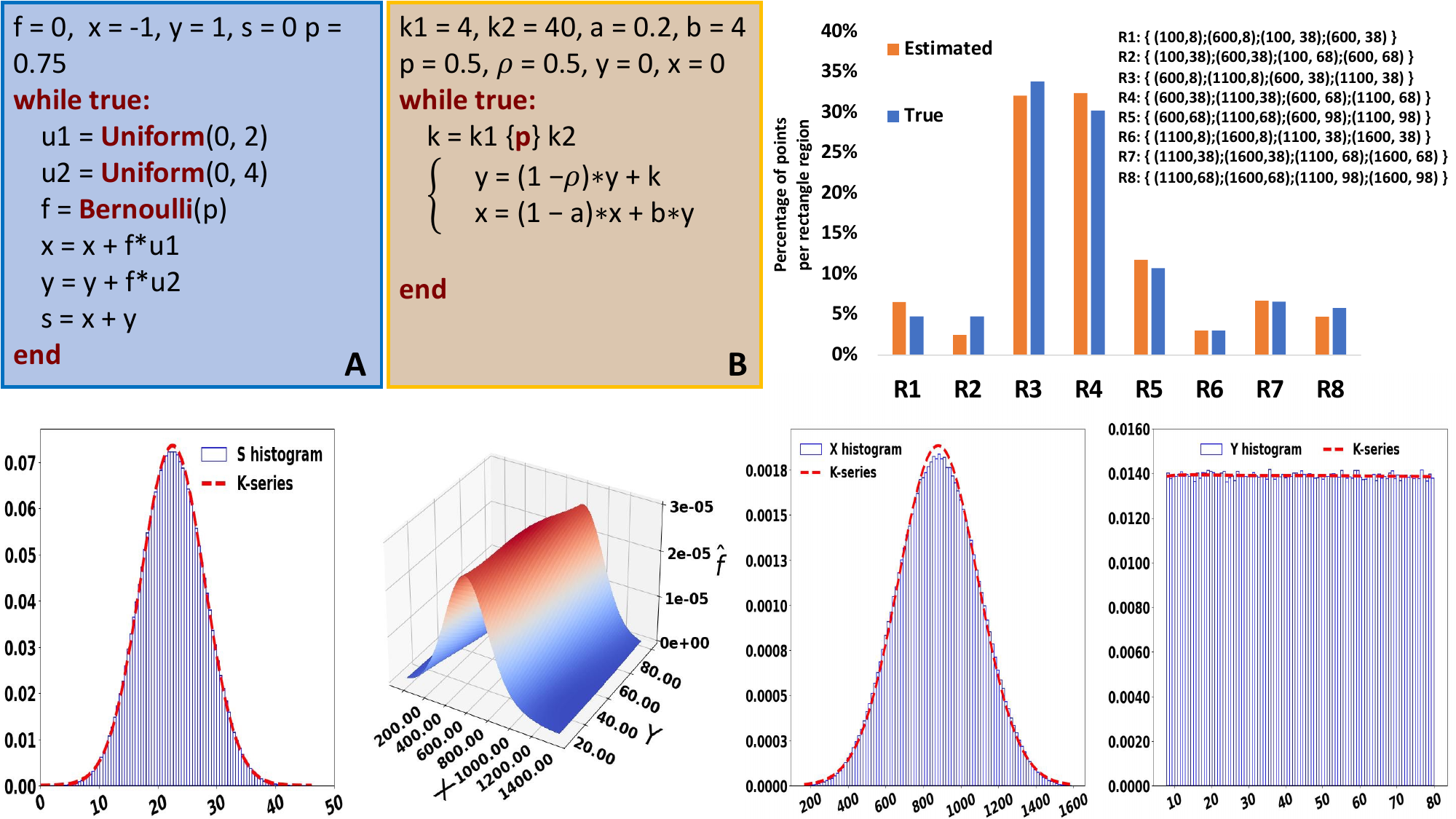}
  \caption{K-series estimates of the pdf of variable $S$ in Stuttering P model~\cite{BartocciKS19} (A) at iteration $t = 10$, marginal pdfs for variables $X$, $Y$ and joint distribution for variables ($X$, $Y$) in PDP model~\cite{Innocentinietal2018} (B) at the iterations $t = 100$.}
  \label{fig:discrt_vars}
\end{figure}

\section{Conclusion}\label{conclusion}

 K-series is a general distribution recovery method that approximates the density function as a series in terms of a finite number of its moments. It targets the unknown density via the choice of the reference probability density function and includes existing series-based density estimation, such as Gram-Charlier, as special cases. The K-series estimator converges to the true pdf in $L_{1}$, satisfies the moment matching principle, and is fast to compute. The method is complemented by an estimation algorithm of the minimal support of the target distribution. 

K-series requires the target pdf  have bounded support. This is not a serious limitation since, in practice, as in nature, observable values occur with effectively nonzero probability within an interval, and values outside a certain range are never realized. 
The choice of the reference is based on subject-matter knowledge, if available. We study the effect of the reference pdf on estimation in Appendix \ref{sec:refdstn_app}. The uniform reference distribution results in accurate estimates provided its support is close to the support of the true pdf. Both truncated and regular normal reference pdfs lead to accurate K-series estimates the closer the target pdf is to a normal. Overall, the truncated normal distribution typically results in better estimation. 

Characterizing the distribution of random quantities generated in probabilistic programming languages (PPLs)~\cite{barthe2020} 
is essential: Distributions are the building blocks of inference. PPLs codify probabilistic models and are used, for example, in computer security/privacy protocols~\cite{Dwork06}, distributed consensus algorithms~\cite{Herman90}, randomized algorithms~\cite{MotwaniRaghavan1995}, generative machine learning models~\cite{Ghahramani15} and scenario-based testing~\cite{FremontDGYSS19} of cyber-physical systems operating in uncertain environments. 

In future work, we will extend K-series to recover probability mass functions for discrete random variables.  We also aim to compute error bounds and explore the Fourier series representation of functions in conjunction with \cite{Jasouretal2021},  which obtains exact moments for sine and cosine assignments, to reduce the estimation error for fixed loop iterations. We will also develop a tool to automate the entire procedure in Algorithm \ref{alg:Univar_K_series}.

\vspace{0.1in}

\begin{table}[t!]
\centering
\begin{threeparttable}
\ra{0.5}
\begin{tabular}{@{}lcccrr@{}}\toprule
Model &  $Var$& \phantom{abc} & $|M|$ & \showclock{0}{0} Orthogonalization  &  \showclock{0}{0} K-series \\
 &   &  \phantom{abc} &   &  Runtime (in seconds)&   Runtime (in seconds) 
\\ \midrule \\
Differential-Drive Robot\\
  & $X$ 
  && 6
 & 
       0.67484    
  & 
       0.15971  \\
  & 
  $Y$ 
  && 6 
  & 0.57628 
  &  0.15921 
 \\
  &
  $(X,Y)$
  && 48 
  & 1.23404
  & 0.18318
\\
  \midrule \\
  PDP\\

  & $X$ && 2 & 0.03708 & 0.13438  \\
  & $Y$ && 6 & 0.24701 & 0.07646  \\ 
  & $(X, Y)$ && 8 & 0.03880 & 0.27987 \\
  \midrule \\
Turning vehicle\\

& $X$ && 8 & 0.46561 & 0.17895
\\ 
& $Y$ && 8 & 0.48054 & 0.17901
\\ 
& $(X, Y)$ && 80 & 0.84230 & 0.99251\\
  \midrule \\
Turning vehicle \\
 (small variance) \\

& $X$ && 8 & 0.68676  & 0.17829 
\\  
& $Y$ && 8 & 0.62172 & 0.17739
\\ 
& $(X, Y)$ && 80 & 1.19591 &  0.98794 \\
  \midrule \\
Taylor rule model\\
  & i   && 6 & 2.66375  & 0.16011\\
  \midrule \\
2D Robotic Arm \\
  & $X$ && 2 & 0.15185 & 0.13439 
  \\ 
  & $Y$ && 2 & 0.13663 & 0.13528
  \\ 
  & $(X, Y)$ && 8 & 0.28913 & 0.36801\\
  \midrule \\
Rimless Wheel Walker\\
  & $X$ && 2 & 0.10627 & 0.10915 \\
  \hline\\
    Vasicek model\\
  & $r$ && 2 & 0.16654 & 0.09937 \\
  \hline\\
    1D Random Walk\tnote{1}\\
  & $X$ && 2 & 0.09753 & 0.15714  \\
  \midrule \\
  2D Random Walk\tnote{1}\\

   & $X$ && 2 & 0.13076 & 0.15916 \\ 
  & $Y$ && 2 &  0.13033 & 0.15678 
  \\
  & $(X, Y)$ && 8 & 0.25956 & 0.40327  \\
  \midrule \\
    Stuttering P\\
  & $S$ && 2 & 0.06936 & 0.15530 \\
 \bottomrule
\end{tabular}
{\small
\begin{tablenotes}
     \item[1] See Appendix~\ref{app:examples}
   \end{tablenotes}}
\caption{Runtimes of orthogonalization procedure and K-series estimation for the benchmarks in Sec. \ref{experiments}. \\ $|M|$ denotes number of used moments and $Var$ the variable(s) whose density is estimated.}  
\label{tab:Experiments_time}
\end{threeparttable}
\end{table}


\begin{acks}
We would like to thank the associate editor and three reviewers for their helpful feedback and suggestions, which improved our work.

The research in this paper has been funded by the Vienna Science and Technology Fund (WWTF) [10.47379/ICT19018], the TU Wien Doctoral College (SecInt) and the FWF research
project P 30690-N35. 
\end{acks}

\newpage
\bibliographystyle{ACM-Reference-Format}
\bibliography{refs}


\begin{thebibliography}{49}


\ifx \showCODEN    \undefined \def \showCODEN     #1{\unskip}     \fi
\ifx \showDOI      \undefined \def \showDOI       #1{#1}\fi
\ifx \showISBNx    \undefined \def \showISBNx     #1{\unskip}     \fi
\ifx \showISBNxiii \undefined \def \showISBNxiii  #1{\unskip}     \fi
\ifx \showISSN     \undefined \def \showISSN      #1{\unskip}     \fi
\ifx \showLCCN     \undefined \def \showLCCN      #1{\unskip}     \fi
\ifx \shownote     \undefined \def \shownote      #1{#1}          \fi
\ifx \showarticletitle \undefined \def \showarticletitle #1{#1}   \fi
\ifx \showURL      \undefined \def \showURL       {\relax}        \fi
\providecommand\bibfield[2]{#2}
\providecommand\bibinfo[2]{#2}
\providecommand\natexlab[1]{#1}
\providecommand\showeprint[2][]{arXiv:#2}

\bibitem[Barthe et~al\mbox{.}(2020)]%
        {barthe2020}
\bibfield{author}{\bibinfo{person}{Gilles Barthe}, \bibinfo{person}{Joost-Pieter Katoen}, {and} \bibinfo{person}{Alexandra Silva}.} \bibinfo{year}{2020}\natexlab{}.
\newblock \bibinfo{booktitle}{\emph{Foundations of Probabilistic Programming}}.
\newblock \bibinfo{publisher}{Cambridge University Press}.
\newblock


\bibitem[Bartocci et~al\mbox{.}(2019a)]%
        {BartocciKS19}
\bibfield{author}{\bibinfo{person}{Ezio Bartocci}, \bibinfo{person}{Laura Kov{\'{a}}cs}, {and} \bibinfo{person}{Miroslav Stankovic}.} \bibinfo{year}{2019}\natexlab{a}.
\newblock \showarticletitle{Automatic Generation of Moment-Based Invariants for Prob-Solvable Loops}. In \bibinfo{booktitle}{\emph{Proc. of {ATVA} 2019: the 17th International Symposium on Automated Technology for Verification and Analysis}} \emph{(\bibinfo{series}{LNCS}, Vol.~\bibinfo{volume}{11781})}. \bibinfo{publisher}{Springer}.
\newblock
\urldef\tempurl%
\url{https://doi.org/10.1007/978-3-030-31784-3\_15}
\showDOI{\tempurl}


\bibitem[Bartocci et~al\mbox{.}(2019b)]%
        {Bartoccietal2019}
\bibfield{author}{\bibinfo{person}{Ezio Bartocci}, \bibinfo{person}{Laura Kov{\'{a}}cs}, {and} \bibinfo{person}{Miroslav Stankovic}.} \bibinfo{year}{2019}\natexlab{b}.
\newblock \showarticletitle{Automatic Generation of Moment-Based Invariants for Prob-Solvable Loops}.
\newblock \bibinfo{journal}{\emph{CoRR}}  \bibinfo{volume}{abs/1905.02835} (\bibinfo{year}{2019}).
\newblock
\showeprint[arxiv]{1905.02835}
\urldef\tempurl%
\url{http://arxiv.org/abs/1905.02835}
\showURL{%
\tempurl}


\bibitem[Billingsley(2012)]%
        {Billingsley2012}
\bibfield{author}{\bibinfo{person}{Patrick Billingsley}.} \bibinfo{year}{2012}\natexlab{}.
\newblock \bibinfo{booktitle}{\emph{Probability and Measure} (\bibinfo{edition}{anniversary edition} ed.)}.
\newblock \bibinfo{publisher}{Wiley}.
\newblock


\bibitem[Biswas and Bhattacharya(2010)]%
        {Biswas2010}
\bibfield{author}{\bibinfo{person}{Parthapratim Biswas} {and} \bibinfo{person}{Arun~K Bhattacharya}.} \bibinfo{year}{2010}\natexlab{}.
\newblock \showarticletitle{{Function Reconstruction as a Classical Moment Problem: A Maximum Entropy Approach}}.
\newblock \bibinfo{journal}{\emph{Journal of Physics A: Mathematical and Theoretical}} \bibinfo{volume}{43}, \bibinfo{number}{405003} (\bibinfo{year}{2010}), \bibinfo{pages}{1--19}.
\newblock


\bibitem[Bouissou et~al\mbox{.}(2016)]%
        {Bouissouetal2016}
\bibfield{author}{\bibinfo{person}{Olivier Bouissou}, \bibinfo{person}{Eric Goubault}, \bibinfo{person}{Sylvie Putot}, \bibinfo{person}{Aleksandar Chakarov}, {and} \bibinfo{person}{Sriram Sankaranarayanan}.} \bibinfo{year}{2016}\natexlab{}.
\newblock \showarticletitle{Uncertainty propagation using probabilistic affine forms and concentration of measure inequalities}. In \bibinfo{booktitle}{\emph{International Conference on Tools and Algorithms for the Construction and Analysis of Systems}}. Springer, \bibinfo{pages}{225--243}.
\newblock


\bibitem[Casella and Berger(2001)]%
        {CasellaBerger2001}
\bibfield{author}{\bibinfo{person}{George Casella} {and} \bibinfo{person}{Roger~L. Berger}.} \bibinfo{year}{2001}\natexlab{}.
\newblock \bibinfo{booktitle}{\emph{Statistical Inference} (\bibinfo{edition}{2} ed.)}.
\newblock \bibinfo{publisher}{Cengage Learning}.
\newblock


\bibitem[Chihara(1978)]%
        {Chihara1978}
\bibfield{author}{\bibinfo{person}{Theodore~S. Chihara}.} \bibinfo{year}{1978}\natexlab{}.
\newblock \bibinfo{booktitle}{\emph{An Introduction to Orthogonal Polynomials}}.
\newblock \bibinfo{publisher}{Gordon and Breach, Science Publishers}.
\newblock


\bibitem[Cram\'{e}r(1957)]%
        {Cramer1957}
\bibfield{author}{\bibinfo{person}{Harald Cram\'{e}r}.} \bibinfo{year}{1957}\natexlab{}.
\newblock \bibinfo{booktitle}{\emph{Mathematical Methods of Statistics}}.
\newblock \bibinfo{publisher}{Princeton Univ. Press}, \bibinfo{address}{Princeton, NJ}.
\newblock
\urldef\tempurl%
\url{http://www.jstor.org/stable/j.ctt1bpm9r4}
\showURL{%
\tempurl}


\bibitem[Dwork(2006)]%
        {Dwork06}
\bibfield{author}{\bibinfo{person}{Cynthia Dwork}.} \bibinfo{year}{2006}\natexlab{}.
\newblock \showarticletitle{Differential Privacy}. In \bibinfo{booktitle}{\emph{Proc. of {ICALP} 2006: the 33rd International Colloquium on Automata, Languages and Programming}} \emph{(\bibinfo{series}{LNCS}, Vol.~\bibinfo{volume}{4052})}. \bibinfo{publisher}{Springer}, \bibinfo{pages}{1--12}.
\newblock
\urldef\tempurl%
\url{https://doi.org/10.1007/11787006\_1}
\showDOI{\tempurl}


\bibitem[{Ernst, Oliver G.} et~al\mbox{.}(2012)]%
        {Ernstetal2012}
\bibfield{author}{\bibinfo{person}{{Ernst, Oliver G.}}, \bibinfo{person}{{Mugler, Antje}}, \bibinfo{person}{{Starkloff, Hans-J\"org}}, {and} \bibinfo{person}{{Ullmann, Elisabeth}}.} \bibinfo{year}{2012}\natexlab{}.
\newblock \showarticletitle{On the convergence of generalized polynomial chaos expansions}.
\newblock \bibinfo{journal}{\emph{ESAIM: M2AN}} \bibinfo{volume}{46}, \bibinfo{number}{2} (\bibinfo{year}{2012}), \bibinfo{pages}{317--339}.
\newblock
\urldef\tempurl%
\url{https://doi.org/10.1051/m2an/2011045}
\showDOI{\tempurl}


\bibitem[Filipović et~al\mbox{.}(2013)]%
        {Filipovicetal2013}
\bibfield{author}{\bibinfo{person}{Damir Filipović}, \bibinfo{person}{Eberhard Mayerhofer}, {and} \bibinfo{person}{Paul Schneider}.} \bibinfo{year}{2013}\natexlab{}.
\newblock \showarticletitle{Density approximations for multivariate affine jump-diffusion processes}.
\newblock \bibinfo{journal}{\emph{Journal of Econometrics}} \bibinfo{volume}{176}, \bibinfo{number}{2} (\bibinfo{year}{2013}), \bibinfo{pages}{93--111}.
\newblock
\showISSN{0304-4076}
\urldef\tempurl%
\url{https://doi.org/10.1016/j.jeconom.2012.12.003}
\showDOI{\tempurl}


\bibitem[Fremont et~al\mbox{.}(2019)]%
        {FremontDGYSS19}
\bibfield{author}{\bibinfo{person}{Daniel~J. Fremont}, \bibinfo{person}{Tommaso Dreossi}, \bibinfo{person}{Shromona Ghosh}, \bibinfo{person}{Xiangyu Yue}, \bibinfo{person}{Alberto~L. Sangiovanni{-}Vincentelli}, {and} \bibinfo{person}{Sanjit~A. Seshia}.} \bibinfo{year}{2019}\natexlab{}.
\newblock \showarticletitle{Scenic: a language for scenario specification and scene generation}. In \bibinfo{booktitle}{\emph{Proc. of PLDI 2019: the 40th {ACM} {SIGPLAN} Conference on Programming Language Design and Implementation}}. \bibinfo{publisher}{{ACM}}, \bibinfo{pages}{63--78}.
\newblock
\urldef\tempurl%
\url{https://doi.org/10.1145/3314221}
\showDOI{\tempurl}


\bibitem[Gehr et~al\mbox{.}(2016)]%
        {Gehretal2016}
\bibfield{author}{\bibinfo{person}{Timon Gehr}, \bibinfo{person}{Sasa Misailovic}, {and} \bibinfo{person}{Martin~T. Vechev}.} \bibinfo{year}{2016}\natexlab{}.
\newblock \showarticletitle{{PSI:} Exact Symbolic Inference for Probabilistic Programs}. In \bibinfo{booktitle}{\emph{Proc. of {CAV} 2016: the 28th International Conference on Computer Aided Verification}} \emph{(\bibinfo{series}{LNCS}, Vol.~\bibinfo{volume}{9779})}. \bibinfo{publisher}{Springer}, \bibinfo{pages}{62--83}.
\newblock
\urldef\tempurl%
\url{https://doi.org/10.1007/978-3-319-41528-4\_4}
\showDOI{\tempurl}


\bibitem[Gehr et~al\mbox{.}(2020)]%
        {Gehretal2020}
\bibfield{author}{\bibinfo{person}{Timon Gehr}, \bibinfo{person}{Samuel Steffen}, {and} \bibinfo{person}{Martin Vechev}.} \bibinfo{year}{2020}\natexlab{}.
\newblock \showarticletitle{$\lambda${P}{S}{I}: {Exact} {Inference} for {Higher}-{Order} {Probabilistic} {Programs}}. In \bibinfo{booktitle}{\emph{Proceedings of the 41st {ACM} {SIGPLAN} {Conference} on {Programming} {Language} {Design} and {Implementation}}}. \bibinfo{publisher}{ACM}.
\newblock
\showISBNx{978-1-4503-7613-6}
\urldef\tempurl%
\url{https://doi.org/10.1145/3385412.3386006}
\showDOI{\tempurl}


\bibitem[Ghahramani(2015)]%
        {Ghahramani15}
\bibfield{author}{\bibinfo{person}{Zoubin Ghahramani}.} \bibinfo{year}{2015}\natexlab{}.
\newblock \showarticletitle{Probabilistic machine learning and artificial intelligence}.
\newblock \bibinfo{journal}{\emph{Nature}} \bibinfo{volume}{521}, \bibinfo{number}{7553} (\bibinfo{year}{2015}), \bibinfo{pages}{452--459}.
\newblock
\urldef\tempurl%
\url{https://doi.org/10.1038/nature14541}
\showDOI{\tempurl}


\bibitem[Hald(2000)]%
        {Hald2000}
\bibfield{author}{\bibinfo{person}{Anders Hald}.} \bibinfo{year}{2000}\natexlab{}.
\newblock \showarticletitle{{The Early History of the Cumulants and the Gram-Charlier Series}}.
\newblock \bibinfo{journal}{\emph{International Statistical Review}} \bibinfo{volume}{68}, \bibinfo{number}{2} (\bibinfo{year}{2000}), \bibinfo{pages}{137--153}.
\newblock


\bibitem[Herman(1990)]%
        {Herman90}
\bibfield{author}{\bibinfo{person}{Ted Herman}.} \bibinfo{year}{1990}\natexlab{}.
\newblock \showarticletitle{Probabilistic Self-Stabilization}.
\newblock \bibinfo{journal}{\emph{Inf. Process. Lett.}} \bibinfo{volume}{35}, \bibinfo{number}{2} (\bibinfo{year}{1990}), \bibinfo{pages}{63--67}.
\newblock
\urldef\tempurl%
\url{https://doi.org/10.1016/0020-0190(90)90107-9}
\showDOI{\tempurl}


\bibitem[Hollander et~al\mbox{.}(2013)]%
        {Hollanderetal2013}
\bibfield{author}{\bibinfo{person}{Myles Hollander}, \bibinfo{person}{Douglas~A. Wolfe}, {and} \bibinfo{person}{Eric Chicken}.} \bibinfo{year}{2013}\natexlab{}.
\newblock \bibinfo{booktitle}{\emph{Nonparametric Statistical Methods} (\bibinfo{edition}{3} ed.)}.
\newblock \bibinfo{publisher}{John Wiley {\&} Sons}, \bibinfo{address}{New York-London-Sydney}.
\newblock


\bibitem[Innocentini et~al\mbox{.}(2018)]%
        {Innocentinietal2018}
\bibfield{author}{\bibinfo{person}{Guilherme C.~P. Innocentini}, \bibinfo{person}{Arran Hodgkinson}, {and} \bibinfo{person}{Ovidiu Radulescu}.} \bibinfo{year}{2018}\natexlab{}.
\newblock \showarticletitle{Time Dependent Stochastic mRNA and Protein Synthesis in Piecewise-Deterministic Models of Gene Networks}.
\newblock \bibinfo{journal}{\emph{Frontiers in Physics}}  \bibinfo{volume}{6} (\bibinfo{year}{2018}).
\newblock
\showISSN{2296-424X}
\urldef\tempurl%
\url{https://doi.org/10.3389/fphy.2018.00046}
\showDOI{\tempurl}


\bibitem[Jasour et~al\mbox{.}(2021)]%
        {Jasouretal2021}
\bibfield{author}{\bibinfo{person}{Ashkan Jasour}, \bibinfo{person}{Allen Wang}, {and} \bibinfo{person}{Brian~C. Williams}.} \bibinfo{year}{2021}\natexlab{}.
\newblock \showarticletitle{Moment-Based Exact Uncertainty Propagation Through Nonlinear Stochastic Autonomous Systems}.
\newblock \bibinfo{journal}{\emph{CoRR}}  \bibinfo{volume}{abs/2101.12490} (\bibinfo{year}{2021}).
\newblock
\showeprint[arXiv]{2101.12490}
\urldef\tempurl%
\url{https://arxiv.org/abs/2101.12490}
\showURL{%
\tempurl}


\bibitem[Karimi et~al\mbox{.}(2022)]%
        {Karimietal2022}
\bibfield{author}{\bibinfo{person}{Ahmad Karimi}, \bibinfo{person}{Marcel Moosbrugger}, \bibinfo{person}{Miroslav Stankovi{\v{c}}}, \bibinfo{person}{Laura Kov{\'a}cs}, \bibinfo{person}{Ezio Bartocci}, {and} \bibinfo{person}{Efstathia Bura}.} \bibinfo{year}{2022}\natexlab{}.
\newblock \showarticletitle{Distribution Estimation for Probabilistic Loops}. In \bibinfo{booktitle}{\emph{Quantitative Evaluation of Systems}}, \bibfield{editor}{\bibinfo{person}{Erika {\'A}brah{\'a}m} {and} \bibinfo{person}{Marco Paolieri}} (Eds.). \bibinfo{publisher}{Springer International Publishing}, \bibinfo{address}{Cham}, \bibinfo{pages}{26--42}.
\newblock
\showISBNx{978-3-031-16336-4}


\bibitem[Kendall and Stuart(1977)]%
        {KendallStuart1977}
\bibfield{author}{\bibinfo{person}{Maurice Kendall} {and} \bibinfo{person}{Alan Stuart}.} \bibinfo{year}{1977}\natexlab{}.
\newblock \bibinfo{booktitle}{\emph{{The Advanced Theory of Statistics. Volume 1: Distribution Theory}}}.
\newblock \bibinfo{publisher}{Macmillan}, \bibinfo{address}{New York, NY}.
\newblock


\bibitem[Klinkenberg et~al\mbox{.}(2023)]%
        {Klinkenbergetal2023}
\bibfield{author}{\bibinfo{person}{Lutz Klinkenberg}, \bibinfo{person}{Christian Blumenthal}, \bibinfo{person}{Mingshuai Chen}, {and} \bibinfo{person}{Joost-Pieter Katoen}.} \bibinfo{year}{2023}\natexlab{}.
\newblock \bibinfo{title}{Exact Bayesian Inference for Loopy Probabilistic Programs}.
\newblock
\newblock
\showeprint[arxiv]{2307.07314}~[cs.PL]


\bibitem[Kofnov et~al\mbox{.}(2022)]%
        {Kofnovetal2022}
\bibfield{author}{\bibinfo{person}{Andrey Kofnov}, \bibinfo{person}{Marcel Moosbrugger}, \bibinfo{person}{Miroslav Stankovi{\v{c}}}, \bibinfo{person}{Ezio Bartocci}, {and} \bibinfo{person}{Efstathia Bura}.} \bibinfo{year}{2022}\natexlab{}.
\newblock \showarticletitle{Moment-Based Invariants for Probabilistic Loops with Non-polynomial Assignments}. In \bibinfo{booktitle}{\emph{Quantitative Evaluation of Systems}}, \bibfield{editor}{\bibinfo{person}{Erika {\'A}brah{\'a}m} {and} \bibinfo{person}{Marco Paolieri}} (Eds.). \bibinfo{publisher}{Springer International Publishing}, \bibinfo{address}{Cham}, \bibinfo{pages}{3--25}.
\newblock
\showISBNx{978-3-031-16336-4}


\bibitem[Kofnov et~al\mbox{.}(2024)]%
        {Kofnovetal2024}
\bibfield{author}{\bibinfo{person}{Andrey Kofnov}, \bibinfo{person}{Marcel Moosbrugger}, \bibinfo{person}{Miroslav Stankovi\v{c}}, \bibinfo{person}{Ezio Bartocci}, {and} \bibinfo{person}{Efstathia Bura}.} \bibinfo{year}{2024}\natexlab{}.
\newblock \showarticletitle{Exact and Approximate Moment Derivation for Probabilistic Loops With Non-Polynomial Assignments}.
\newblock \bibinfo{journal}{\emph{ACM Trans. Model. Comput. Simul.}} (\bibinfo{date}{jan} \bibinfo{year}{2024}).
\newblock
\showISSN{1049-3301}
\urldef\tempurl%
\url{https://doi.org/10.1145/3641545}
\showDOI{\tempurl}


\bibitem[Kolmogorov and Fomin(1976)]%
        {KolmogorovFomin1976}
\bibfield{author}{\bibinfo{person}{A.N. Kolmogorov} {and} \bibinfo{person}{S.V. Fomin}.} \bibinfo{year}{1976}\natexlab{}.
\newblock \bibinfo{booktitle}{\emph{Elements of the Theory of Functions and Functional Analysis} (\bibinfo{edition}{4} ed.)}.
\newblock \bibinfo{publisher}{Nauka}, \bibinfo{address}{Moscow}.
\newblock


\bibitem[Kolossa(2006)]%
        {Kolossa2006}
\bibfield{author}{\bibinfo{person}{J.~E. Kolossa}.} \bibinfo{year}{2006}\natexlab{}.
\newblock \bibinfo{booktitle}{\emph{Series Approximation Methods in Statistics} (\bibinfo{edition}{3} ed.)}.
\newblock \bibinfo{publisher}{Springer}, \bibinfo{address}{New York, NY}.
\newblock


\bibitem[Kura et~al\mbox{.}(2019)]%
        {Kura19}
\bibfield{author}{\bibinfo{person}{Satoshi Kura}, \bibinfo{person}{Natsuki Urabe}, {and} \bibinfo{person}{Ichiro Hasuo}.} \bibinfo{year}{2019}\natexlab{}.
\newblock \showarticletitle{Tail Probabilities for Randomized Program Runtimes via Martingales for Higher Moments}. In \bibinfo{booktitle}{\emph{Tools and Algorithms for the Construction and Analysis of Systems}}, \bibfield{editor}{\bibinfo{person}{Tom{\'a}{\v{s}} Vojnar} {and} \bibinfo{person}{Lijun Zhang}} (Eds.). \bibinfo{publisher}{Springer International Publishing}, \bibinfo{address}{Cham}, \bibinfo{pages}{135--153}.
\newblock


\bibitem[Lebaz et~al\mbox{.}(2016)]%
        {Lebazetal2016}
\bibfield{author}{\bibinfo{person}{Noureddine Lebaz}, \bibinfo{person}{Arnaud Cockx}, \bibinfo{person}{Mathieu Sp{\'e}randio}, {and} \bibinfo{person}{J{\'e}r{\^o}me Morchain}.} \bibinfo{year}{2016}\natexlab{}.
\newblock \showarticletitle{{Reconstruction of a Distribution From a Finite Number of Its Moments: A Comparative Study in the Case of Depolymerization Process}}.
\newblock \bibinfo{journal}{\emph{Computers \& Chemical Engineering}}  \bibinfo{volume}{84} (\bibinfo{year}{2016}), \bibinfo{pages}{326--337}.
\newblock


\bibitem[Moosbrugger et~al\mbox{.}(2022)]%
        {Moosbruggeretal2022}
\bibfield{author}{\bibinfo{person}{Marcel Moosbrugger}, \bibinfo{person}{Miroslav Stankovi\v{c}}, \bibinfo{person}{Ezio Bartocci}, {and} \bibinfo{person}{Laura Kov\'{a}cs}.} \bibinfo{year}{2022}\natexlab{}.
\newblock \showarticletitle{This is the Moment for Probabilistic Loops}.
\newblock \bibinfo{journal}{\emph{Proc. {ACM} Program. Lang.}} \bibinfo{volume}{6}, \bibinfo{number}{OOPSLA2}, Article \bibinfo{articleno}{178} (\bibinfo{date}{oct} \bibinfo{year}{2022}), \bibinfo{numpages}{29}~pages.
\newblock
\urldef\tempurl%
\url{https://doi.org/10.1145/3563341}
\showDOI{\tempurl}


\bibitem[Munkhammar et~al\mbox{.}(2017)]%
        {Munkhammar_etal_2017}
\bibfield{author}{\bibinfo{person}{Joakim Munkhammar}, \bibinfo{person}{Lars Mattsson}, {and} \bibinfo{person}{Jesper Rydén}.} \bibinfo{year}{2017}\natexlab{}.
\newblock \showarticletitle{Polynomial probability distribution estimation using the method of moments}.
\newblock \bibinfo{journal}{\emph{PLOS ONE}} \bibinfo{volume}{12}, \bibinfo{number}{4} (\bibinfo{date}{04} \bibinfo{year}{2017}), \bibinfo{pages}{1--14}.
\newblock
\urldef\tempurl%
\url{https://doi.org/10.1371/journal.pone.0174573}
\showDOI{\tempurl}


\bibitem[Rahman(2018)]%
        {Rahman2018}
\bibfield{author}{\bibinfo{person}{Sharif Rahman}.} \bibinfo{year}{2018}\natexlab{}.
\newblock \showarticletitle{A polynomial chaos expansion in dependent random variables}.
\newblock \bibinfo{journal}{\emph{J. Math. Anal. Appl.}} \bibinfo{volume}{464}, \bibinfo{number}{1} (\bibinfo{year}{2018}), \bibinfo{pages}{749--775}.
\newblock
\showISSN{0022-247X}
\urldef\tempurl%
\url{https://doi.org/10.1016/j.jmaa.2018.04.032}
\showDOI{\tempurl}


\bibitem[Rajeev~Motwani(1995)]%
        {MotwaniRaghavan1995}
\bibfield{author}{\bibinfo{person}{Prabhakar~Raghavan Rajeev~Motwani}.} \bibinfo{year}{1995}\natexlab{}.
\newblock \bibinfo{booktitle}{\emph{Randomized Algorithms}}.
\newblock \bibinfo{publisher}{Cambridge University Press}.
\newblock
\showISBNx{3257227892}


\bibitem[Rizzo and Szekely(2022)]%
        {RizzoSzekely2022}
\bibfield{author}{\bibinfo{person}{Maria Rizzo} {and} \bibinfo{person}{Gabor Szekely}.} \bibinfo{year}{2022}\natexlab{}.
\newblock \bibinfo{booktitle}{\emph{energy: E-Statistics: Multivariate Inference via the Energy of Data}}.
\newblock
\urldef\tempurl%
\url{https://CRAN.R-project.org/package=energy}
\showURL{%
\tempurl}
\newblock
\shownote{R package version 1.7-11}.


\bibitem[Rudin(1976)]%
        {Rudin1976}
\bibfield{author}{\bibinfo{person}{Walter Rudin}.} \bibinfo{year}{1976}\natexlab{}.
\newblock \bibinfo{booktitle}{\emph{Principles of mathematical analysis} (\bibinfo{edition}{3d ed.} ed.)}.
\newblock \bibinfo{publisher}{McGraw-Hill New York}. x, 342 p. ; pages.
\newblock
\showISBNx{007054235}
\urldef\tempurl%
\url{http://www.loc.gov/catdir/toc/mh031/75017903.html}
\showURL{%
\tempurl}


\bibitem[Rudin(1986)]%
        {rudin1986}
\bibfield{author}{\bibinfo{person}{Walter Rudin}.} \bibinfo{year}{1986}\natexlab{}.
\newblock \bibinfo{booktitle}{\emph{Real and Complex Analysis}}.
\newblock \bibinfo{publisher}{McGraw-Hill Science/Engineering/Math}.
\newblock
\showISBNx{0070542341}


\bibitem[Sankaranarayanan et~al\mbox{.}(2020)]%
        {Srirametal2020}
\bibfield{author}{\bibinfo{person}{Sriram Sankaranarayanan}, \bibinfo{person}{Yi Chou}, \bibinfo{person}{Eric Goubault}, {and} \bibinfo{person}{Sylvie Putot}.} \bibinfo{year}{2020}\natexlab{}.
\newblock \showarticletitle{Reasoning about Uncertainties in Discrete-Time Dynamical Systems using Polynomial Forms.}. In \bibinfo{booktitle}{\emph{Advances in Neural Information Processing Systems}}, \bibfield{editor}{\bibinfo{person}{H.~Larochelle}, \bibinfo{person}{M.~Ranzato}, \bibinfo{person}{R.~Hadsell}, \bibinfo{person}{M.~F. Balcan}, {and} \bibinfo{person}{H.~Lin}} (Eds.), Vol.~\bibinfo{volume}{33}. \bibinfo{publisher}{Curran Associates, Inc.}, \bibinfo{pages}{17502--17513}.
\newblock
\urldef\tempurl%
\url{https://proceedings.neurips.cc/paper/2020/file/ca886eb9edb61a42256192745c72cd79-Paper.pdf}
\showURL{%
\tempurl}


\bibitem[Silverman(1998)]%
        {Silverman1998}
\bibfield{author}{\bibinfo{person}{B.W. Silverman}.} \bibinfo{year}{1998}\natexlab{}.
\newblock \bibinfo{booktitle}{\emph{Density Estimation for Statistics and Data Analysis} (\bibinfo{edition}{1} ed.)}.
\newblock \bibinfo{publisher}{Routledge}.
\newblock
\urldef\tempurl%
\url{https://doi.org/10.1201/9781315140919}
\showDOI{\tempurl}


\bibitem[Stankovič et~al\mbox{.}(2022)]%
        {Stankovicetal2022}
\bibfield{author}{\bibinfo{person}{Miroslav Stankovič}, \bibinfo{person}{Ezio Bartocci}, {and} \bibinfo{person}{Laura Kovács}.} \bibinfo{year}{2022}\natexlab{}.
\newblock \showarticletitle{Moment-based analysis of Bayesian network properties}.
\newblock \bibinfo{journal}{\emph{Theoretical Computer Science}}  \bibinfo{volume}{903} (\bibinfo{year}{2022}), \bibinfo{pages}{113--133}.
\newblock
\showISSN{0304-3975}
\urldef\tempurl%
\url{https://doi.org/10.1016/j.tcs.2021.12.021}
\showDOI{\tempurl}


\bibitem[Steinhardt and Tedrake(2012a)]%
        {SteinhardtRuss2012}
\bibfield{author}{\bibinfo{person}{Jacob Steinhardt} {and} \bibinfo{person}{Russ Tedrake}.} \bibinfo{year}{2012}\natexlab{a}.
\newblock \showarticletitle{Finite-time regional verification of stochastic non-linear systems}.
\newblock \bibinfo{journal}{\emph{The International Journal of Robotics Research}} \bibinfo{volume}{31}, \bibinfo{number}{7} (\bibinfo{year}{2012}), \bibinfo{pages}{901--923}.
\newblock


\bibitem[Steinhardt and Tedrake(2012b)]%
        {SteinhardtT12}
\bibfield{author}{\bibinfo{person}{Jacob Steinhardt} {and} \bibinfo{person}{Russ Tedrake}.} \bibinfo{year}{2012}\natexlab{b}.
\newblock \showarticletitle{Finite-time regional verification of stochastic non-linear systems}.
\newblock \bibinfo{journal}{\emph{Int. J. Robotics Res.}} \bibinfo{volume}{31}, \bibinfo{number}{7} (\bibinfo{year}{2012}), \bibinfo{pages}{901--923}.
\newblock
\urldef\tempurl%
\url{https://doi.org/10.1177/0278364912444146}
\showDOI{\tempurl}


\bibitem[Szabłowski(2015)]%
        {SZABLOWSKI_2015}
\bibfield{author}{\bibinfo{person}{Paweł~J. Szabłowski}.} \bibinfo{year}{2015}\natexlab{}.
\newblock \showarticletitle{A few remarks on orthogonal polynomials}.
\newblock \bibinfo{journal}{\emph{Appl. Math. Comput.}}  \bibinfo{volume}{252} (\bibinfo{year}{2015}), \bibinfo{pages}{215--228}.
\newblock
\showISSN{0096-3003}
\urldef\tempurl%
\url{https://doi.org/10.1016/j.amc.2014.11.112}
\showDOI{\tempurl}


\bibitem[Szegő(1939)]%
        {Szego1939}
\bibfield{author}{\bibinfo{person}{Gábor Szegő}.} \bibinfo{year}{1939}\natexlab{}.
\newblock \bibinfo{booktitle}{\emph{Orthogonal Polynomials}}.
\newblock \bibinfo{publisher}{American Mathematical Society}.
\newblock


\bibitem[Sz{\'e}kely and Rizzo(2004)]%
        {SzekelyRizzo2004}
\bibfield{author}{\bibinfo{person}{G{\'a}bor~J. Sz{\'e}kely} {and} \bibinfo{person}{Maria~L. Rizzo}.} \bibinfo{year}{2004}\natexlab{}.
\newblock \showarticletitle{TESTING FOR EQUAL DISTRIBUTIONS IN HIGH DIMENSION} \emph{(\bibinfo{series}{InterStat}, Vol.~\bibinfo{volume}{5})}.
\newblock


\bibitem[Taylor(1993)]%
        {Taylor1993}
\bibfield{author}{\bibinfo{person}{John~B. Taylor}.} \bibinfo{year}{1993}\natexlab{}.
\newblock \showarticletitle{{Discretion versus policy rules in practice}}.
\newblock \bibinfo{journal}{\emph{Carnegie-Rochester Conference Series on Public Policy}} \bibinfo{volume}{39}, \bibinfo{number}{1} (\bibinfo{date}{December} \bibinfo{year}{1993}), \bibinfo{pages}{195--214}.
\newblock
\urldef\tempurl%
\url{https://ideas.repec.org/a/eee/crcspp/v39y1993ip195-214.html}
\showURL{%
\tempurl}


\bibitem[Tekel and Cohen(2012)]%
        {Tekel_Cohen}
\bibfield{author}{\bibinfo{person}{J. Tekel} {and} \bibinfo{person}{L. Cohen}.} \bibinfo{year}{2012}\natexlab{}.
\newblock \showarticletitle{{Constructing and estimating probability distributions from moments}}. In \bibinfo{booktitle}{\emph{Automatic Target Recognition XXII}}, \bibfield{editor}{\bibinfo{person}{Firooz~A. Sadjadi} {and} \bibinfo{person}{Abhijit Mahalanobis}} (Eds.), Vol.~\bibinfo{volume}{8391}. International Society for Optics and Photonics, \bibinfo{publisher}{SPIE}, \bibinfo{pages}{83910E}.
\newblock
\urldef\tempurl%
\url{https://doi.org/10.1117/12.919443}
\showDOI{\tempurl}


\bibitem[Vasicek(1977)]%
        {Vasicek1977}
\bibfield{author}{\bibinfo{person}{Oldrich Vasicek}.} \bibinfo{year}{1977}\natexlab{}.
\newblock \showarticletitle{An equilibrium characterization of the term structure}.
\newblock \bibinfo{journal}{\emph{Journal of Financial Economics}} \bibinfo{volume}{5}, \bibinfo{number}{2} (\bibinfo{year}{1977}), \bibinfo{pages}{177--188}.
\newblock
\showISSN{0304-405X}
\urldef\tempurl%
\url{https://doi.org/10.1016/0304-405X(77)90016-2}
\showDOI{\tempurl}


\bibitem[Xiu and Karniadakis(2002)]%
        {XiuKarniadakis2002a}
\bibfield{author}{\bibinfo{person}{Dongbin Xiu} {and} \bibinfo{person}{George Karniadakis}.} \bibinfo{year}{2002}\natexlab{}.
\newblock \showarticletitle{The {W}iener-{A}skey Polynomial Chaos for Stochastic Differential Equations}.
\newblock \bibinfo{journal}{\emph{SIAM J. Sci. Comput.}} \bibinfo{volume}{24}, \bibinfo{number}{2} (\bibinfo{date}{Feb.} \bibinfo{year}{2002}), \bibinfo{pages}{619–644}.
\newblock
\showISSN{1064-8275}
\urldef\tempurl%
\url{https://doi.org/10.1137/S1064827501387826}
\showDOI{\tempurl}


\end{thebibliography}

\appendix

\section{Proofs of Theorems 1 and 2}\label{app:proofs}


\textit{Proof of Theorem~\ref{thm:MM}: }
Suppose $f$ is supported on $(a,b)$. Then, $\phi(x)=\phi=1/(b-a)$. Let $\{l_{j}(x) = \sum_{i=0}^{j}\lambda_{ji} x^{i}\}_{j = 0}^{n}$ be the set of the first $n$ shifted scaled Legendre polynomials that are orthonormal on $\left[a, b\right]$, so that $\Lambdabf = (\lambda_{ji})_{j, i = 0}^{n}$ is  a lower triangular matrix. 

Every polynomial of degree $n$ can be expressed as a weighted sum of polynomials of degree up to $n$. In such a case, we can represent the MM estimator $\hat{f}_{\mm} = \sum_{i=0}^{n}p_{i}x^{i}$ as a weighted sum of Legendre polynomials $1, l_{1}(x), \ldots, l_{n}(x)$ with weight coefficients $\phi\cdot a_{j}$, $j = 0,\ldots, n$. Then,
\begin{equation}\notag
    \hat{f}_{\mm}(x) = 
    \sum\limits_{i=0}^{n}p_{i}x^{i} = \phi \sum\limits_{j=0}^{n}a_{j}l_{j}(x) = \phi\sum\limits_{j=0}^{n}a_{j}\sum\limits_{i=0}^{j}\lambda_{ji} x^{i},
\end{equation}
or, equivalently,
   $ \mathbf{p}_n^T\mathbf{x}_n= \phi  \cdot \mathbf{a}_n^{T} \mathbf{l}_n = \phi \cdot \mathbf{a}_n^{T} \Lambdabf \mathbf{x}_n,$
where $\mathbf{l}_n=(1,l_{1}(x),\ldots, l_{n}(x))^{T}$ and $\mathbf{a}_n=(1,a_{1},\ldots, a_{n})^{T}$.
Thus, 
  $ \mathbf{p}_n = \phi \cdot \Lambdabf^{T}\mathbf{a}_n.$
Now,  \begin{equation}\label{MM_system}
    \mathbf{m}_n = \M_{ab} \cdot \mathbf{p}_n,
\end{equation}
implies 
   $\mathbf{m}_n = \M_{ab} \cdot \mathbf{p}_n = \phi \cdot \M_{ab} \cdot \Lambdabf^{T}\mathbf{a}_n, $
where  $\M_{ab}$ is the matrix with elements the integrals of powers of $x$ over the interval $\left[a, b\right]$,
\begin{equation}
    \M_{ab} = \begin{pmatrix}
                    b-a & \frac{b^{2}-a^{2}}{2} & \ldots & \frac{b^{n+1}-a^{n+1}}{n+1}\\
                    \frac{b^{2}-a^{2}}{2} & \frac{b^{3}-a^{3}}{3} & \ldots & \frac{b^{n+2}-a^{n+2}}{n+2}\\
                    \vdots & & \ddots & \\
                    \frac{b^{n+1}-a^{n+1}}{n+1} & \frac{b^{n+2}-a^{n+2}}{n+2} & \ldots & \frac{b^{2n+1}-a^{2n+1}}{2n+1}
                    \end{pmatrix}.
\end{equation}
In the matrix form of the K-series estimator (7, main paper), $\Lambdabf \cdot \mathbf{m}_n = \mathbf{a}_n$. It suffices to show that
\begin{equation}\label{theorem_key_statement}
    \phi \cdot \M_{ab} \cdot \Lambdabf^{T} = \Lambdabf^{-1}.
\end{equation}
The matrix $\phi \cdot \M_{ab}$ contains the moments of the uniform distribution. Therefore, $\Lambdabf$ is a matrix with entries the coefficients of orthonormal polynomials and the left lower triangular factor of the Cholesky decomposition of the moment matrix $\phi \cdot \M_{ab}$. Thus,  \eqref{theorem_key_statement} follows from  \cite[Prop. 2(i)]{SZABLOWSKI_2015}. $\qed$

\medskip

\textit{Proof of Theorem~\ref{thm:L1_convergence}: }
 \begin{align}
       \left\Vert \frac{\widetilde{f}(x)}{\phi(x)} - \sum\limits_{i=0}^{n}\alpha_{i}h_{i}(x) \right\Vert_{\phi}^{2} &= \int\limits_{\Theta}\left[\frac{\widetilde{f}(x)}{\phi(x)} - \sum\limits_{i=0}^{n}\alpha_{i}h_{i}(x) \right]^{2}\phi(x)dx \notag \\ 
        &=\int\limits_{\Theta}\left[\widetilde{f}(x) - \phi(x)\sum\limits_{i=0}^{n}\alpha_{i}h_{i}(x) \right]^{2}\frac{1}{\phi(x)}dx \notag\\ 
        &
        =\int\limits_{\Theta}\phi(x)dx
        \int\limits_{\Theta}\left[\widetilde{f}(x) - \phi(x)\sum\limits_{i=0}^{n}\alpha_{i}h_{i}(x) \right]^{2}\frac{1}{\phi(x)}dx  \notag \\
        &=||\sqrt{\phi(x)}||_{1}^2 \cdot \left\Vert \left(\widetilde{f}(x) - \phi(x)\sum\limits_{i=0}^{n}\alpha_{i}h_{i}(x) \right)\frac{1}{\sqrt{\phi(x)}} \right\Vert_1^2 \notag \\
        & \geq \left(\int\limits_{\Theta}\Bigl\lvert \widetilde{f}(x) - \phi(x)\sum\limits_{i=0}^{n}\alpha_{i}h_{i}(x)\Bigl\rvert dx\right)^2, \label{thm2}
    \end{align}
   where the last inequality is due to Cauchy-Schwarz  
   inequality. The function $\widetilde{f}(x)$ in \eqref{def:ftilde} is a  density. In the case where $\Theta$ is bounded, $\phi(x)$ is uniquely identifiable by its moments. When $\Theta$ is unbounded, $\phi(x)$ is exponentially integrable by the assumption (a) of the theorem. Hence,  for all $n \ge 1$, $\phi(x) \Bigl\lvert\sum_{i=0}^{n}\alpha_{i}h_{i}(x)\Bigl\rvert$ is integrable. 
   
Since the truncated series $\sum\limits_{i=0}^{n}\alpha_{i}h_{i}(x)$ converges to $g(x) = \widetilde{f}(x) / \phi(x)$ in $L_{2}(\Theta, \phi)$, as $n \to \infty$, from \eqref{thm2} we obtain that the K-series estimator \eqref{estimator_3} converges to the extended true pdf $\widetilde{f}(x)$ in $L_{1}(\Theta, 1)$.

Next, suppose  $\phi(x)$ is the pdf of the uniform distribution, so that $\Theta$ is bounded, and $\phi(x) = c$. Then,
\begin{align*}
\left\Vert \frac{\widetilde{f}}{\phi} - \sum\limits_{i=0}^{n}\alpha_{i}h_{i}(x) \right\Vert_{\phi}^{2} &= \int\limits_{\Theta}\left[\frac{\widetilde{f}(x)}{\phi(x)} - \sum\limits_{i=0}^{n}\alpha_{i}h_{i}(x) \right]^{2}\phi(x)dx  =\frac{1}{c}\int\limits_{\Theta}\left[\widetilde{f}(x) - c\sum\limits_{i=0}^{n}\alpha_{i}h_{i}(x) \right]^{2}dx
    \end{align*}
Hence, $\widetilde{f}(x) = c \cdot g(x)$ is in $L_{2}(\Theta, 1)$, and $\int\limits_{\Theta}c^{2} \left[\sum_{i=0}^{n}\alpha_{i}h_{i}(x)\right]^{2} dx$ is an integral of a polynomial over a bounded interval, so that the K-series estimator \eqref{estimator_3} converges to the true pdf $\widetilde{f}(x)$ in $L_{2}(\Theta, 1)$. $\qed$
   


\section{1D and 2D Random Walk}\label{app:examples}

Panel A in Fig. \ref{fig:random_walk} describes the \textit{1D Random Walk}, and panel B  the \textit{2D Random Walk} \citep{Kura19}. For the former, we used a truncated normal distribution on $(-98, 102)$ as reference.  For the 2D Random Walk, we used two independent truncated normal distributions on $(-100, 100) \times (-100, 100)$ with true means and variances of corresponding marginal pdfs obtained with the algorithm in \cite{BartocciKS19}.
The K-series estimator exhibits excellent performance for both 1D and 2D random walks, as can be seen in Fig. \ref{fig:random_walk}.

\begin{figure}[h!]
\centering
  \includegraphics[width=1.0\textwidth, height=0.6\textwidth]{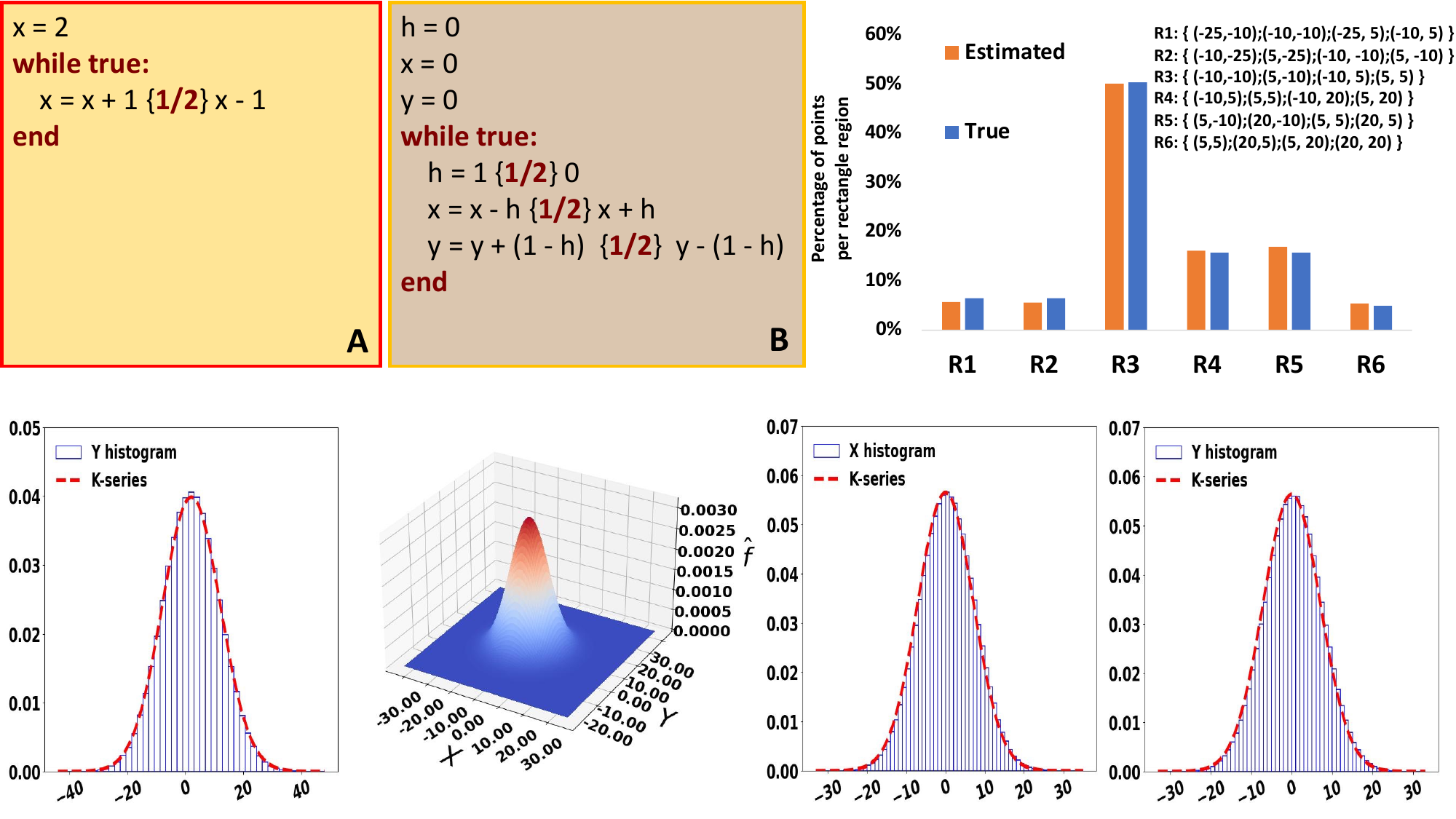}
  \caption{K-series estimates of the pdf of $X$ in 1D Random Walk (A)~\cite{Kura19} at iteration $t = 100$, marginal pdfs for variables $X$, $Y$ and joint distribution of ($X$, $Y$) in 2D Random Walk (B)~\cite{Kura19} at the iterations $t = 100$.}
  \label{fig:random_walk}
\end{figure}

\section{Truncated Bivariate Normal}\label{app:bivar_normal}

Suppose we want to recover the joint pdf of two random variables $X$ and $Y$  on a set $\Omega = \left[-2, 2\right] \times \left[-4, 5\right]$ using their first eight  cross-moments,
\begin{align}\label{ex3:moms}
   \left(m_{x^{j}y^{i}}=
   \E(X^jY^i)\right)_{i, j = 0,\ldots, 2}
    &=\begin{pmatrix} \notag
        m_{x^{0}y^{0}} & m_{x^{1}y^{0}} & m_{x^{2}y^{0}} \\
        m_{x^{0}y^{1}} & m_{x^{1}y^{1}} & m_{x^{2}y^{1}}  \\
        m_{x^{0}y^{2}} & m_{x^{1}y^{2}} & m_{x^{2}y^{2}} 
    \end{pmatrix}\\ \notag
    \\
    &=\begin{pmatrix}
        1.00000 & 0.71721 & 1.13054 \\
        1.99556 & 1.43124 & 2.25606 \\
        4.96894 & 3.56379 & 5.61757
    \end{pmatrix}.
\end{align}
We choose the reference marginal pdfs be both truncated normal $\phi_{x}(z_{x})$ and $\phi_{y}(z_{y})$ with $Z_{x} \sim Trunc$ $\mathcal{N}(m_{x}, m_{x^{2}} - m^{2}_{x}, \left[-2, 2\right]) = Trunc$ $ \mathcal{N}(0.71721, 0.61614, $ $\left[-2, 2\right])$, and 
$Z_{y} \sim Trunc$ $\mathcal{N}(m_{y}, m_{y^{2}} - m^{2}_{y}, \left[-4, 5\right])=Trunc$ $\mathcal{N}(1.99556, 0.98667, $ $\left[-4, 5\right])$, respectively.

We construct sets of univariate orthonormal polynomials using, for example, the Gram-Schmidt orthogonalization procedure, and obtain
\[
\begin{matrix}
    h^{x}_{0}(z_{x}) = 1, &  h^{y}_{0}(z_{y}) = 1,\\
    h^{x}_{1}(z_{x}) = 1.42119z_{x} - 0.89705, & h^{y}_{1}(z_{y}) = 1.01307z_{y} - 2.01751,\\
    h^{x}_{2}(z_{x}) = 1.58907z^{2}_{x} - 1.63885z_{x} - 0.38542, \,  & h^{y}_{2}(z_{y}) = 0.74083z^{2}_{y} - 2.92557z_{y} + 2.16624
\end{matrix}
\]
Hence, starting from a reference joint pdf that is the product of the pdfs of the independent random variables $Z_x$ and $Z_y$, $\widetilde{\phi}(z_{x}, z_{y}) = \phi_{x}(z_{x})\phi_{y}(z_{y})$, the multivariate orthogonal polynomials are simply all the pairwise products of univariate polynomials:
\begin{align*}
\tilde{h}_{0,0}(z_{x},z_{y}) &= 1\\
\tilde{h}_{0,1}(z_{x},z_{y}) &= 1.01307z_{y} - 2.01751\\
\tilde{h}_{0,2}(z_{x},z_{y}) &= 0.74083z^{2}_{y} - 2.92557z_{y} + 2.16624\\
\tilde{h}_{1,0}(z_{x},z_{y}) &= 1.42119z_{x} - 0.89705\\
\tilde{h}_{1,1}(z_{x},z_{y}) &= 1.43976z_{x}z_{y} - 2.86727z_{x} - 0.90877z_{y} + 1.80981\\
\tilde{h}_{1,2}(z_{x},z_{y}) &= 1.05286z_{x}z^{2}_{y} - 4.15779z_{x}z_{y} + 3.07864z_{x} - 0.66456z^{2}_{y} + 2.62438z_{y} \\
& \quad- 1.94323\\
\tilde{h}_{2,0}(z_{x},z_{y}) &= 1.58907z^{2}_{x} - 1.63885z_{x} - 0.38542\\
\tilde{h}_{2,1}(z_{x},z_{y}) &= 1.60984z^{2}_{x}z_{y} - 3.20596z^{2}_{x} - 1.66027z_{x}z_{y} + 3.30634z_{x} - 0.39046z_{y} \\
&\quad+ 0.77759\\
\tilde{h}_{2,2}(z_{x},z_{y}) &= 1.17723z^{2}_{x}z^{2}_{y} - 4.64894z^{2}_{x}z_{y} + 3.44231z^{2}_{x} - 1.21411z_{x}z^{2}_{y} \\
&\quad+ 4.79457z_{x}z_{y} -3.55014z_{x} - 0.28553z^{2}_{y} + 1.12757z_{y} - 0.83491
\end{align*}
In order to compute the coefficients $\alpha(i_1,i_2)$ of the PCE along the reference pdf $\widetilde{\phi}(z_{x}, z_{y})$ for each polynomial $\Tilde{h}_{i_1,i_2}(z_{x}, z_{y})$, 
we need to substitute every monomial factor $z_{x}^{j}z_{y}^{i}$ by the corresponding moment $m_{x^{j}y^{i}}$ from \eqref{ex3:moms} in each polynomial. For example, the coefficient of $\tilde{h}_{1,1}(z_x,z_y)$ is $1.43976m_{xy} - 2.86727m_{x} - 0.90877m_{y} + 1.80981 = 1.43976\cdot 1.43124 - 2.86727 \cdot 0.71721 - 0.90877 \cdot 1.99556 + 1.80981 = 0.00051$.
The resulting estimator is 
\begin{align*}
\hat{f}(z_{x},z_{y}) &= \phi_{1}(z_{x})\phi_{2}(z_{y})\sum_{i_{1}, i_{2}=(0,0)}^{(2,2)}\alpha(i_1, i_2)\Tilde{h}_{i_{1}, i_{2}}(z_{x},z_{y})\\
&= \phi_{1}(z_{x})\phi_{2}(z_{y}) \times \left[ 1 +0.00415  \cdot \Tilde{h}_{0,1}(z_{x},z_{y})+  0.00924 \cdot \Tilde{h}_{0,2}(z_{x},z_{y}) \right. \\ 
&\quad + 0.12224 \cdot \Tilde{h}_{1,0}(z_{x},z_{y}) 
+ 0.00051 \cdot \Tilde{h}_{1,1}(z_{x},z_{y}) +  0.00113 \cdot \Tilde{h}_{1,2}(z_{x},z_{y})  \\
&\quad \left. +0.23568 \cdot \Tilde{h}_{2,0}(z_{x},z_{y}) +  0.00098 \cdot \Tilde{h}_{2,1}(z_{x},z_{y}) +  0.00218 \cdot \Tilde{h}_{2,2}(z_{x},z_{y}) \right]
\end{align*}
The  estimated bivariate density is plotted in Fig. \ref{fig:multivar_example} (a). In panel (b), we plot the frequencies of $X$ and $Y$ under the true $f(x, y)$ (blue bars)  and its K-series (red bars) pdf estimate over a 2D grid comprising of eight parallelograms, where we can see their close agreement.  

\begin{figure}[!h]
\centering
  \includegraphics[width=0.90\textwidth, height=0.45\textwidth]
  {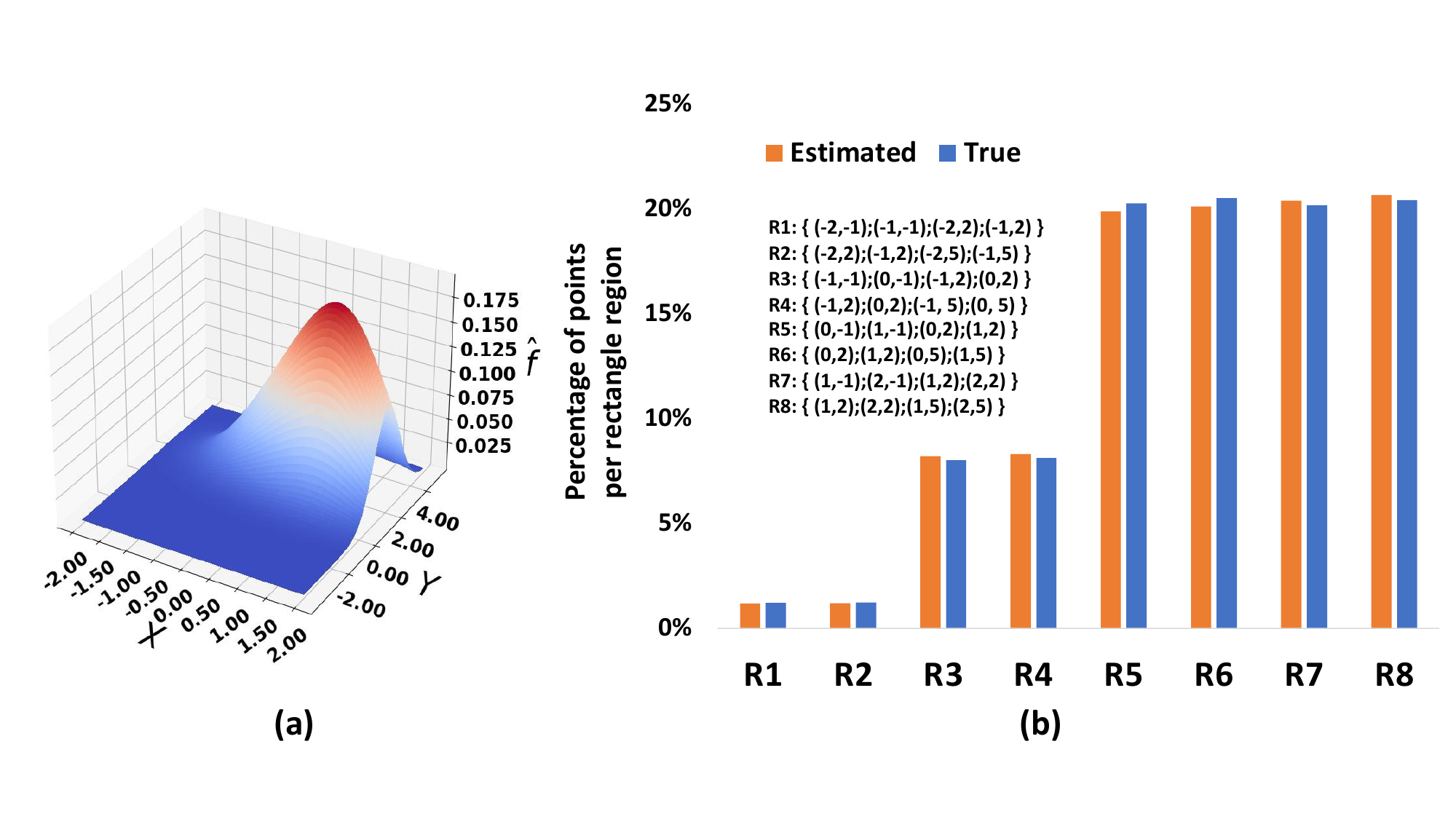}
  \caption{K-series estimates  of the truncated bivariate normal distribution $f(x, y) = Trunc \hspace{0.1cm} Normal((1, 2), (1, 1), -0.3, \left[-2, 2\right], \left[-4, 5\right])$.}
    \label{fig:multivar_example}
\end{figure}

\vspace{0.5cm}
\begin{threeparttable}
\centering
\ra{0.5}
\begin{tabular}{@{}lcccrr@{}}\toprule
Example &  $Var$& \phantom{abc} & $|M|$ & \showclock{0}{0} Orthogonalization  &  \showclock{0}{0} K-series \\
 &   &  \phantom{abc} &   &  Runtime (in seconds)&   Runtime (in seconds) 

\\
  \midrule \\
Truncated exponential
  & $X$ && 2 & 0.00246 & 0.04379 \\
  \hline\\
   The Irwin-Hall Distribution
  & $X$ && 6 & 0.05029 & 0.06922 \\
  \hline\\
     Probabilistic loop with \\
non-polynomial assignment  & $r$ && 4 & 0.03568 & 0.04936  
\\ \midrule \\
  Truncated Bivariate Normal
  &
  $(X,Y)$
  && 8 
  & 0.07077
  & 0.03613
  \\
 \bottomrule
\end{tabular}
\caption{Runtimes of orthogonalization procedure and K-series estimation for the illustrative benchmarks. \\ $|M|$ denotes number of used moments and $Var$ the variable(s) whose density is estimated.}  
\label{tab:Illustrations_time}
\end{threeparttable}

\newpage

\section{Kolmogorov-Smirnov and Energy Tests for Equality of Distributions}\label{sec:KS}

The Kolmogorov-Smirnov (K-S) test~\cite{Hollanderetal2013}  compares two  cumulative distribution functions (cdfs). 
We compute the cdf $\hat{F}_{\ks}$ 
of the estimated pdf $\hat{f}_{\ks}$. 
We also compute the (empirical) cdf~$F_{Sample}$ of the data resulting from sampling the probabilistic program variables.
The 2-sample Kolmogorov-Smirnov (K-S) test statistic for testing equality of the population (true) cdfs is 
\begin{equation}\label{eq:KSTest}
     D_{\ks}= \max_x\left(|F_{\ks}(x)-F_{Sample}(x)|\right),
\end{equation}
where $N_1$ and $N_2$ are the sample sizes from the K-series and empirical cdf, respectively.
We reject the equality of the two distributions if 
\[ D_{\ks} > c(\alpha) \sqrt{\frac{N_1 +N_2}{N_1\cdot N_2}}=\sqrt{-\frac{1}{2}\ln\frac{\alpha}{2}} \;  \sqrt{\frac{N_1 +N_2}{N_1\cdot N_2}}
\]
at significance level $\alpha$. 

The two-sample E-statistic for testing for equality of multivariate distributions proposed by \cite{SzekelyRizzo2004} is the \textit{energy} distance $e(S_1,S_2)$, which is defined  by
$e(S_1,S_2)= N_1 N_2\left(2 D_{12} -D_{11}-D_{22}\right)/(N_1+N_2)$,
for two samples $S_1, S_2$ of respective sizes $N_1$, $N_2$, where
$D_{ij} =\sum_{p=1}^{N_i}\sum_{q=1}^{N_j} || \X_{ip} -\X_{jq}||/(N_{i}N_{j}),$ $i,j=1,2,$ $||\cdot||$ denotes the Euclidean norm, and $\X_{1p}$ denotes the $p$-th and $\X_{2q}$ the $q$-th (vector-valued) observations in the first and second sample, respectively. The test is implemented by nonparametric bootstrap, an approximate permutation test in the \texttt{R}-package \texttt{energy} \citep{RizzoSzekely2022}.

We used the Kolmorov-Smirnov test to compare univariate distributions and the \textit{energy}  test for multivariate distributions \citep{SzekelyRizzo2004}. We draw 1000 observations from the sampling (``true'') and estimated distributions. The critical values are $0.0607$ and $0.0479$ for significance levels $0.05$ and $0.2$, respectively. Except for very few instances, when a small number of moments is used in the K-series estimation, our estimate is statistically the same as the true distribution. We also test the agreement of the K-series with the GC estimates. When the true distribution is similar to normal, K-series is statistically indistinguishable from Gram-Charlier. But when the true distribution is not close to normal, K-series provides a far more accurate estimate than Gram-Charlier.

\newpage
\begin{threeparttable}
\centering
\ra{0.5}
\resizebox{0.87\columnwidth}{!}{%
\hspace{-0.85cm}
\begin{tabular}{@{}lcccccrclcr@{}}\toprule
Problem & \phantom{abc}&  $Var$& \phantom{abc} & $|M|$ &  \phantom{abc} & KS Distance  & \phantom{abc} & KS Distance  & \phantom{abc}& Energy test \\
 & \phantom{abc}&  & \phantom{abc} &  &  \phantom{abc} &   & \phantom{abc}& {\hspace{0.5cm}(GC)}  & \phantom{abc}& {(p-value)}
\\ \midrule \\
Differential-Drive Robot\\
  && $X$ 
  && 6
 && 
       0.00069 {\color{green}{\ding{52}}} {\color{red}{!}}  
  && 
       0.00072 {\color{green}{\ding{52}}} {\color{red}{!}}  \\
&& 
  \\ 
  && 
  $Y$ 
  && 6 
  && 0.00059 {\color{green}{\ding{52}}} {\color{red}{!}}
  &&  0.00059 {\color{green}{\ding{52}}} {\color{red}{!}} 
  && \\
  &&  &&  &&  &&   && \\
  &&
  $(X,Y)$
  && 48 
  && 
  && 
  && 0.4700 \\
  \midrule \\
  PDP\\
  && $X$ && 2 && 0.00664 {\color{green}{\ding{52}}} {\color{red}{!}} && 0.00680 {\color{green}{\ding{52}}} {\color{red}{!}} && 
  \\ &&  &&  &&  &&   && \\ 
  && $Y$ && 6 && 0.00033 {\color{green}{\ding{52}}} {\color{red}{!}} && 0.05190 {\color{green}{\ding{52}}}  && 
  \\ &&  &&  &&  &&   && \\ 
  && $(X, Y)$ && 8 &&  &&   && 0.4250\\
  \midrule \\
Turning vehicle\\
&& $X$ && 8 && 0.00807 {\color{green}{\ding{52}}} {\color{red}{!}} && 0.02109 {\color{green}{\ding{52}}} {\color{red}{!}}  && 
\\ &&  &&  &&  &&   && \\ 
&& $Y$ && 8 && 0.00494 {\color{green}{\ding{52}}} {\color{red}{!}} && 0.01030 {\color{green}{\ding{52}}} {\color{red}{!}}  && 
\\ &&  &&  &&  &&   && \\ 
&& $(X, Y)$ && 80 &&  &&   && 0.4150\\
  \midrule \\
Turning vehicle \\
 (small variance) \\
&& $X$ && 8 && 0.02614 {\color{green}{\ding{52}}} {\color{red}{!}} && 0.11054 {\color{red}{\ding{55}}}  && 
\\ &&  &&  &&  &&   && \\ 
&& $Y$ && 8 && 0.00070 {\color{green}{\ding{52}}} {\color{red}{!}} && 0.00169 {\color{green}{\ding{52}}} {\color{red}{!}}  && 
\\ &&  &&  &&  &&   && \\ 
&& $(X, Y)$ && 80 &&  &&   && 0.5000\\
  \midrule \\
Taylor rule model\\
  && i   && 6 && 0.00037 {\color{green}{\ding{52}}} {\color{red}{!}} && 0.00037 {\color{green}{\ding{52}}} {\color{red}{!}} && \\
  \midrule \\
2D Robotic Arm \\
  && $X$ && 2 && 0.00037 {\color{green}{\ding{52}}} {\color{red}{!}} && 0.00037 {\color{green}{\ding{52}}} {\color{red}{!}} && 
  \\ &&  &&  &&  &&   && \\ 
  && $Y$ && 2 && 0.00048 {\color{green}{\ding{52}}} {\color{red}{!}} && 0.00048 {\color{green}{\ding{52}}} {\color{red}{!}}  && 
  \\ &&  &&  &&  &&   && \\ 
  && $(X, Y)$ && 8 &&  &&   && 0.9650\\
  \midrule \\
Rimless Wheel Walker\\
  && $X$ && 2 && 0.00180 {\color{green}{\ding{52}}} {\color{red}{!}} && 0.00180 {\color{green}{\ding{52}}} {\color{red}{!}} &&  \\
  \hline\\
    Vasicek model\\
  && $r$ && 2 && 0.00074 {\color{green}{\ding{52}}} {\color{red}{!}} && 0.00074 {\color{green}{\ding{52}}} {\color{red}{!}} &&  \\
  \hline\\
    1D Random Walk\\
  && $X$ && 2 && 0.03834 {\color{green}{\ding{52}}} {\color{red}{!}} && 0.03834 {\color{green}{\ding{52}}} {\color{red}{!}} &&  \\
  \midrule \\
  2D Random Walk\\
   && $X$ && 2 && 0.02743 {\color{green}{\ding{52}}} {\color{red}{!}} && 0.02743 {\color{green}{\ding{52}}} {\color{red}{!}} && 
  \\ &&  &&  &&  &&   && \\ 
  && $Y$ && 2 && 0.02714 {\color{green}{\ding{52}}} {\color{red}{!}} && 0.02714 {\color{green}{\ding{52}}} {\color{red}{!}}  && 
  \\ &&  &&  &&  &&   && \\ 
  && $(X, Y)$ && 8 &&  &&   && 0.4902\\
  \midrule \\
    Stuttering P\\
  && $S$ && 2 && 0.00351 {\color{green}{\ding{52}}} {\color{red}{!}} && 0.00354 {\color{green}{\ding{52}}} {\color{red}{!}} && \\
 \bottomrule
\end{tabular}
}
{\small
\begin{tablenotes}
     \item[{\color{green}{\ding{52}}}] \hspace{0.05cm} Null hypothesis is not rejected at  significance level 0.05.
     \item[{\color{green}{\ding{52}}} {\color{red}{!}}] Null hypothesis is not rejected at significance level 0.2.
     \item[{\color{red}{\ding{55}}}] \hspace{0.11cm} Null hypothesis is rejected at  significance level 0.05.
   \end{tablenotes}}
\caption{Kolmogorov-Smirnov distances for univariate distributions and testing for equality of multivariate distributions.}  
\label{tab:Experiments_results}
\end{threeparttable}

\section{Effect of Reference Distribution}\label{sec:refdstn_app}

We study the effect of the choice of the reference distribution in K-series on estimation accuracy. 
We consider reference distributions with the same support as the target unknown pdf $f$, with bounded support that contains the support of $f$ and with unbounded support in absence of any knowledge about the possible values of the target distribution.

Table \ref{tab:distns} lists the combinations of target and reference distributions we consider in our experiments. We plot the true target pdfs (red) and the K-series estimates for different numbers of moments using reference pdfs with the same support as the target in Figure \ref{fig:exact_support_K_series_exp}. Our method does not suffer from the numerical instability associated with closeness to zero. In most cases, the uniform reference pdf works better on exact support.

\begin{table*}\centering
\ra{0.8}
\begin{tabular}{@{}rcr@{}}\toprule
Target pdf $f$ & \phantom{abc}& Reference pdf $\phi$ \\ 
\\ \midrule
Trunc Exp$(\lambda=2/3,\left[0,4\right])$\\
   && $Uniform(0,4)$ \\
  &&  $Trunc$ $Normal(\E(f),\var(f), \left[0,4\right])$ \\
     && $Uniform(-2,6)$ \\
    &&  $Trunc$ $Normal(\E(f),\var(f), \left[-2,6\right])$ \\
    &&  $Normal(\E(f),\var(f))$ \\
Trunc Gamma$(\alpha=2,\beta=0.5,\left[0,5\right])$\\
  && $\phi \sim Uniform(0,5)$ \\
  &&  $Trunc$ $Normal(\E(f),\var(f), \left[0,5\right])$ \\
     && $Uniform(-2,7)$ \\
    &&  $Trunc$ $Normal(\E(f),\var(f), \left[-2,7\right])$ \\
    &&  $Normal(\E(f),\var(f))$ \\ 
  Continuous Bernoulli$(\pi=0.3)$\\
  && $\phi \sim Uniform(0,1)$ \\
  &&  $Trunc$ $Normal(\E(f),\var(f), \left[0,1\right])$ \\
     && $Uniform(-2,3)$ \\
    &&  $Trunc$ $Normal(\E(f),\var(f), \left[-2,3\right])$ \\
    &&  $Normal(\E(f),\var(f))$ \\ 
  Trunc Normal$(1.5, 5.76, \left[-6,6\right])$\\
  && $\phi \sim Uniform(-6,6)$ \\
  &&  $Trunc$ $Normal(\E(f),\var(f), \left[-6,6\right])$ \\
     && $Uniform(-8,8)$ \\
    &&  $Trunc$ $Normal(\E(f),\var(f), \left[-8,8\right])$ \\
    &&  $Normal(\E(f),\var(f))$ \\ \bottomrule
\end{tabular}
\caption{Target and reference distributions.}
\label{tab:distns}
\end{table*}

In Figure \ref{fig:ext_Uniform_K_series}, we plot the true four pdfs in Table \ref{tab:distns} and their K-series estimates using different number of moments and the uniform reference supported on an interval that contains the support of the target pdf. Specifically, the reference pdf is supported on the interval that extends by 2 units the true support in either side. The estimation improves significantly as the number of moments increases. 
The left panels of Figure \ref{fig:normal_ref_K_series} plot the true pdfs and their K-series estimates using different numbers of moments and a truncated normal reference supported on the interval that extends by 2 units the true support in both ends. The right panels of Figure \ref{fig:normal_ref_K_series} plot the true pdfs and their K-series estimates using different numbers of moments and  a normal reference pdf supported on the entire real line. 

Visual inspection of these plots indicates that the estimation is better if the support of all reference pdfs is close to the support of the target pdf. The uniform reference distribution results in accurate estimates provided its support is close to the support of the true pdf. On the other hand, both truncated and regular normal reference pdfs lead to accurate K-series estimates the closer the target pdf is to a normal.  Moreover, the truncated normal distribution tends to work better on a support wider than the true in comparison with the uniform. 

Formal assessment of the estimation accuracy is carried out with Kolmogorov-Smirnov tests. Tables \ref{fig:KS_exact_supports}, \ref{fig:KS_uniform_extended_supports} and \ref{fig:KS_normal_ref}
report the values of the Kolmogorov-Smirnov test statistic comparing the K-series estimates with the true pdfs and whether the null of their equality is rejected for different numbers of moments and reference distributions.  
The sample size for both the estimated and true distribution is 1000. The critical values are $0.0607$ and $0.0479$ for significance levels $0.05$ and $0.2$, respectively.\\

\begin{figure}
     \newcommand{\WIDTH}{0.49\textwidth}
     \newcommand{\HEIGHT}{0.5\textwidth}
     \newcommand{\pWIDTH}{0.95\textwidth}
     \newcommand{\pHEIGHT}{0.555\textwidth}
     \centering
     \vspace{0.8cm}
     \begin{minipage}{1.0\textwidth}
     \centering
     \vspace{-0.8cm}
     \begin{subfigure}[b]{\WIDTH}
         \centering
         \includegraphics[width=\pWIDTH, height = \pHEIGHT]{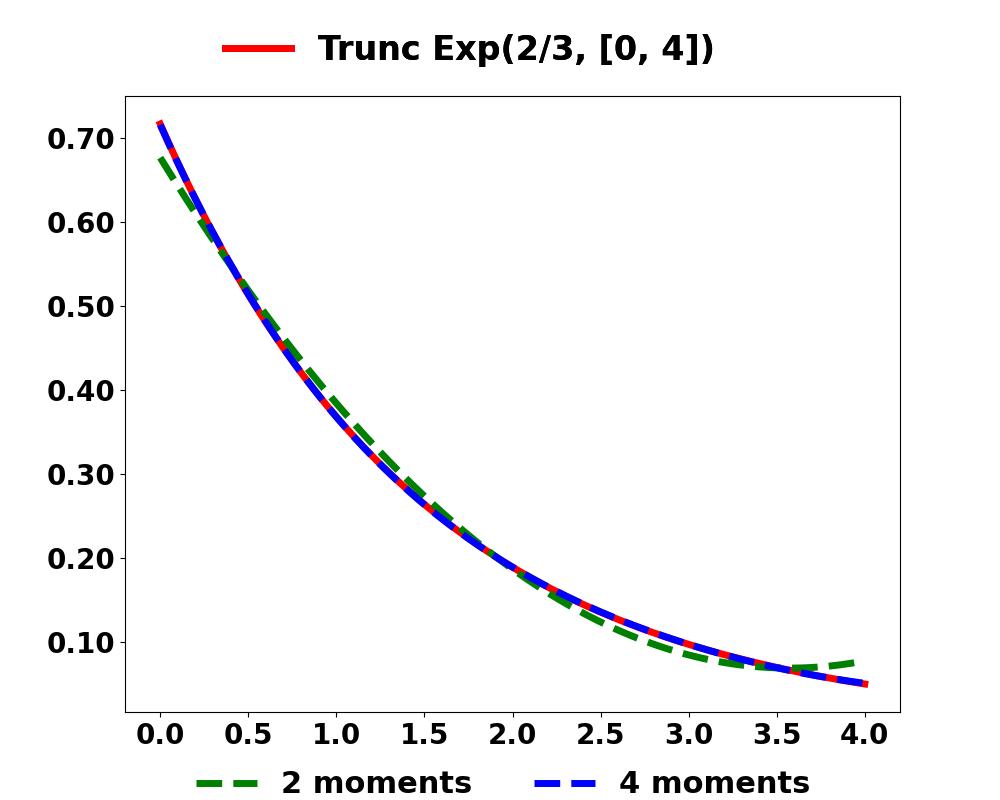}
         \caption{$\phi \sim U(0,4)$}
         \label{fig:y equals x}
     \end{subfigure}
     \hfill
     \begin{subfigure}[b]{\WIDTH}
         \centering
         \includegraphics[width=\pWIDTH, height = \pHEIGHT]{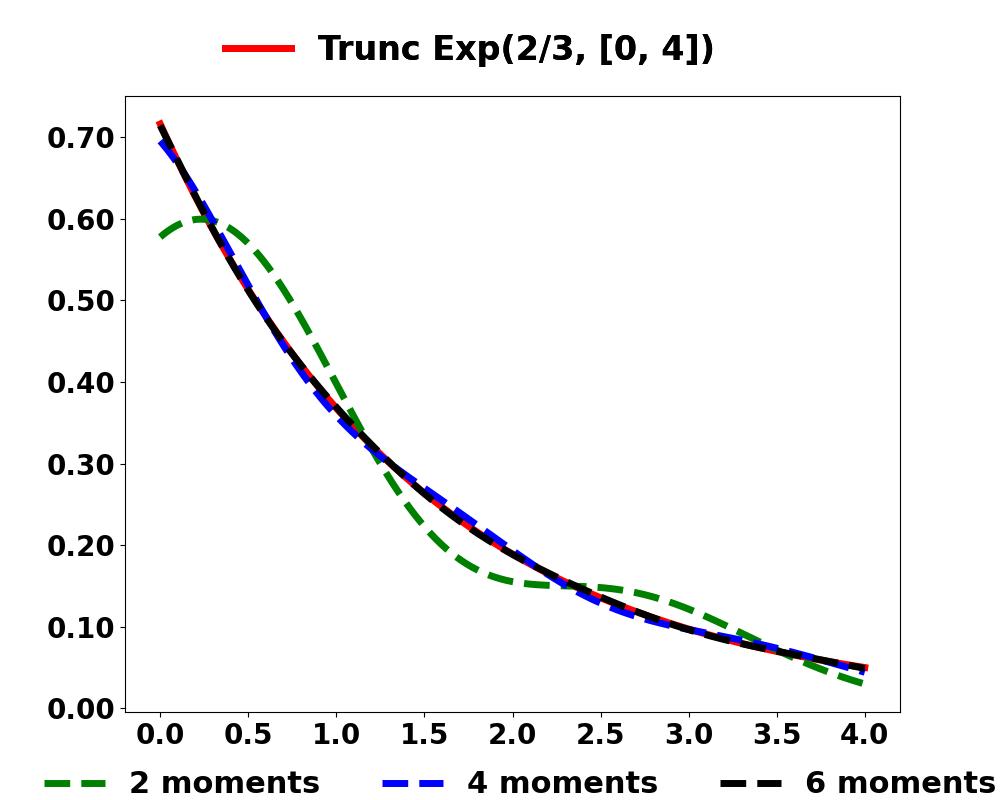}
         \caption{$\phi \sim Trunc$ $Normal(\E(f),\var(f), \left[0, 4\right])$}
         \label{fig:three sin x}
     \end{subfigure}
     \hfill
     \vspace{0.3cm}
     \begin{subfigure}[b]{\WIDTH}
         \centering
         \includegraphics[width=\pWIDTH, height = \pHEIGHT]{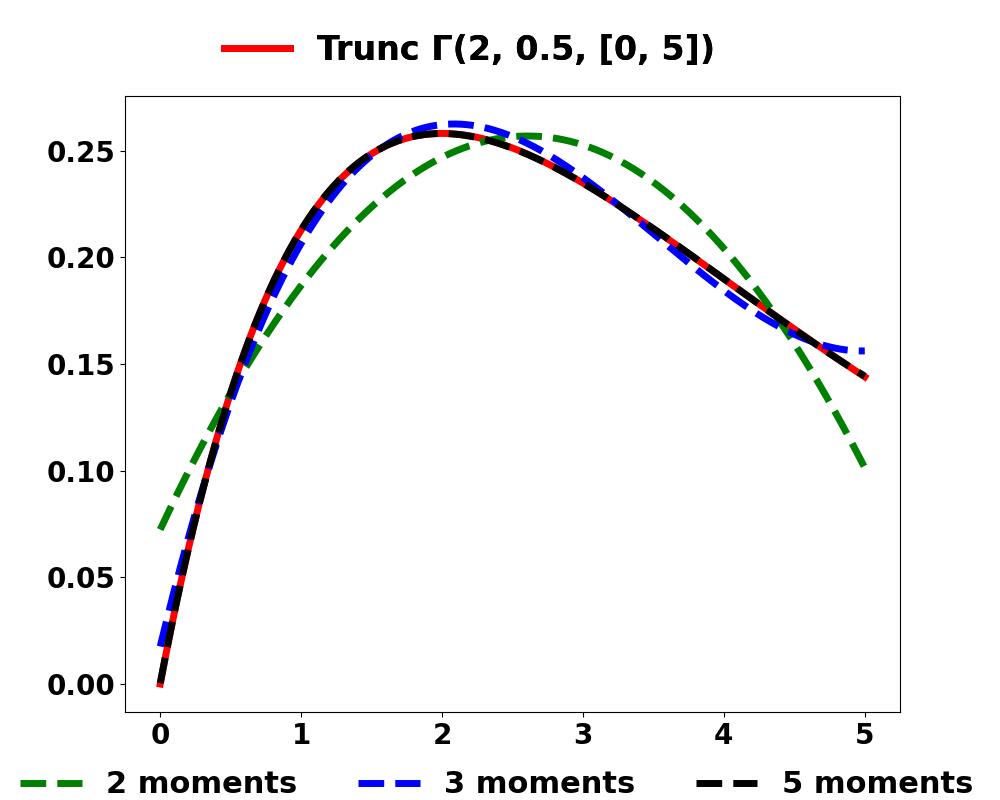}
         \caption{$\phi \sim U(0,5)$}
         \label{fig:five over x}
     \end{subfigure}
     \hfill
     \begin{subfigure}[b]{\WIDTH}
         \centering
         \includegraphics[width=\pWIDTH, height = \pHEIGHT]{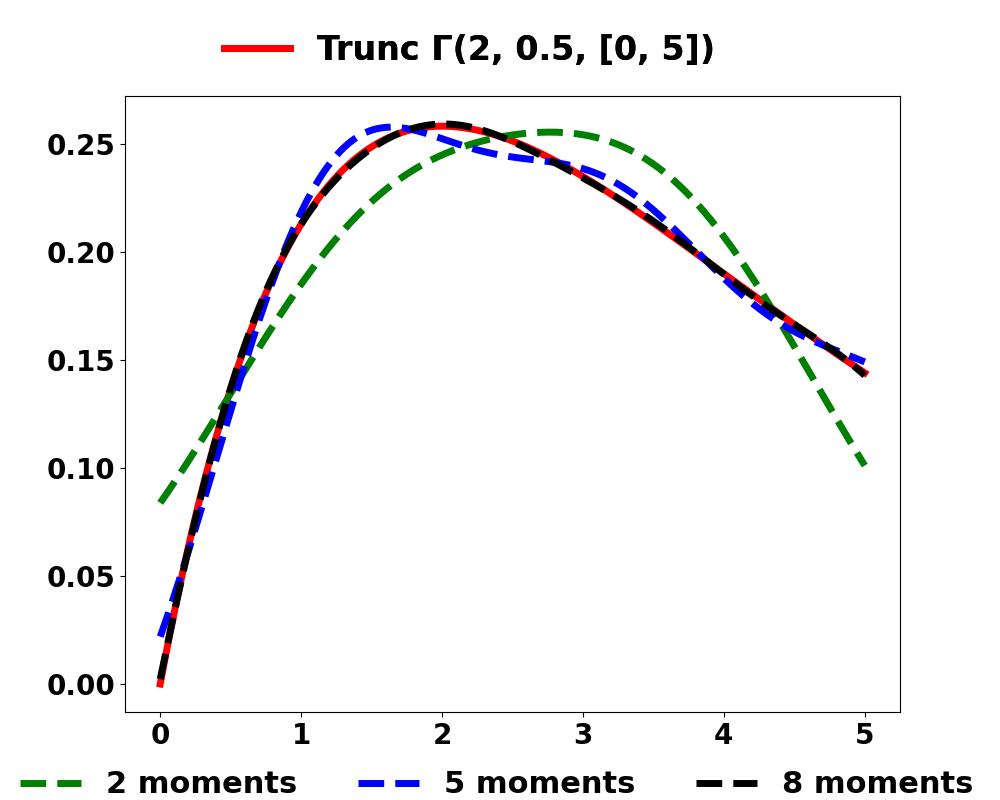}
         \caption{$\phi \sim Trunc$ $Normal(\E(f),\var(f), \left[0, 5\right])$}
         \label{fig:five over x}
     \end{subfigure}
     \hfill
     \vspace{0.3cm}
     \begin{subfigure}[b]{\WIDTH}
         \centering
         \includegraphics[width=\pWIDTH, height = \pHEIGHT]{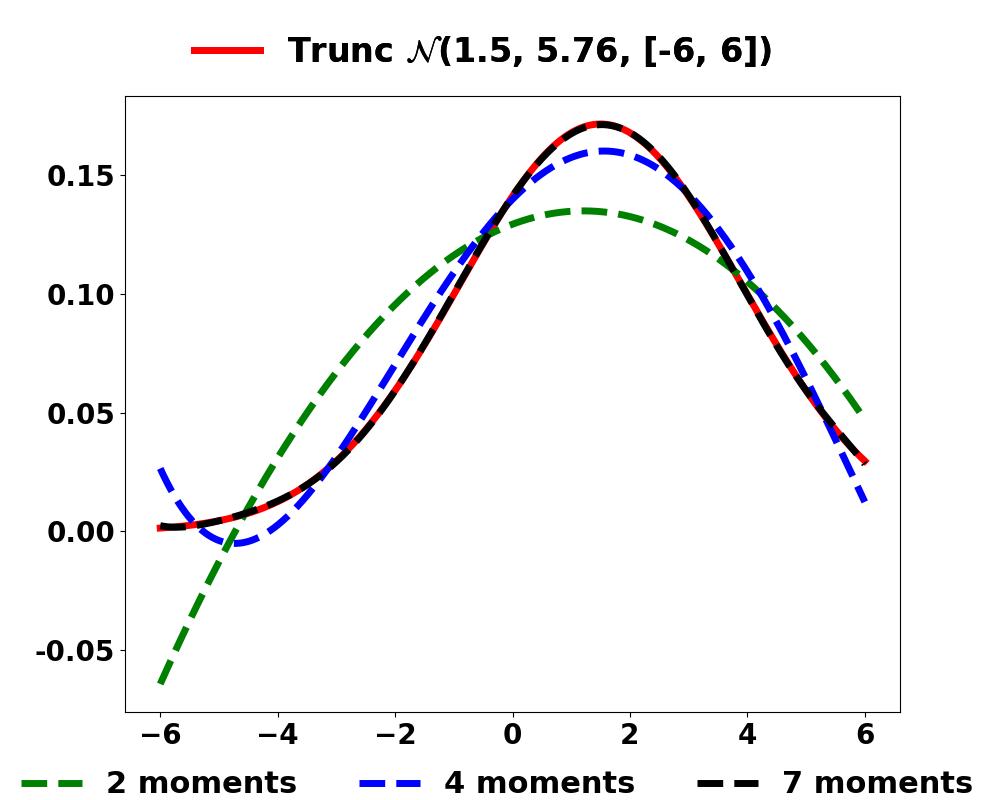}
         \caption{$\phi \sim U(-6,6)$}
         \label{fig:five over x}
     \end{subfigure}
     \hfill
     \begin{subfigure}[b]{\WIDTH}
         \centering
         \includegraphics[width=\pWIDTH, height = \pHEIGHT]{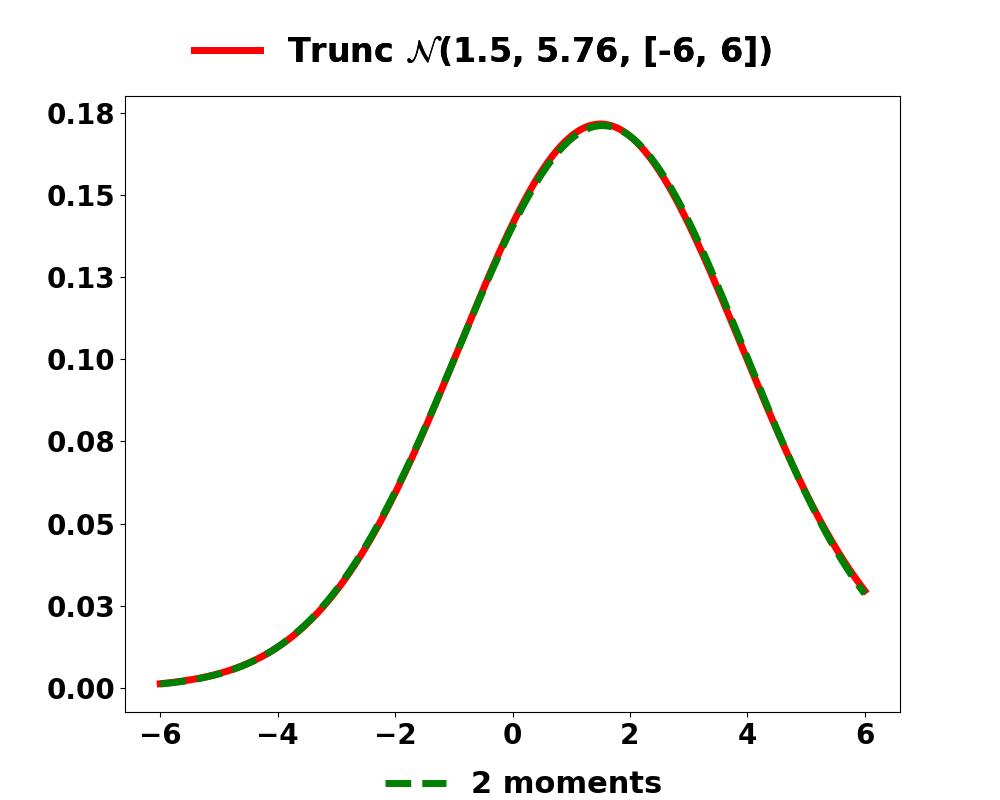}
         \caption{$\phi \sim Trunc$ $Normal(\E(f),\var(f), \left[-6, 6\right])$}
         \label{fig:five over x}
     \end{subfigure}
     \hfill
     \vspace{0.3cm}
     \begin{subfigure}[b]{\WIDTH}
         \centering
         \includegraphics[width=\pWIDTH, height = \pHEIGHT]{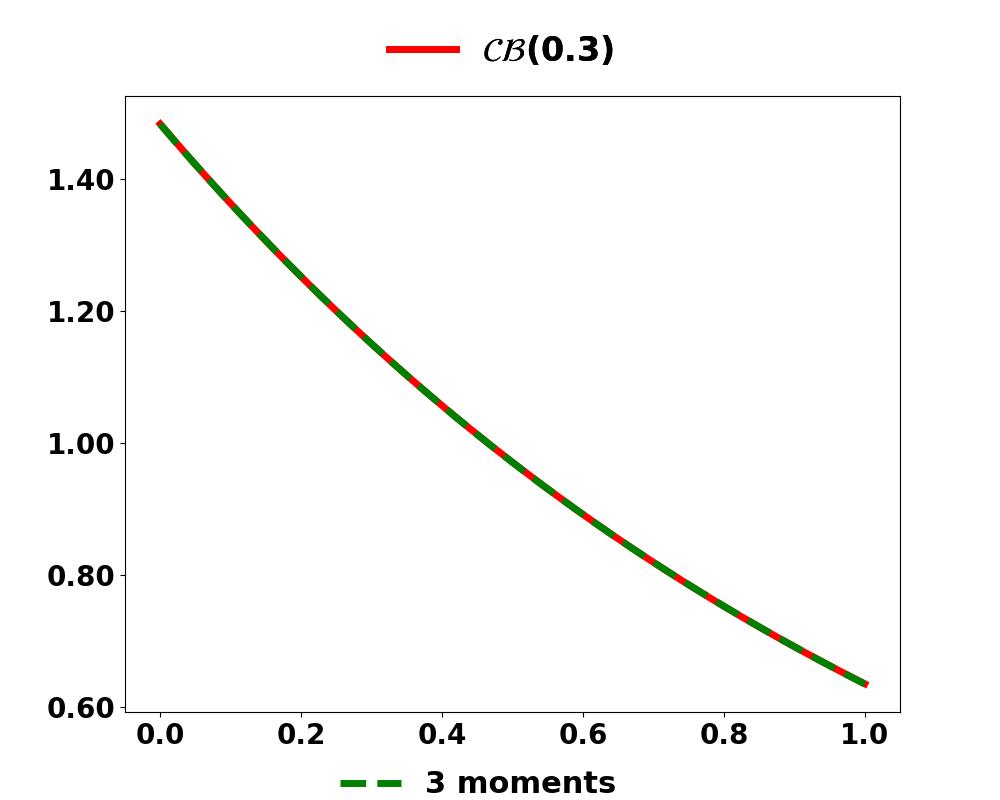}
         \caption{$\phi \sim U(0,1)$}
         \label{fig:five over x}
     \end{subfigure}
     \hfill
     \begin{subfigure}[b]{\WIDTH}
         \centering
         \includegraphics[width=\pWIDTH, height = \pHEIGHT]{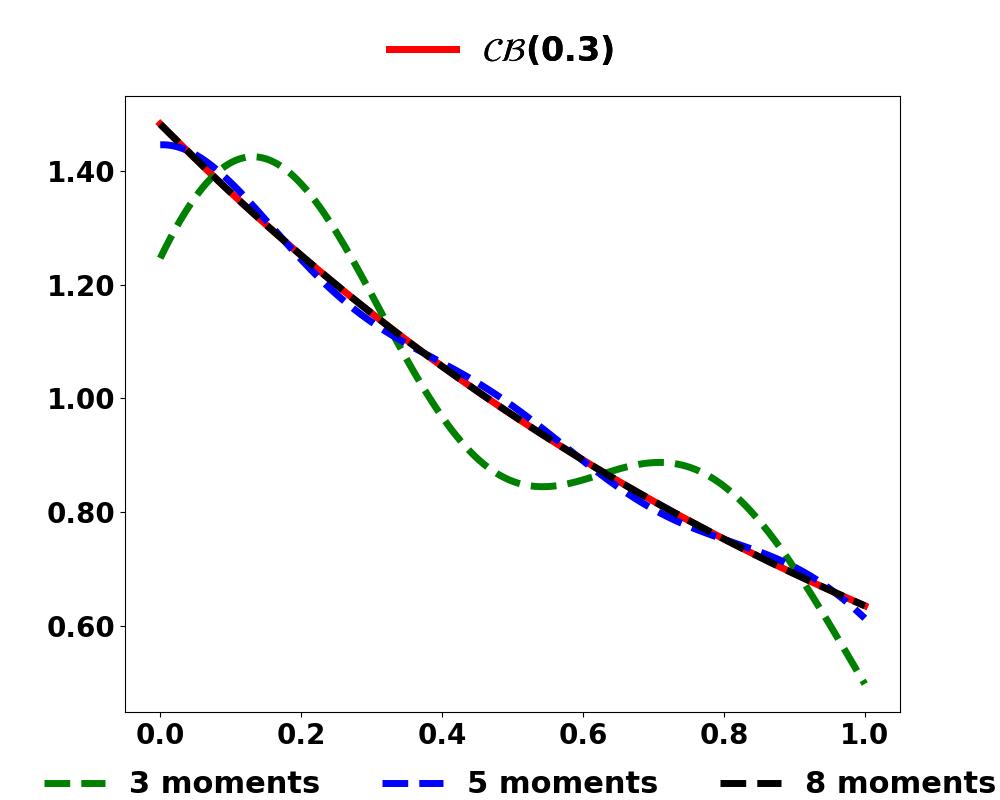}
         \caption{$\phi \sim Trunc$ $Normal(\E(f),\var(f), \left[0, 1\right])$}
         \label{fig:five over x}
     \end{subfigure}
      \hfill
      \end{minipage}
      \caption{K-series estimates of the truncated exponential pdf, the truncated gamma pdf, the truncated normal pdf and the continuous Bernoulli  with uniform reference (Method of Moments~\cite{Munkhammar_etal_2017}), left panels) and truncated normal (right panels) on  exact support. } 
        \label{fig:exact_support_K_series_exp}
\end{figure}

\begin{threeparttable}
\centering
\ra{0.5}
\begin{tabular}{@{}lcrcrcr@{}}\toprule
Target pdf $f$ & \phantom{abc}& $|M|$ &  \phantom{abc} & Uniform  & \phantom{abc}& Trunc Normal \\ 
 & \phantom{abc}&  &  \phantom{abc} & (Same support) & \phantom{abc}& (Same support)\\
\\ \midrule
Trunc Gamma$(\alpha=2,\beta=0.5,\left[0,5\right])$\\
  && 2 && 0.0172 {\color{green}{\ding{52}}} {\color{red}{!}} && 0.0188 {\color{green}{\ding{52}}} {\color{red}{!}}\\
  &&  3 && 0.0031 {\color{green}{\ding{52}}} {\color{red}{!}} && 0.0093 {\color{green}{\ding{52}}} {\color{red}{!}} \\
     && 5 && $<1e-4$ {\color{green}{\ding{52}}} {\color{red}{!}} && 0.0033 {\color{green}{\ding{52}}} {\color{red}{!}}\\
    &&  8 && $<1e-4$ {\color{green}{\ding{52}}} {\color{red}{!}} && 0.0002 {\color{green}{\ding{52}}} {\color{red}{!}}\\
Trunc Normal$(1.5, 5.76, \left[-6,6\right])$\\
  && 2 && {0.0617 {\color{red}{\ding{55}}}\hspace{0.29cm}} && 0.0011 {\color{green}{\ding{52}}} {\color{red}{!}} \\
  &&  4 && 0.0122 {\color{green}{\ding{52}}} {\color{red}{!}} && $<1e-4$ {\color{green}{\ding{52}}} {\color{red}{!}} \\
     && 7 && 0.0002 {\color{green}{\ding{52}}} {\color{red}{!}} && $<1e-4$ {\color{green}{\ding{52}}} {\color{red}{!}} \\
Continuous Bernoulli$(\pi=0.3)$\\
  &&  3 && $<1e-4$ {\color{green}{\ding{52}}} {\color{red}{!}} && 0.0124 {\color{green}{\ding{52}}} {\color{red}{!}}\\
     && 5 && $<1e-4$ {\color{green}{\ding{52}}} {\color{red}{!}} && 0.0012 {\color{green}{\ding{52}}} {\color{red}{!}}\\
    &&  8 && $<1e-4$ {\color{green}{\ding{52}}} {\color{red}{!}} && $<1e-4$ {\color{green}{\ding{52}}} {\color{red}{!}} \\
Trunc Exp$(\lambda=2/3,\left[0,4\right])$\\
  && 2 && 0.0082 {\color{green}{\ding{52}}} {\color{red}{!}} && 0.0212 {\color{green}{\ding{52}}} {\color{red}{!}} \\
     && 4 && 0.0001 {\color{green}{\ding{52}}} {\color{red}{!}} && 0.0025 {\color{green}{\ding{52}}} {\color{red}{!}}\\
    &&  6 && $<1e-4$ {\color{green}{\ding{52}}} {\color{red}{!}} && 0.0003 {\color{green}{\ding{52}}} {\color{red}{!}}\\
 \bottomrule
\end{tabular}
{\small
\begin{tablenotes}
     \item[{\color{green}{\ding{52}}}] \hspace{0.05cm} Null hypothesis is not rejected at  significance level 0.05.
     \item[{\color{green}{\ding{52}}} {\color{red}{!}}] Null hypothesis is not rejected at significance level 0.2.
     \item[{\color{red}{\ding{55}}}] \hspace{0.11cm} Null hypothesis is rejected at  significance level 0.05.
   \end{tablenotes}}
\caption{Kolmogorov-Smirnov distances and significance test results for reference distributions on the same support as the true pdf.}  
\label{fig:KS_exact_supports}
\end{threeparttable}

\begin{figure}[t!]
     \newcommand{\WIDTH}{0.49\textwidth}
     \newcommand{\HEIGHT}{\textwidth}
     \newcommand{\pWIDTH}{0.95\textwidth}
     \newcommand{\pHEIGHT}{0.625\textwidth}
     \centering
     \begin{minipage}{1.0\textwidth}
     \centering
     \vspace{0.8cm}
     \begin{subfigure}[b]{\WIDTH}
         \centering
         \includegraphics[width=\pWIDTH, height = \pHEIGHT]{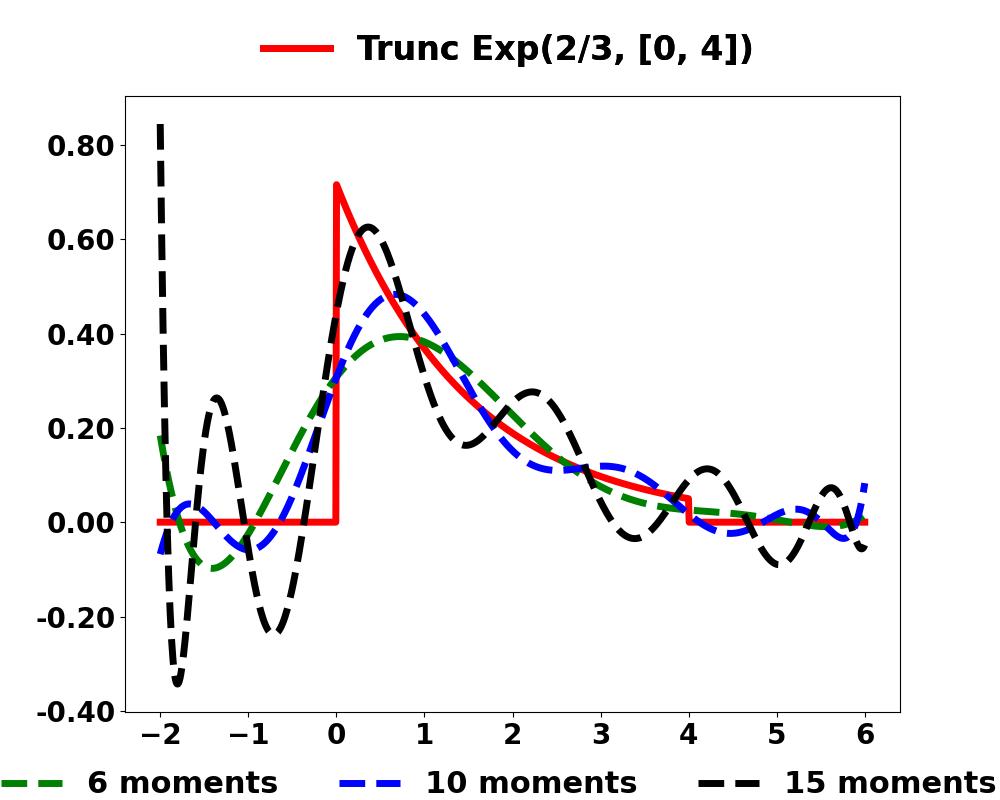}
         \caption{$\phi \sim U(-2,6)$}
         \label{fig:y equals x}
     \end{subfigure}
     \hfill
     \begin{subfigure}[b]{\WIDTH}
         \centering
         \includegraphics[width=\pWIDTH, height = \pHEIGHT]{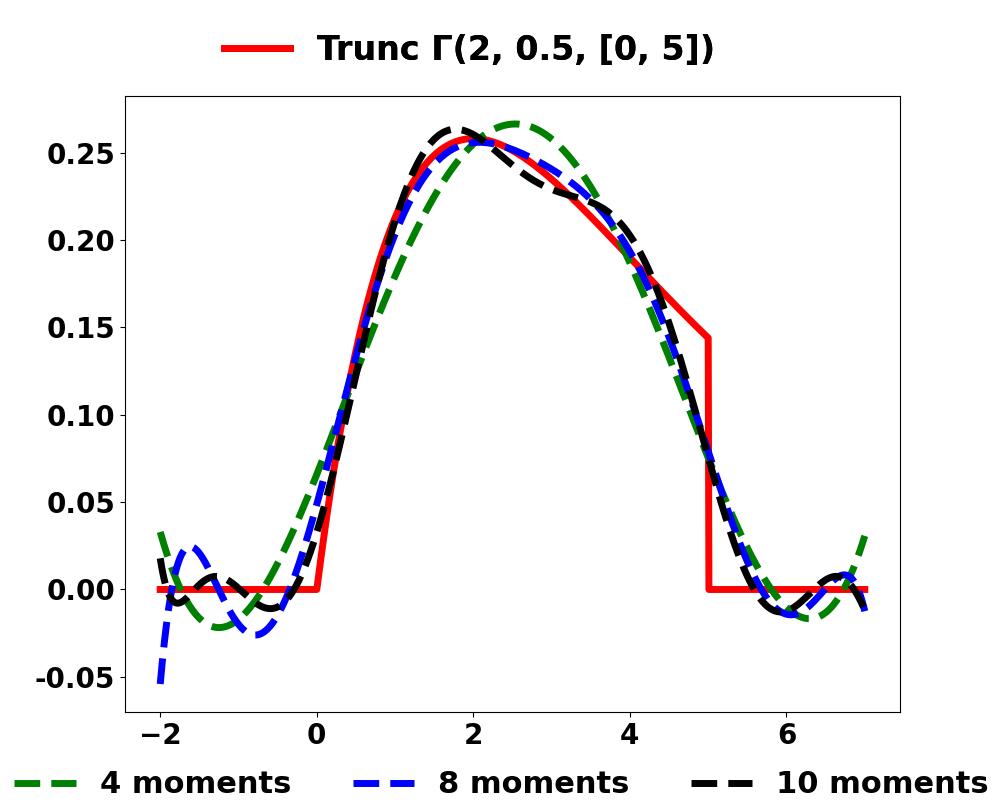}
         \caption{$\phi \sim U(-2,7)$}
         \label{fig:three sin x}
     \end{subfigure}
     \hfill
     \vspace{0.3cm}
     \begin{subfigure}[b]{\WIDTH}
         \centering
         \includegraphics[width=\pWIDTH, height = \pHEIGHT]{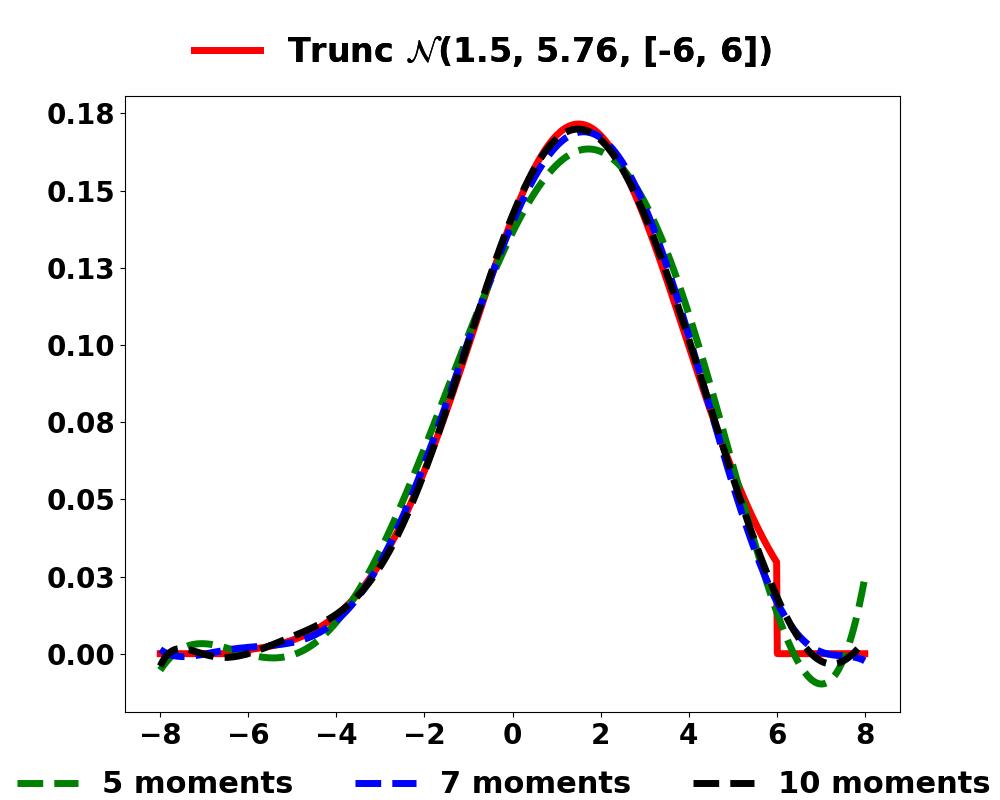}
         \caption{$\phi \sim U(-8,8)$}
         \label{fig:five over x}
     \end{subfigure}
     \hfill
     \begin{subfigure}[b]{\WIDTH}
         \centering
         \includegraphics[width=\pWIDTH, height = \pHEIGHT]{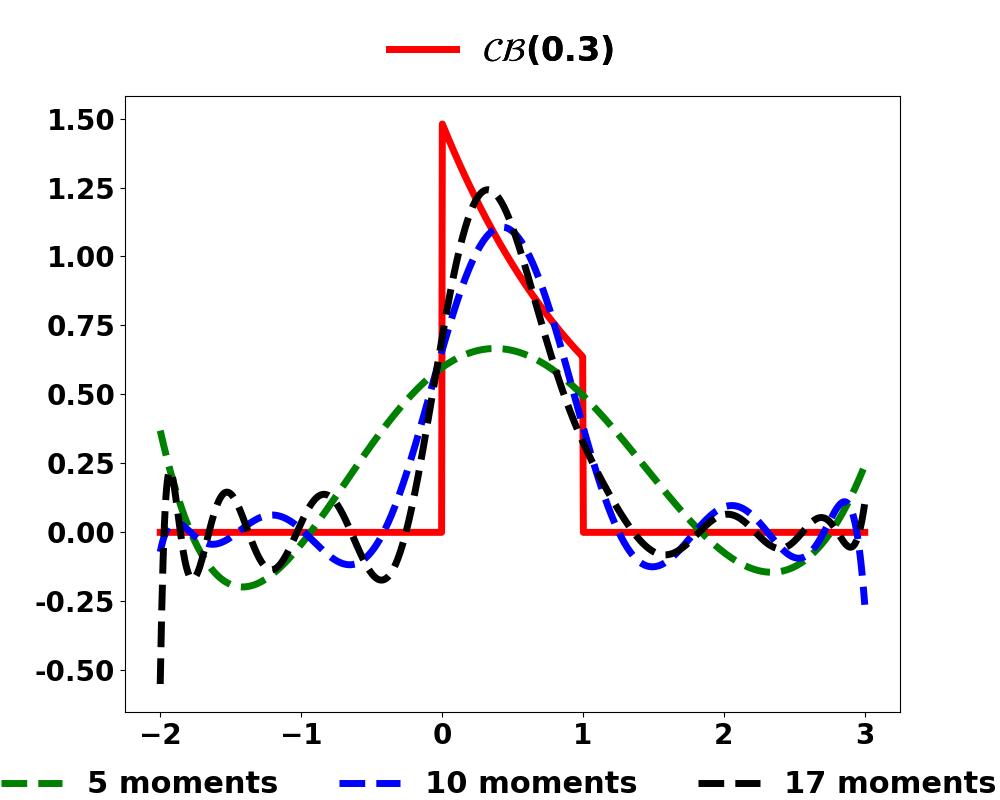}
         \caption{$\phi \sim U(-2,3)$}
         \label{fig:five over x}
     \end{subfigure}
      \hfill
      \end{minipage}
      \caption{Approximations of the truncated exponential pdf, the truncated gamma pdf, the truncated normal pdf and the continuous Bernoulli using K-series with uniform reference on the  extended support. } 
        \label{fig:ext_Uniform_K_series}
\end{figure}

\clearpage
\begin{center}
\begin{threeparttable}
\centering
\ra{0.5}
\begin{tabular}{@{}lcrcr@{}}\toprule
Target pdf $f$ & \phantom{abc}& $|M|$ &  \phantom{abc} & Uniform  \\ 
 & \phantom{abc}&  &  \phantom{abc} & (Extended support) \\
\\ \midrule
Trunc Gamma$(\alpha=2,\beta=0.5,\left[0,5\right])$\\
  && 4 && 0.0213 {\color{green}{\ding{52}}} {\color{red}{!}} \\
  &&  8 && 0.0186 {\color{green}{\ding{52}}} {\color{red}{!}}  \\
     && 10 && 0.0152 {\color{green}{\ding{52}}} {\color{red}{!}} \\
Trunc Normal$(1.5, 5.76, \left[-6,6\right])$\\
  && 5 && 0.0099 {\color{green}{\ding{52}}} {\color{red}{!}} \\
  &&  7 && 0.0061 {\color{green}{\ding{52}}} {\color{red}{!}}\\
     && 10 && 0.0048 {\color{green}{\ding{52}}} {\color{red}{!}}  \\
Continuous Bernoulli$(\pi=0.3)$\\
  &&  5 && {0.2285 {\color{red}{\ding{55}}}\hspace{0.28cm}} \\
     && 10 && {0.0939 {\color{red}{\ding{55}}}\hspace{0.28cm}} \\
    &&  17 && {0.0579 {\color{green}{\ding{52}}}\hspace{0.18cm}}  \\
Trunc Exp$(\lambda=2/3,\left[0,4\right])$\\
  && 6 && {0.1099 {\color{red}{\ding{55}}}\hspace{0.28cm}}  \\
     && 10 && {0.0713 {\color{red}{\ding{55}}}\hspace{0.28cm}} \\
    &&  15 && {0.0546 {\color{green}{\ding{52}}}\hspace{0.18cm}}  \\
 \bottomrule
\end{tabular}
{\small
\begin{tablenotes}
     \item[{\color{green}{\ding{52}}}] \hspace{0.05cm} Null hypothesis is not rejected at  significance level 0.05.
     \item[{\color{green}{\ding{52}}} {\color{red}{!}}] Null hypothesis is not rejected at significance level 0.2.
     \item[{\color{red}{\ding{55}}}] \hspace{0.11cm} Null hypothesis is rejected at  significance level 0.05.
   \end{tablenotes}}
\caption{Kolmogorov-Smirnov distances and significance test results for the uniform reference distribution on extended support.}  
\label{fig:KS_uniform_extended_supports}
\end{threeparttable}
\end{center}

\begin{figure}[!t]
     \newcommand{\WIDTH}{0.49\textwidth}
     \newcommand{\HEIGHT}{\textwidth}
     \newcommand{\pWIDTH}{0.95\textwidth}
     \newcommand{\pHEIGHT}{0.525\textwidth}
     \centering
     \vspace{0.8cm}
     \begin{minipage}{1.0\textwidth}
     \centering
     \begin{subfigure}[b]{\WIDTH}
         \centering
         \includegraphics[width=\pWIDTH, height = \pHEIGHT]{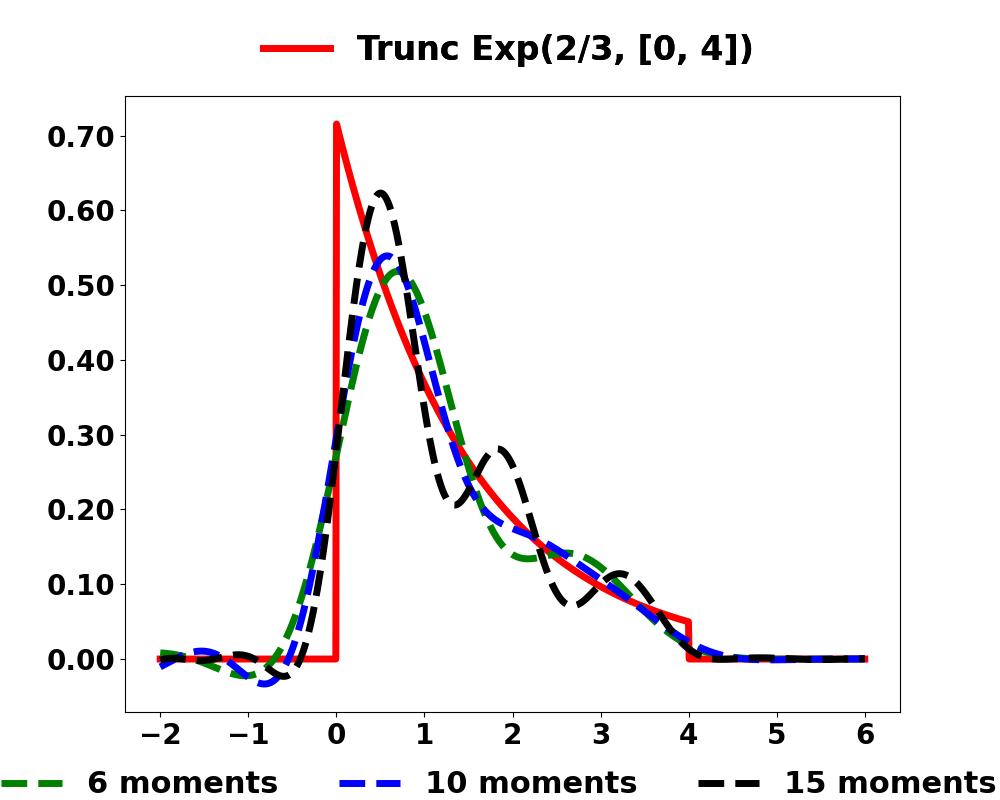}
         \caption{$\phi \sim Trunc$ $Normal(\E(f),\var(f), \left[-2, 6\right])$}
         \label{fig:y equals x}
     \end{subfigure}
     \hfill
     \begin{subfigure}[b]{\WIDTH}
         \centering
         \includegraphics[width=\pWIDTH, height = \pHEIGHT]{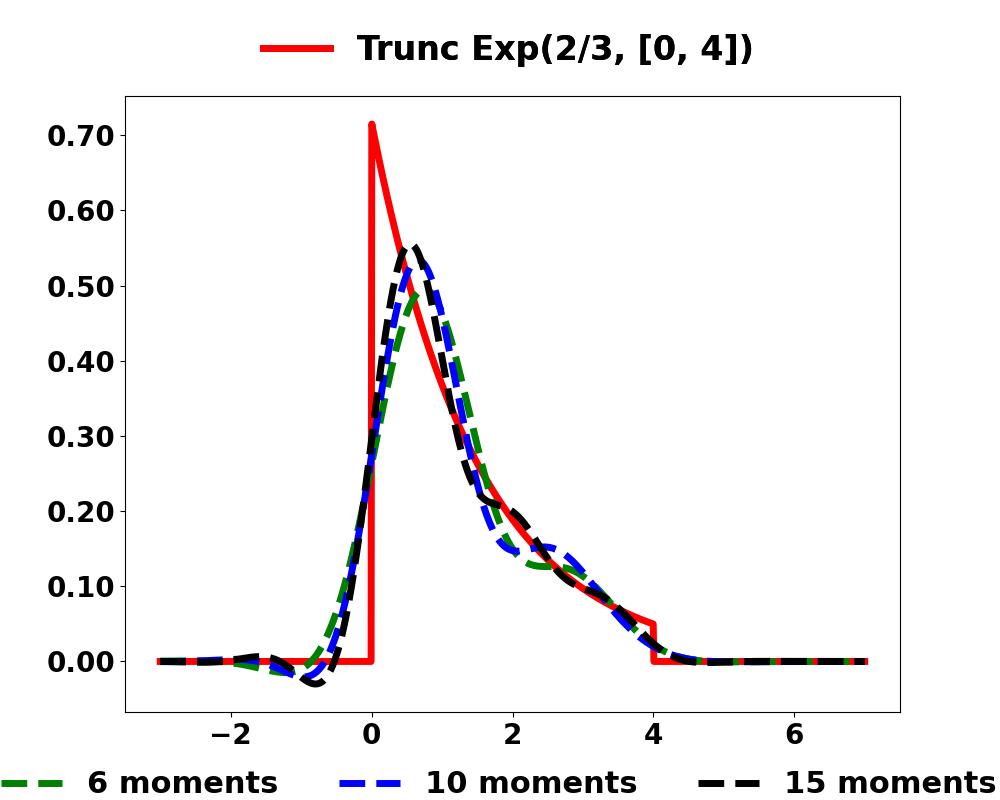}
         \caption{$\phi \sim Normal(\E(f),\var(f))$}
         \label{fig:three sin x}
     \end{subfigure}
     \hfill
     \vspace{0.3cm}
     \begin{subfigure}[b]{\WIDTH}
         \centering
         \includegraphics[width=\pWIDTH, height = \pHEIGHT]{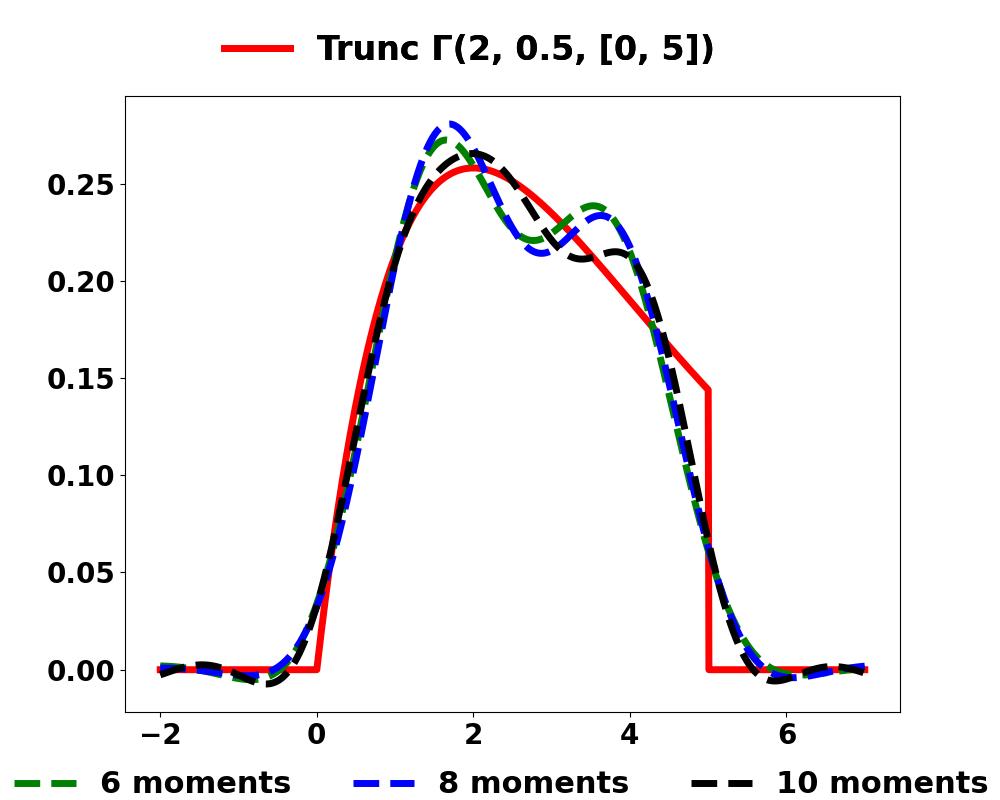}
         \caption{$\phi \sim Trunc$ $Normal(\E(f),\var(f), \left[-2, 7\right])$}
         \label{fig:five over x}
     \end{subfigure}
     \hfill
     \begin{subfigure}[b]{\WIDTH}
         \centering
         \includegraphics[width=\pWIDTH, height = \pHEIGHT]{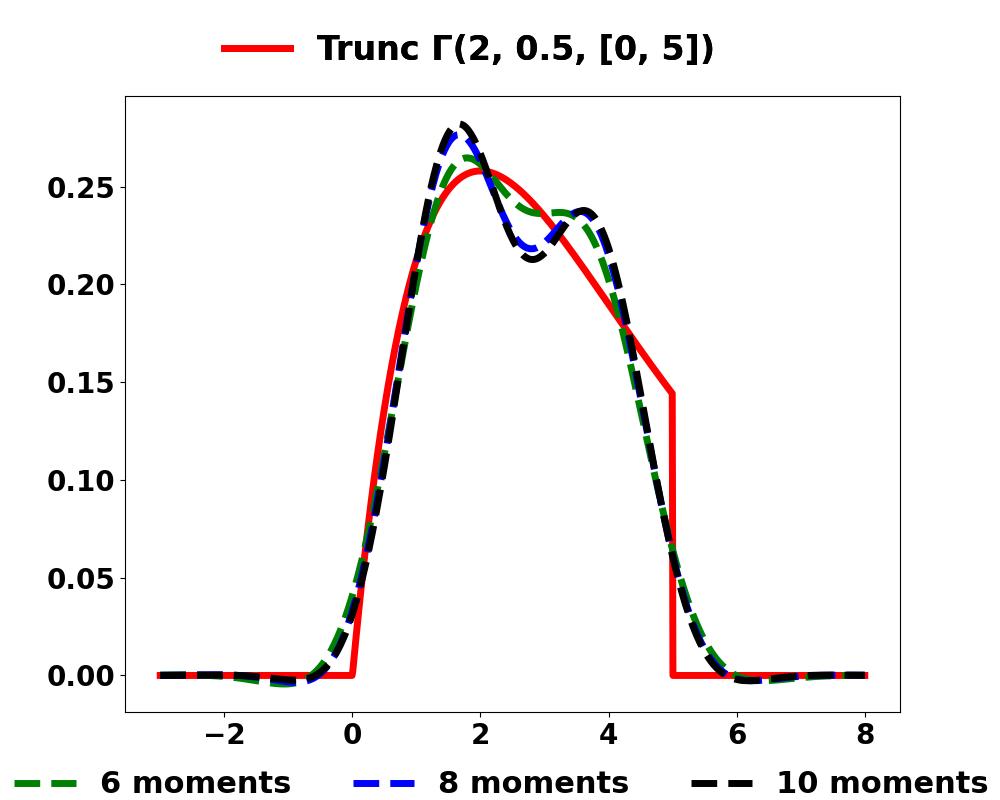}
         \caption{$\phi \sim Normal(\E(f),\var(f))$}
         \label{fig:five over x}
     \end{subfigure}
     \hfill
     \vspace{0.3cm}
     \begin{subfigure}[b]{\WIDTH}
         \centering
         \includegraphics[width=\pWIDTH, height = \pHEIGHT]{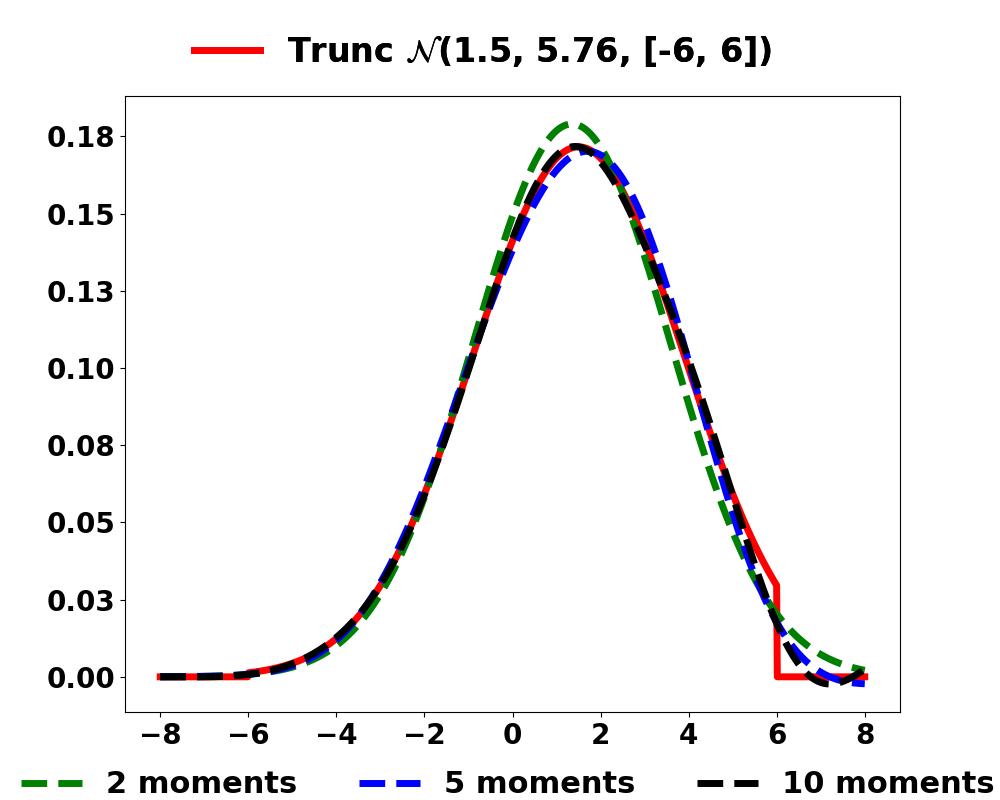}
         \caption{$\phi \sim Trunc$ $Normal(\E(f),\var(f), \left[-8, 8\right])$}
         \label{fig:five over x}
     \end{subfigure}
     \hfill
     \begin{subfigure}[b]{\WIDTH}
         \centering
         \includegraphics[width=\pWIDTH, height = \pHEIGHT]{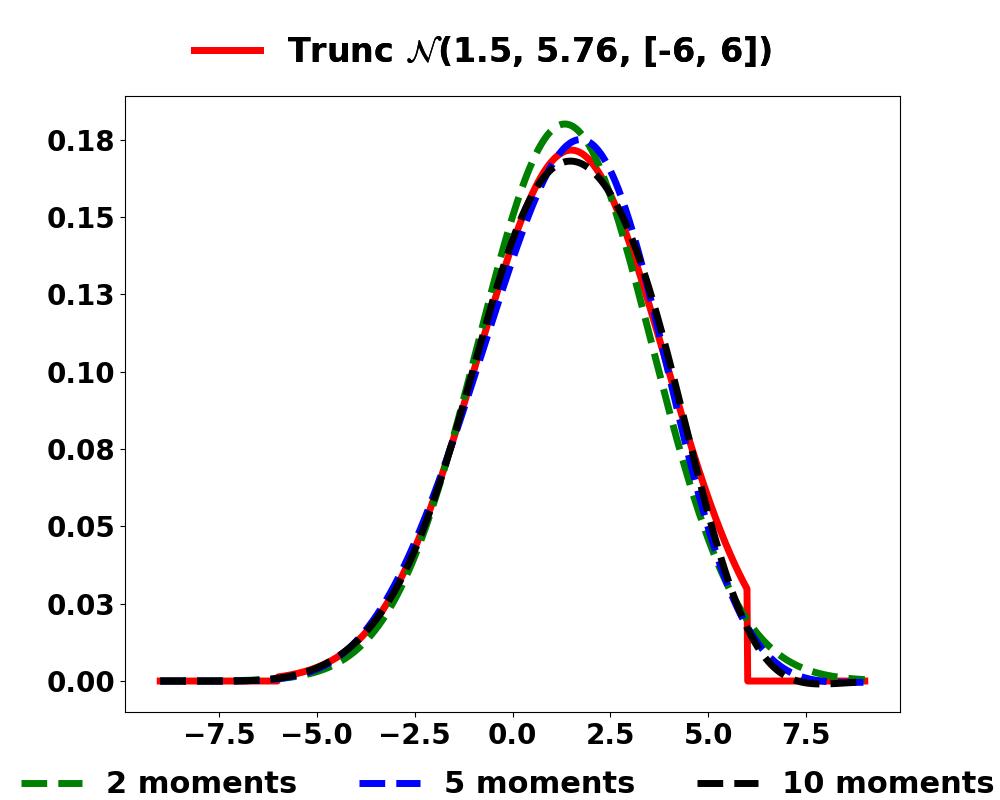}
         \caption{$\phi \sim Normal(\E(f),\var(f))$}
         \label{fig:five over x}
     \end{subfigure}
     \hfill
     \vspace{0.3cm}
     \begin{subfigure}[b]{\WIDTH}
         \centering
         \includegraphics[width=\pWIDTH, height = \pHEIGHT]{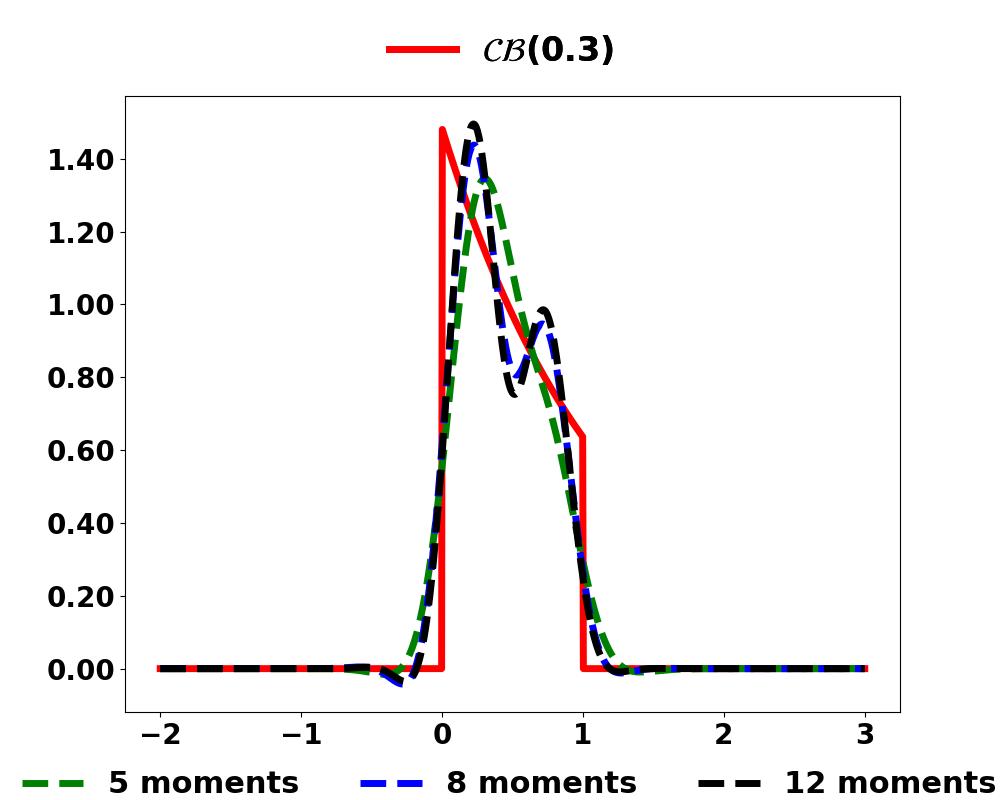}
         \caption{$\phi \sim Trunc$ $Normal(\E(f),\var(f), \left[-2, 3\right])$}
         \label{fig:five over x}
     \end{subfigure}
     \hfill
     \begin{subfigure}[b]{\WIDTH}
         \centering
         \includegraphics[width=\pWIDTH, height = \pHEIGHT]{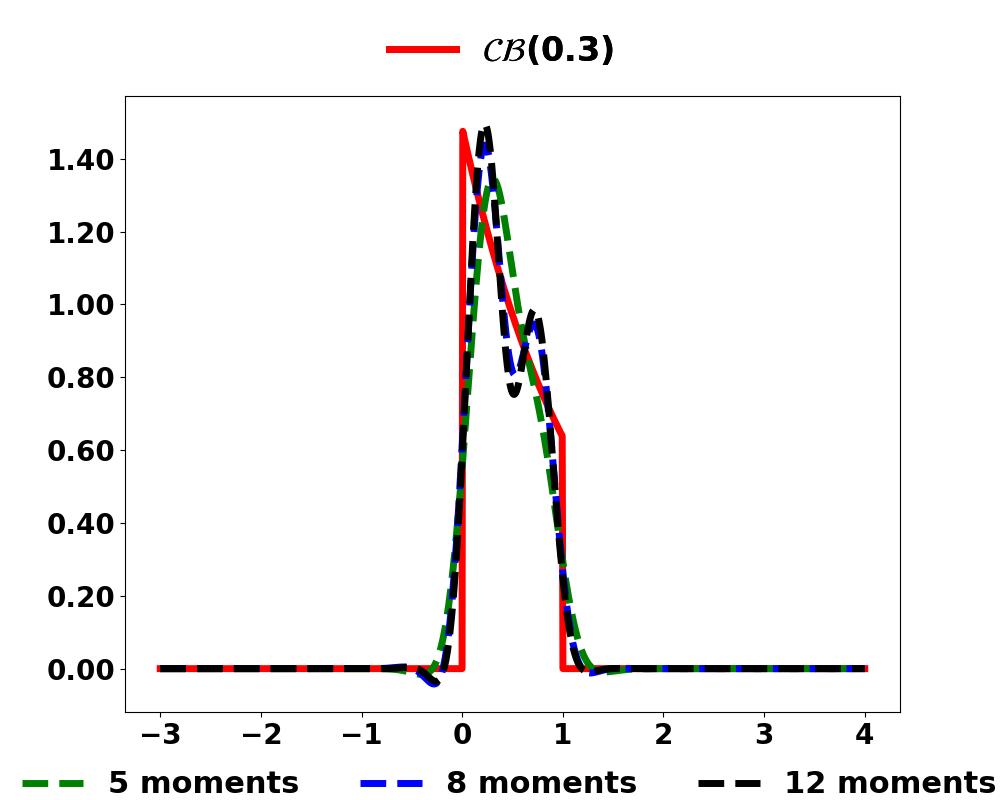}
         \caption{$\phi \sim Normal(\E(f),\var(f))$}
         \label{fig:five over x}
     \end{subfigure}
      \hfill
      \end{minipage}
      \caption{Approximations of the truncated exponential pdf, the truncated gamma pdf, the truncated normal pdf and the continuous Bernoulli using K-series with truncated normal reference on the  extended support (left) and normal reference on the whole real line (Gram-Charlier, right).} 
        \label{fig:normal_ref_K_series}
\end{figure}

\begin{center}
\begin{threeparttable}
\centering
\ra{0.5}
\begin{tabular}{@{}lcrcrcr@{}}\toprule
Target pdf $f$ & \phantom{abc}& $|M|$ &  \phantom{abc} & Trunc Normal  & \phantom{abc}& Normal \\ 
 & \phantom{abc}&  &  \phantom{abc} & (Extended support) & \phantom{abc}& (real line)\\
\\ \midrule
Trunc Gamma$(\alpha=2,\beta=0.5,\left[0,5\right])$\\
  && 6 && 0.0172 {\color{green}{\ding{52}}} {\color{red}{!}} && 0.0202 {\color{green}{\ding{52}}} {\color{red}{!}}\\
  &&  8 && 0.0158 {\color{green}{\ding{52}}} {\color{red}{!}} && 0.0169 {\color{green}{\ding{52}}} {\color{red}{!}} \\
     && 10 && 0.0132 {\color{green}{\ding{52}}} {\color{red}{!}} && 0.0033 {\color{green}{\ding{52}}} {\color{red}{!}}\\
Trunc Normal$(1.5, 5.76, \left[-6,6\right])$\\
  && 2 && 0.0171 {\color{green}{\ding{52}}} {\color{red}{!}} && 0.0182 {\color{green}{\ding{52}}} {\color{red}{!}} \\
  &&  5 && 0.0071 {\color{green}{\ding{52}}} {\color{red}{!}} && 0.0095 {\color{green}{\ding{52}}} {\color{red}{!}} \\
     && 10 && 0.0044 {\color{green}{\ding{52}}} {\color{red}{!}} && 0.0066 {\color{green}{\ding{52}}} {\color{red}{!}} \\
Continuous Bernoulli$(\pi=0.3)$\\
  &&  5 && {0.0516 {\color{green}{\ding{52}}}\hspace{0.20cm}}  && {0.0527 {\color{green}{\ding{52}}}\hspace{0.20cm}} \\
     && 8 && 0.0374 {\color{green}{\ding{52}}} {\color{red}{!}} && 0.0387 {\color{green}{\ding{52}}} {\color{red}{!}}\\
    &&  12 && 0.0340 {\color{green}{\ding{52}}} {\color{red}{!}} && 0.0352 {\color{green}{\ding{52}}} {\color{red}{!}} \\
Trunc Exp$(\lambda=2/3,\left[0,4\right])$\\
  && 6 && {0.0667 {\color{red}{\ding{55}}}\hspace{0.30cm}} && {0.0757 {\color{red}{\ding{55}}}\hspace{0.32cm}} \\
     && 10 && {0.0558 {\color{green}{\ding{52}}}\hspace{0.20cm}}  && 
     {0.0617 {\color{red}{\ding{55}}}\hspace{0.32cm}} \\
    &&  15 && 0.0391 {\color{green}{\ding{52}}} {\color{red}{!}} && {0.0524 {\color{green}{\ding{52}}}\hspace{0.23cm}} \\
 \bottomrule
\end{tabular}
{\small
\begin{tablenotes}
     \item[{\color{green}{\ding{52}}}] \hspace{0.05cm} Null hypothesis is not rejected at  significance level 0.05.
     \item[{\color{green}{\ding{52}}} {\color{red}{!}}] Null hypothesis is not rejected at significance level 0.2.
     \item[{\color{red}{\ding{55}}}] \hspace{0.11cm} Null hypothesis is rejected at  significance level 0.05.
   \end{tablenotes}}
\caption{Kolmogorov-Smirnov distances and significance test results for truncated normal on extended support and normal reference distributions.}  
\label{fig:KS_normal_ref}
\end{threeparttable}
\end{center}

\clearpage

To show that any continuous reference pdf that is positive on its support which contains the support of the unknown target can be used in K-series, we present an example where the reference is exponential with scale parameter 0.2 in Fig. \ref{fig:extra_example1}, and an example where the reference is Gamma with shape parameter 2 in Fig. \ref{fig:extra_example2}.
The latter serves to illustrate that the requirement for the reference to be positive everywhere on its support is sufficient, but not necessary, in general. The limit at point $x=0$ of the ratio $f^{2}(x) / \phi(x)$, in this case, is zero and the integral exists. But in general, this condition cannot be checked when the true target pdf is unknown.

\begin{figure}
     \centering
     \begin{minipage}{.45\textwidth}
        \centering
\includegraphics[scale=0.3]{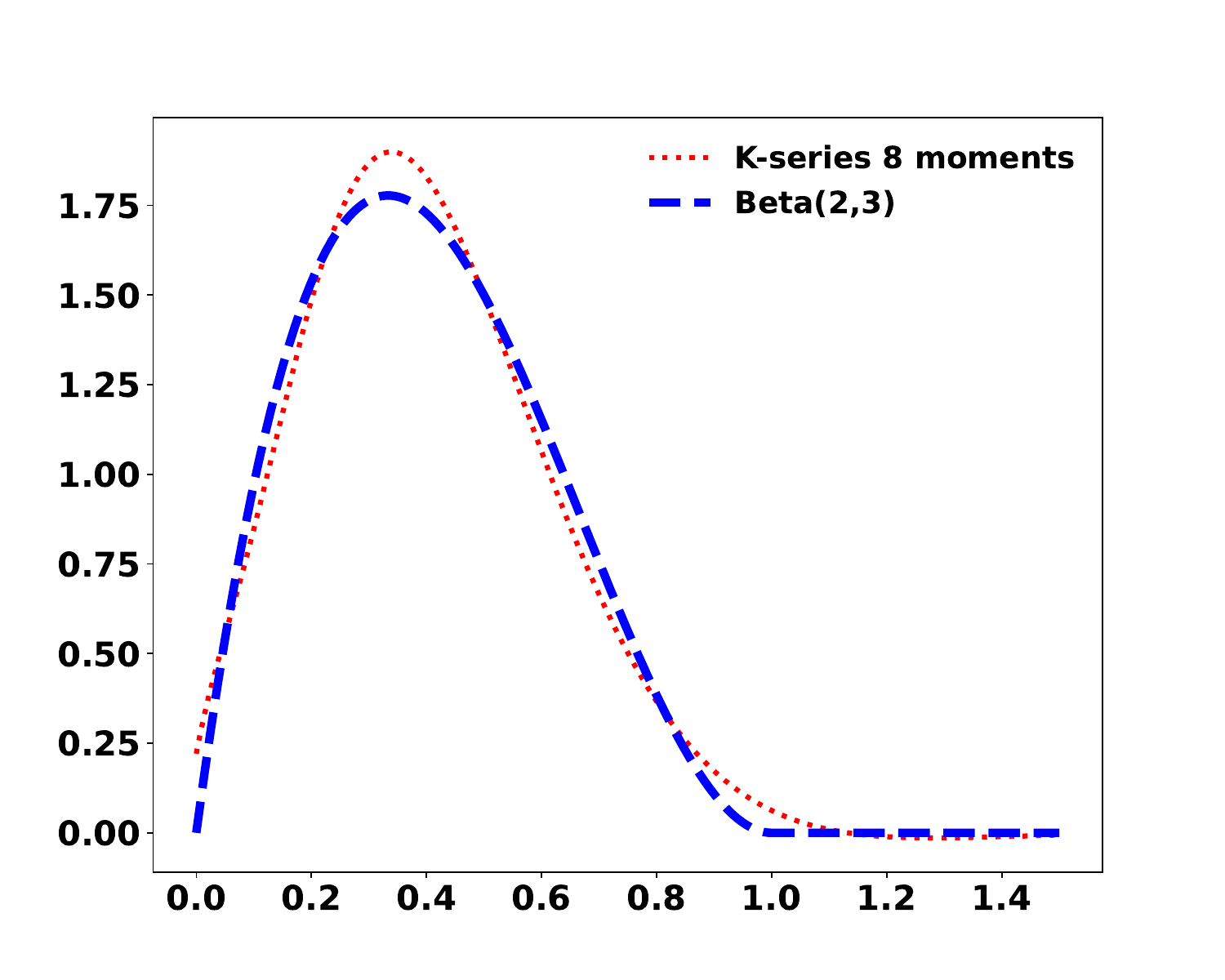}
  \caption{K-series estimates using first 8  moments of a Beta distribution with parameters (2, 3) and exponential reference with scale parameter 0.2.} 
  \label{fig:extra_example1}
\label{fig:(a)}
     \end{minipage}
     \hfill 
     \begin{minipage}{.45\textwidth}
         \centering
         \includegraphics[scale=0.3]{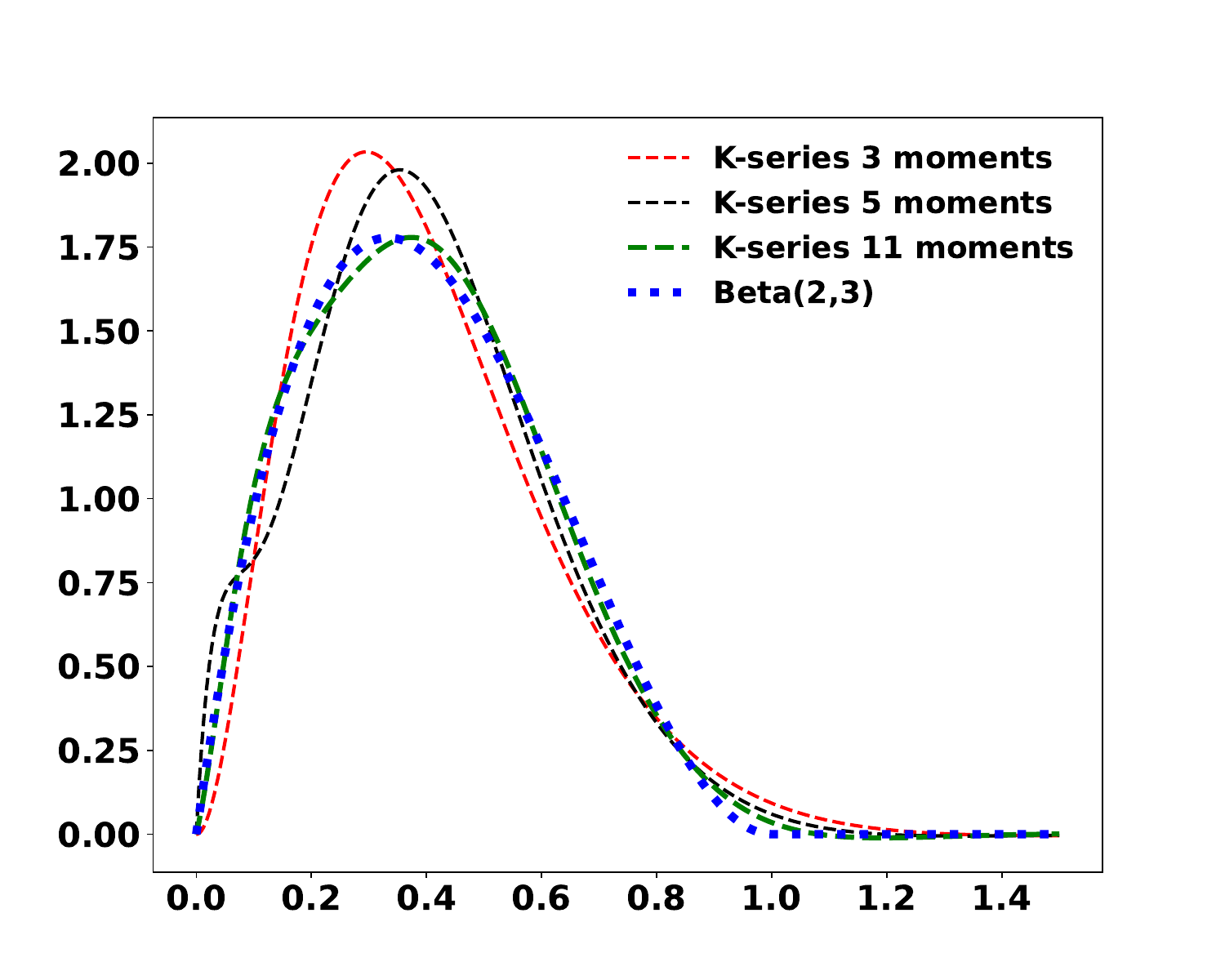}
  \caption{K-series estimates using the first  3, 5, and 11 moments, respectively, of a Beta distribution with parameters (2, 3) and a Gamma reference with shape and scale 2 and 0.14, respectively.} 
  \label{fig:extra_example2}
  \label{fig:(b)}
     \end{minipage}
     
\end{figure}

\end{document}